\documentclass[12pt,letterpaper]{article}
\usepackage{epsfig,rotating,setspace,latexsym,amsmath,epsf,amssymb,bm,theorem}
\usepackage{cite,appendix}

\usepackage{amsfonts}

\title{The Secrecy Capacity Region of the Gaussian MIMO Multi-Receiver Wiretap
Channel\thanks{This work was supported by NSF Grants CCF 04-47613, CCF 05-14846, CNS
07-16311 and CCF 07-29127.}}

\author{Ersen Ekrem \qquad Sennur Ulukus \\
\normalsize Department of Electrical and Computer Engineering\\
\normalsize University of Maryland, College Park, MD 20742 \\
\normalsize {\it ersen@umd.edu} \qquad {\it ulukus@umd.edu}}

\newcommand{\bone}{\bm 1}
\newcommand{\bblambda}{\bm \Lambda}
\newcommand{\brho}{\bm \rho}
\newcommand{\bbdelta}{\bm \Delta}
\newcommand{\bbsigma}{\bm \Sigma}

\newcommand{\bbi}{{\mathbf{I}}}
\newcommand{\bzero}{{\mathbf{0}}}
\newcommand{\bbv}{{\mathbf{V}}}
\newcommand{\bv}{{\mathbf{v}}}
\newcommand{\bw}{{\mathbf{w}}}
\newcommand{\bbh}{{\mathbf{H}}}
\newcommand{\bbw}{{\mathbf{W}}}
\newcommand{\bbm}{{\mathbf{M}}}

\newcommand{\bbk}{{\mathbf{K}}}
\newcommand{\bbz}{{\mathbf{Z}}}
\newcommand{\bbn}{{\mathbf{N}}}
\newcommand{\bn}{{\mathbf{n}}}
\newcommand{\bba}{{\mathbf{A}}}
\newcommand{\bbd}{{\mathbf{D}}}
\newcommand{\bbe}{{\mathbf{E}}}
\newcommand{\bbb}{{\mathbf{B}}}
\newcommand{\bbc}{{\mathbf{C}}}

\newcommand{\bbs}{{\mathbf{S}}}
\newcommand{\bu}{{\mathbf{u}}}
\newcommand{\bbj}{{\mathbf{J}}}
\newcommand{\bbu}{{\mathbf{U}}}
\newcommand{\bx}{{\mathbf{x}}}
\newcommand{\bbx}{{\mathbf{X}}}

\newcommand{\bby}{{\mathbf{Y}}}

\newtheorem{Theo}{Theorem}
\newtheorem{Prop}{Proposition}
\newtheorem{Lem}{Lemma}
\newtheorem{Cor}{Corollary}
\newtheorem{Def}{Definition}
\newenvironment{proof}[1]{\medskip\par\noindent{\bf Proof:\,}\,#1}{{\mbox{\,$\blacksquare$}\par}}

\begin{document}


\maketitle

\begin{abstract}
In this paper, we consider the Gaussian multiple-input
multiple-output (MIMO) multi-receiver wiretap channel in which a
transmitter wants to have confidential communication with an
arbitrary number of users in the presence of an external
eavesdropper. We derive the secrecy capacity region of this
channel for the most general case. We first show that even for the
single-input single-output (SISO) case, existing converse
techniques for the Gaussian scalar broadcast channel cannot be
extended to this secrecy context, to emphasize the need for a new
proof technique. Our new proof technique makes use of the
relationships between the minimum-mean-square-error and the mutual
information, and equivalently, the relationships between the
Fisher information and the differential entropy. Using the
intuition gained from the converse proof of the SISO channel, we
first prove the secrecy capacity region of the degraded MIMO
channel, in which all receivers have the same number of antennas,
and the noise covariance matrices can be arranged according to a
positive semi-definite order. We then generalize this result to
the aligned case, in which all receivers have the same number of
antennas, however there is no order among the noise covariance
matrices. We accomplish this task by using the channel enhancement
technique. Finally, we find the secrecy capacity region of the
general MIMO channel by using some limiting arguments on the
secrecy capacity region of the aligned MIMO channel. We show that
the capacity achieving coding scheme is a variant of dirty-paper
coding with Gaussian signals.

\end{abstract}

\newpage
\section{Introduction}

Information theoretic secrecy was initiated by Wyner in his
seminal work~\cite{Wyner}. Wyner considered a degraded wiretap
channel, where the eavesdropper gets a degraded version of the
legitimate receiver's observation. For this degraded model, he
found the capacity-equivocation rate region where the equivocation
rate refers to the portion of the message rate that can be
delivered to the legitimate receiver, while the eavesdropper is
kept totally ignorant of this part. Later, Csiszar and Korner
considered the general wiretap channel, where there is no presumed
degradation order between the legitimate user and the
eavesdropper~\cite{Korner}. They found the capacity-equivocation
rate region of this general, {\it not necessarily degraded},
wiretap channel.

In recent years, information theoretic secrecy has gathered a
renewed interest, where most of the attention has been devoted to
the multiuser extensions of the wiretap channel, see for example
\cite{Aylin_2,Aylin_Cooperative,Allerton_08,Ruoheng,Ruoheng2,Broadcasting_Wornell,Khandani,Ekrem_Ulukus_Asilomar08,Ekrem_Ulukus_BC_Secrecy,
Oohama, Hesham,Melda_1,He_journal,CISS_08,ISIT_08,
Bloch_relay,Yingbin_1,Ruoheng_3,Ruoheng_4}. One natural extension
of the wiretap channel to the multiuser setting is the problem of
{\it secure broadcasting}. In this case, there is one transmitter
which wants to communicate with several legitimate users in the
presence of an external eavesdropper. Hereafter, we call this
channel model the {\it multi-receiver wiretap channel}. Finding
the secrecy capacity region for the most general form of this
channel model seems to be quite challenging, especially if one
remembers that, even without the eavesdropper, we do not know the
the capacity region for the underlying channel, which is a general
broadcast channel with an arbitrary number of users. However, we
know the capacity region for some special classes of broadcast
channels, which suggests that we might be able to find the secrecy
capacity region for some special classes of multi-receiver wiretap
channels. This suggestion has been taken into consideration in
\cite{Broadcasting_Wornell,Khandani,Ekrem_Ulukus_Asilomar08,Ekrem_Ulukus_BC_Secrecy}.
In particular,
in~\cite{Khandani,Ekrem_Ulukus_BC_Secrecy,Ekrem_Ulukus_Asilomar08},
the degraded multi-receiver wiretap channel is considered, where
there is a certain degradation order among the legitimate users
and the eavesdropper. The corresponding secrecy capacity region is
derived for the two-user case in~\cite{Khandani}, and for an
arbitrary number of users
in~\cite{Ekrem_Ulukus_Asilomar08,Ekrem_Ulukus_BC_Secrecy}. The
importance of this class lies in the fact that the Gaussian scalar
multi-receiver wiretap channel belongs to this class.

In this work, we start with the Gaussian scalar multi-receiver
wiretap channel, and find its secrecy capacity region. Although,
in the later parts of the paper, we provide the secrecy capacity
region of the Gaussian multiple-input multiple-output (MIMO)
multi-receiver wiretap channel which subsumes the scalar case,
there are two reasons for the presentation of the scalar case
separately. The first one is to show that, existing converse
techniques for the Gaussian scalar broadcast channel, i.e., the
converse proofs of Bergmans~\cite{Bergmans} and El
Gamal~\cite{El_Gamal_Converse}, cannot be extended in a
straightforward manner to provide a converse proof for the
Gaussian scalar multi-receiver wiretap channel. We explicitly show
that the main ingredient of these two converses in
\cite{Bergmans,El_Gamal_Converse}, which is the entropy-power
inequality~\cite{Stam,Blachman}, is not sufficient to conclude a
converse for the secrecy capacity region. The second reason for
the separate presentation is to present the main ingredients of
the technique that we will use to provide a converse proof for the
general MIMO channel in an isolated manner in a simpler context.
We provide two converse proofs for the Gaussian scalar
multi-receiver wiretap channel. The first one uses the connection
between the minimum-mean-square-error (MMSE) and the mutual
information along with the properties of the
MMSE~\cite{Guo_MMSE,Guo_ISIT_08}. In additive Gaussian channels,
the Fisher information, another important quantity in estimation
theory, and the MMSE have a complementary relationship in the
sense that one of them determines the other one, and vice
versa~\cite{Olivier_Rioul}. Thus, the converse proof relying on
the MMSE has a counterpart which replaces the Fisher information
with the MMSE in the corresponding converse proof. Hence, the
second converse uses the connection between the Fisher information
and the differential entropy via the de Bruin
identity~\cite{Blachman,Stam} along with the properties of the
Fisher information. This reveals that either the Fisher
information matrix or the MMSE matrix should play an important
role in the converse proof of the MIMO case.

Keeping this motivation in mind, we consider the Gaussian MIMO
multi-receiver wiretap channel next. Instead of directly tackling
the most general case in which each receiver has an arbitrary
number of antennas and an arbitrary noise covariance matrix, we
first consider two sub-classes of MIMO channels. In the first
sub-class, all receivers have the same number of antennas and the
noise covariance matrices exhibit a positive semi-definite order,
which implies the degradedness of these channels. Hereafter, we
call this channel model the {\it degraded Gaussian MIMO
multi-receiver wiretap channel}. In the second sub-class, although
all receivers still have the same number of antennas as in the
degraded case, the noise covariance matrices do not have to
satisfy any positive semi-definite order. Hereafter, we call this
channel model the {\it aligned Gaussian MIMO multi-receiver
wiretap channel}. Our approach will be to first find the secrecy
capacity region of the degraded case, then to generalize this
result to the aligned case by using the {\it channel enhancement}
technique~\cite{Shamai_MIMO}. Once we obtain the secrecy capacity
region of the aligned case, we use this result to find the secrecy
capacity region of the most general case by some limiting
arguments as in~\cite{Shamai_MIMO,Tie_Liu_MIMO_WT}. Thus, the main
contribution and the novelty of our work is the way we prove the
secrecy capacity region of the degraded Gaussian MIMO
multi-receiver wiretap channel, since the remaining steps from
then on are mainly adaptations of the existing proof
techniques~\cite{Shamai_MIMO,Tie_Liu_MIMO_WT} to an eavesdropper
and/or multiuser setting.

At this point, to clarify our contributions, it might be useful to
note the similarity of the proof steps that we follow with those
in~\cite{Shamai_MIMO}, where the capacity region of the Gaussian
MIMO broadcast channel was established. In~\cite{Shamai_MIMO}
also, the authors considered the degraded, the aligned and the
general cases successively. Although, both~\cite{Shamai_MIMO} and
this paper has these same proof steps, there are differences
between how and why these steps are taken. In \cite{Shamai_MIMO},
the main difficulty in obtaining the capacity region of the
Gaussian MIMO broadcast channel was to extend Bergmans' converse
for the scalar case to the degraded vector channel. This
difficulty was overcome in~\cite{Shamai_MIMO} by the invention of
the {\it channel enhancement} technique. However, as discussed
earlier, Bergmans' converse cannot be extended to our secrecy
context, even for the degraded scalar case. Thus, we need a new
technique which we construct by using the Fisher information
matrix and the generalized de Bruin
identity~\cite{Palomar_Gradient}. After we obtain the secrecy
capacity region of the degraded MIMO channel, we adapt the channel
enhancement technique to our setting to find the secrecy capacity
region of the aligned MIMO channel. The difference of the way
channel enhancement is used here as compared to the one
in~\cite{Shamai_MIMO} comes from the presence of an eavesdropper,
and its difference from the one in~\cite{Tie_Liu_MIMO_WT} is due
to the presence of many legitimate users. After we find the
secrecy capacity region of the aligned MIMO channel, we use the
limiting arguments that appeared
in~\cite{Shamai_MIMO,Tie_Liu_MIMO_WT} to prove the secrecy
capacity region of the general MIMO channel.

The single user version of the Gaussian MIMO multi-receiver
wiretap channel we study here, i.e., the Gaussian MIMO wiretap
channel, was solved by \cite{Wornell,Hassibi} for the general case
and by \cite{Ulukus} for the 2-2-1 case. Their common proof
technique was to derive a Sato-type outer bound on the secrecy
capacity, and then to tighten this outer bound by searching over
all possible correlation structures among the noise vectors of the
legitimate user and the eavesdropper. Later,
\cite{Tie_Liu_MIMO_WT} gave an alternative, simpler proof by using
the channel enhancement technique.

\section{Multi-Receiver Wiretap Channels}
In this section, we first revisit the multi-receiver wiretap
channel. The general multi-receiver wiretap channel consists of
one transmitter with an input alphabet $\mathcal{X}$, $K$
legitimate receivers with output alphabets $\mathcal{Y}_k$,
$k=1,\ldots,K,$ and an eavesdropper with output alphabet
$\mathcal{Z}$. The transmitter sends a confidential message to
each user, say $w_k\in\mathcal{W}_k$ to the $k$th user, and all
messages are to be kept secret from the eavesdropper. The channel
is memoryless with a transition probability
$p(y_1,y_2,\ldots,y_K,z|x)$.

A $(2^{nR_1},\ldots,2^{nR_K},n)$ code for this channel consists of
$K$ message sets, $\mathcal{W}_k=\{1,\ldots,2^{nR_k}\}$,
$k=1,\ldots,K$, an encoder
$f:\mathcal{W}_1\times\ldots\times\mathcal{W}_K\rightarrow
\mathcal{X}^{n}$, $K$ decoders, one at each legitimate receiver,
$g_k:\mathcal{Y}_k\rightarrow \mathcal{W}_k$, $k=1,\ldots,K$. The
probability of error is defined as
$P_e^n=\max_{k=1,\ldots,K}\Pr\left[g_k(Y_k^n)\neq (W_k)\right]$. A
rate tuple $(R_1,\ldots,R_K)$ is said to be achievable if there
exists a code with $\lim_{n\rightarrow\infty}P_e^n=0$ and
\begin{align}
\lim_{n\rightarrow\infty}\frac{1}{n}H(\mathcal{S}(W)|Z^n)\geq
\sum_{k\in \mathcal{S}(W)}R_k, \quad \forall ~\mathcal{S}(W)
\label{perfect_secrecy}
\end{align}
where $\mathcal{S}(W)$ denotes any subset of $\{W_1,\ldots,W_K\}$.
Hence, we consider only perfect secrecy rates. The secrecy
capacity region is defined as the closure of all achievable rate
tuples.

The degraded multi-receiver wiretap channel exhibits the following
Markov chain
\begin{align}
X\rightarrow Y_1\rightarrow\ldots\rightarrow Y_K\rightarrow Z
\label{degraded_wc}
\end{align}
whose capacity region was established in
\cite{Ekrem_Ulukus_Asilomar08, Ekrem_Ulukus_BC_Secrecy} for an
arbitrary number of users and in \cite{Khandani} for two users.

\begin{Theo}
\label{degraded_multiuser_wiretap} The secrecy capacity region of
the degraded multi-receiver wiretap channel is given by the union
of the rate tuples $(R_1,\ldots,R_K)$ satisfying\footnote{Although
in~\cite{Ekrem_Ulukus_Asilomar08, Ekrem_Ulukus_BC_Secrecy}, this
secrecy capacity region is expressed in a different form, the
equivalence of the two expressions can be shown.}
\begin{align}
R_k & \leq  I(U_k;Y_k|U_{k+1},Z), \quad k=1,\ldots,K
\label{capacity_expressions}
\end{align}
where $U_1=X,U_{K+1}=\phi$, and the union is over all probability
distributions of the form
\begin{align}
p(u_K)p(u_{K-1}|u_{K})\ldots p(u_{2}|u_{3})p(x|u_{2})
\label{joint_distribution}
\end{align}
\end{Theo}

We remark here that since the channel is degraded, i.e., we have
the Markov chain in (\ref{degraded_wc}), the capacity expressions
in (\ref{capacity_expressions}) are equivalent to
\begin{align}
R_k & \leq  I(U_k;Y_k|U_{k+1})-I(U_k;Z|U_{k+1}), \hspace{0.2cm}
k=1,\ldots,K \label{equivalent_degraded_exp}
\end{align}
We will use this equivalent expression frequently hereafter. For
the case of two users and one eavesdropper, i.e., $K=2$, the
expressions in (\ref{equivalent_degraded_exp}) reduce to:
\begin{align}
R_1&\leq I(X;Y_1|U_2)-I(X;Z|U_2)\label{degraded_multiuser_wiretap_rate_two_user_1}\\
R_2&\leq I(U_2;Y_2)-I(U_2;Z)
\label{degraded_multiuser_wiretap_rate_two_user_2}
\end{align}
Finding the secrecy capacity region of the two-user degraded
multi-receiver wiretap channel is tantamount to finding the
optimal joint distributions of $(X,U_2)$ that trace the boundary
of the secrecy capacity region given in
(\ref{degraded_multiuser_wiretap_rate_two_user_1})-(\ref{degraded_multiuser_wiretap_rate_two_user_2}).
For the $K$-user degraded multi-receiver wiretap channel, we need
to find the optimal joint distributions of $(X,U_2,\ldots,U_K)$ in
the form given in (\ref{joint_distribution}) that trace the
boundary of the region expressed in (\ref{capacity_expressions}).

\section{Gaussian MIMO Multi-receiver Wiretap Channel}

\subsection{Degraded Gaussian MIMO Multi-receiver Wiretap Channel}
In this paper, we first consider the degraded Gaussian MIMO
multi-receiver wiretap channel which is defined through
\begin{align}
\bby_k&=\bbx+\bbn_k,\quad k=1,\ldots,K \label{channel_def_1} \\
\bbz&=\bbx+\bbn_Z \label{channel_def_2}
\end{align}
where the channel input $\bbx$ is subject to a covariance
constraint
\begin{align}
E\left[\bbx \bbx^\top \right]\preceq \bbs
\label{covariance_constraint}
\end{align}
where $\bbs\succ 0$, and $\{\bbn_k\}_{k=1}^{K},\bbn_Z$ are
zero-mean Gaussian random vectors with covariance matrices
$\{\bbsigma_k\}_{k=1}^K,\bbsigma_Z$ which satisfy the following
ordering
\begin{align}
\bzero   \prec \bbsigma_1 \preceq \bbsigma_2\preceq \ldots\preceq
\bbsigma_K\preceq \bbsigma_Z \label{covariance_matrices}
\end{align}

In a multi-receiver wiretap channel, since the
capacity-equivocation rate region depends only on the conditional
marginal distributions of the transmitter-receiver links, but not
on the entire joint distribution of the channel, the correlations
among $\{\bbn_k\}_{k=1}^{K},\bbn_Z$ have no consequence on the
capacity-equivocation rate region. Thus, without changing the
corresponding secrecy capacity region, we can adjust the
correlation structure among these noise vectors to ensure that
they satisfy the following Markov chain
\begin{align}
\bbx\rightarrow \bby_1\rightarrow \ldots\rightarrow
\bby_K\rightarrow \bbz \label{degraded_Markov_chain}
\end{align}
which is always possible because of our assumption regarding the
covariance matrices in~(\ref{covariance_matrices}). Moreover, the
Markov chain in (\ref{degraded_Markov_chain}) implies that any
Gaussian MIMO multi-receiver wiretap channel satisfying the
semi-definite ordering in (\ref{covariance_matrices}) can be
treated as a degraded multi-receiver wiretap channel, hence
Theorem~\ref{degraded_multiuser_wiretap} gives its capacity
region. Hereafter, we will assume that the degraded Gaussian MIMO
wiretap channel satisfies the Markov chain in
(\ref{degraded_Markov_chain}).

\subsection{Aligned Gaussian MIMO Multi-receiver Wiretap Channel}

Next, we consider the aligned Gaussian MIMO multi-receiver wiretap
channel which is again defined
by~(\ref{channel_def_1})-(\ref{channel_def_2}), and the input is
again subject to a covariance constraint as
in~(\ref{covariance_constraint}) with $\bbs \succ \bzero$.
However, for the aligned Gaussian MIMO multi-receiver wiretap
channel, noise covariance matrices do not have any semi-definite
ordering, as opposed to the degraded case which exhibits the
ordering in~(\ref{covariance_matrices}). For the aligned Gaussian
MIMO multi-receiver wiretap channel, the only assumption on the
noise covariance matrices is that they are strictly
positive-definite, i.e., $\bbsigma_i\succ \bzero,~i=1,\ldots,K$
and $\bbsigma_Z \succ \bzero $. Since this channel does not have
any ordering among the noise covariance matrices, it cannot be
considered as a degraded channel, thus there is no single-letter
formula for its secrecy capacity region. Moreover, we do not
expect superposition coding with stochastic encoding to be
optimal, as it was optimal for the degraded channel. Indeed, we
will show that dirty-paper coding with stochastic encoding is
optimal in this case.

\subsection{General Gaussian MIMO Multi-receiver Wiretap Channel}

Finally, we consider the most general form of the Gaussian MIMO
multi-receiver wiretap channel which is given by
\begin{align}
\bby_k&=\bbh_k\bbx+\bbn_k,\quad k=1,\ldots,K \label{general_mimo_original_def1}\\
\bbz&=\bbh_{Z}\bbx+\bbn_Z \label{general_mimo_original_def2}
\end{align}
where the channel input $\bbx$, which is a $t\times 1$ column
vector, is again subject to a covariance constraint as in
(\ref{covariance_constraint}) with $\bbs \succeq \bzero $. The
channel output for the $k$th user is denoted by $\bby_k$ which is
a column vector of size $r_k\times 1,~k=1,\ldots,K$. The
eavesdropper's observation $\bbz$ is of size $r_Z\times 1$. The
covariance matrices of the Gaussian random vectors
$\left\{\bbn_k\right\}_{k=1}^K,\bbn_{Z}$ are denoted by
$\left\{\bbsigma_k\right\}_{k=1}^K,\bbsigma_Z$, which are assumed
to be strictly positive definite. The channel gain matrices
$\left\{\bbh_k\right\}_{k=1}^K,\bbh_Z$ are of sizes
$\left\{r_k\times t\right\}_{k=1}^K,r_Z\times t$, respectively,
and they are known to the transmitter, all legitimate users and
the eavesdropper.

\subsection{A Comment on the Covariance Constraint}
In the literature, it is more common to define capacity regions
under a total power constraint, i.e., ${\rm
tr}\left(E\left[\bbx\bbx^\top\right]\right)\leq P$, instead of the
covariance constraint that we imposed, i.e., $E\left[\bbx
\bbx^\top\right]\preceq \bbs$. However, as shown
in~\cite{Shamai_MIMO}, once the capacity region is obtained under
a covariance constraint, then the capacity region under more
lenient constraints on the channel inputs can be obtained, if
these constraints can be expressed as compact sets defined over
the input covariance matrices. For example, the total power
constraint and the per-antenna power constraint can be described
by compact sets of input covariance matrices as follows
\begin{align}
\mathcal{S}^{\rm total}&=\{\bbs\succeq\bzero: {\rm tr}(\bbs)\leq
P\}\\
\mathcal{S}^{\rm per-ant}&=\{\bbs\succeq\bzero: \bbs_{ii}\leq
P_i,~i=1,\ldots,t\}
\end{align}
respectively, where $\bbs_{ii}$ is the $i$th diagonal entry of
$\bbs$, and $t$ denotes the number of transmit antennas. Thus, if
the secrecy capacity region under a covariance constraint
$E\left[\bbx \bbx^\top\right]\preceq\bbs$ is found and denoted by
$\mathcal{C}(\bbs)$, then the secrecy capacity regions under the
total power constraint and the per-antenna power constraint can be
expressed as
\begin{align}
\mathcal{C}^{\rm
total}&=\bigcup_{\bbs\in\mathcal{S}^{\rm total}}\mathcal{C}(\bbs)\\
\mathcal{C}^{\rm per-ant}&=\bigcup_{\bbs\in\mathcal{S}^{\rm
per-ant}}\mathcal{C}(\bbs)
\end{align}
respectively.

One other comment about the covariance constraint on the channel
input is regarding the positive definiteness of $\bbs$. Following
Lemma~2 of \cite{Shamai_MIMO}, it can be shown that, for any
degraded (resp. aligned) Gaussian MIMO multi-receiver channel
under a covariance constraint $E\left[\bbx \bbx^\top\right]\preceq
\bbs$ where $\bbs$ is a non-invertible positive semi-definite
matrix, i.e., $\bbs \succeq \bzero$ and $|\bbs|=0$, we can find
another equivalent degraded (resp. aligned) channel with fewer
transmit and receive antennas under a covariance constraint
$E\big[\hat{\bbx} \hat{\bbx}^\top\big]\preceq \bbs^\prime$ such
that $\bbs ^\prime \succ \bzero$. Here the equivalence refers to
the fact that both of these channels will have the same secrecy
capacity region. Thus, as long as a degraded or an aligned channel
is considered, there is no loss of generality in imposing a
covariance constraint with a strictly positive definite matrix
$\bbs$, and this is why we assumed that $\bbs$ is strictly
positive definite for the degraded and the aligned channels.

\section{Gaussian SISO Multi-receiver Wiretap Channel}

We first visit the Gaussian SISO multi-receiver wiretap channel.
The aims of this section are to show that a straightforward
extension of existing converse techniques for the Gaussian scalar
broadcast channel fails to provide a converse proof for the
Gaussian SISO multi-receiver wiretap channel, and to provide an
alternative proof technique using either the MMSE or the Fisher
information along with their connections with the differential
entropy. To this end, we first define the Gaussian SISO
multi-receiver wiretap channel
\begin{align}
Y_k&=X+N_k,\quad k=1,2 \\
Z&=X+N_{Z}
\end{align}
where we also restrict our attention to the two-user case for
simplicity of the presentation. The channel input $X$ is subject
to a power constraint $E\left[X^2\right]\leq P$. The variances of
the zero-mean Gaussian random variables $N_1,N_2,N_Z$ are given by
$\sigma^2_{1},\sigma^2_2,\sigma_Z^2$, respectively, and satisfy
the following order
\begin{align}
\sigma^2_{1}\leq \sigma^2_2\leq \sigma_Z^2
\end{align}
Since the correlations among $N_1,N_2,N_Z$ have no effect on the
secrecy capacity region, we can adjust the correlation structure
to ensure the following Markov chain
\begin{align}
X\rightarrow Y_1\rightarrow Y_2\rightarrow Z
\end{align}
Thus, this channel can be considered as a degraded channel, and
its secrecy capacity region is given by
Theorem~\ref{degraded_multiuser_wiretap}, in particular, by
(\ref{degraded_multiuser_wiretap_rate_two_user_1}) and
(\ref{degraded_multiuser_wiretap_rate_two_user_2}). Hence, to
compute the secrecy capacity region explicitly, we need to find
the optimal joint distributions of $(X,U_2)$ in
(\ref{degraded_multiuser_wiretap_rate_two_user_1}) and
(\ref{degraded_multiuser_wiretap_rate_two_user_2}). The
corresponding secrecy capacity region is given by the following
theorem.
\begin{Theo}
\label{theorem_epi_is_not_sufficient} The secrecy capacity region
of the two-user Gaussian SISO wiretap channel is given by the
union of the rate pairs $(R_1,R_2)$ satisfying
\begin{align}
R_1 &\leq \frac{1}{2} \log \left(1+\frac{\alpha P
}{\sigma_1^2}\right) -\frac{1}{2} \log \left(1+\frac{\alpha P
}{\sigma_Z^2}\right) \label{siso_rate_user_1}\\
R_2 &\leq \frac{1}{2}\log\left(1+\frac{\bar{\alpha} P}{\alpha
P+\sigma_2^2}\right)- \frac{1}{2}\log\left(1+\frac{\bar{\alpha}
P}{\alpha P+\sigma_Z^2}\right) \label{siso_rate_user_2}
\end{align}
where the union is over all $\alpha \in[0,1]$, and $\bar{\alpha}$
denotes $1-\alpha$.
\end{Theo}

The achievability of this region can be shown by selecting
$(X,U_2)$ to be jointly Gaussian in
Theorem~\ref{degraded_multiuser_wiretap}. We focus on the converse
proof.

\subsection{Insufficiency of the Entropy-Power Inequality}

\label{sec:epi_insufficient}

As a natural approach, one might try to adopt the converse proofs
of the scalar Gaussian broadcast channel for the converse proof of
Theorem~\ref{theorem_epi_is_not_sufficient}. In the literature,
there are two converses for the Gaussian scalar broadcast channel
which share some main principles. The first converse was given by
Bergmans~\cite{Bergmans} who used Fano's lemma in conjunction with
the entropy-power inequality~\cite{Stam,Blachman} to find the
capacity region. Later, El Gamal gave a relatively simple
proof~\cite{El_Gamal_Converse} which does not recourse to Fano's
lemma. Rather, he started from the single-letter expression for
the capacity region and used entropy-power
inequality~\cite{Stam,Blachman} to evaluate this region. Thus, the
entropy-power inequality~\cite{Stam,Blachman} is the main
ingredient of these converses.

We now attempt to extend these converses to our secrecy context,
i.e., to provide the converse proof of
Theorem~\ref{theorem_epi_is_not_sufficient}, and show where the
argument breaks. In particular, what we will show in the following
discussion is that a stand-alone use of the entropy-power
inequality~\cite{Stam,Blachman} falls short of proving the
optimality of Gaussian signalling in this secrecy context, as
opposed to the Gaussian scalar broadcast channel. For that
purpose, we consider El Gamal's converse for the Gaussian scalar
broadcast channel. However, since the entropy-power inequality is
in a central role for both El Gamal's and Bergmans' converse, the
upcoming discussion can be carried out by using Bergmans' proof as
well.

First, we consider the bound on the second user's secrecy rate.
Using (\ref{degraded_multiuser_wiretap_rate_two_user_2}), we have
\begin{align}
I(U_2;Y_2)-I(U_2;Z)=\big[I(X;Y_2)-I(X;Z)\big]-\big[I(X;Y_2|U_2)-I(X;Z|U_2)\big]
\label{epi_is_not_sufficient1}
\end{align}
where the right-hand side is obtained by using the chain rule, and
the Markov chain $U_2\rightarrow X\rightarrow (Y_1,Y_2,Z)$. The
expression in the first bracket is maximized by Gaussian
$X$~\cite{hellman} yielding
\begin{align}
I(X;Y_2)-I(X;Z) \leq \frac{1}{2} \log
\left(1+\frac{P}{\sigma_2^2}\right) -\frac{1}{2} \log
\left(1+\frac{P}{\sigma_Z^2}\right) \label{epi_is_not_sufficient2}
\end{align}
Moreover, using the Markov chain $U_2\rightarrow X \rightarrow
Y_2\rightarrow Z$, we can bound the expression in the second
bracket as
\begin{align}
0&\leq I(X;Y_2|U_2)-I(X;Z|U_2) \\
& \leq I(X;Y_2)-I(X;Z)\\
 &\leq \frac{1}{2} \log
\left(1+\frac{P}{\sigma_2^2}\right) -\frac{1}{2} \log
\left(1+\frac{P}{\sigma_Z^2}\right)
\end{align}
which implies that for any $(X,U_2)$ pair, there exists an
$\alpha\in[0,1]$ such that
\begin{align}
I(X;Y_2|U_2)-I(X;Z|U_2) =\frac{1}{2} \log \left(1+\frac{\alpha
P}{\sigma_2^2}\right) -\frac{1}{2} \log \left(1+\frac{\alpha
P}{\sigma_Z^2}\right) \label{epi_is_not_sufficient3}
\end{align}
Combining (\ref{epi_is_not_sufficient2}) and
(\ref{epi_is_not_sufficient3}) in (\ref{epi_is_not_sufficient1})
yields the desired bound on $R_2$ given in
(\ref{siso_rate_user_2}).

From now on, we focus on obtaining the bound given in
(\ref{siso_rate_user_1}) on the first user's secrecy rate. To this
end, one needs to solve the following
optimization\footnote{Equivalently, one can consider the following
optimization problem
\begin{align}
&\max~I(X;Y_1|U_2)-I(X;Y_2|U_2)\nonumber \\
&\textrm{s.t. }~~I(X;Y_2|U_2)-I(X;Z|U_2)=\frac{1}{2} \log
\left(1+\frac{\alpha P}{\sigma_2^2}\right) -\frac{1}{2} \log
\left(1+\frac{\alpha P}{\sigma_Z^2}\right) \nonumber
\end{align}
which, in turn, would yield a similar contradiction.}
\begin{align}
&\max~I(X;Y_1|U_2)-I(X;Z|U_2) \label{epi_is_not_sufficient_obj}\\
&\textrm{s.t. }~~I(X;Y_2|U_2)-I(X;Z|U_2)=\frac{1}{2} \log
\left(1+\frac{\alpha P}{\sigma_2^2}\right) -\frac{1}{2} \log
\left(1+\frac{\alpha P}{\sigma_Z^2}\right)
\label{epi_is_not_sufficient_constraint}
\end{align}
When the term $I(X;Z|U_2)$ is absent in both the objective
function and the constraint, as in the case of the Gaussian scalar
broadcast channel, the entropy-power
inequality~\cite{Stam,Blachman} can be used to solve this
optimization problem. However, the presence of this term
complicates the situation, and a stand-alone use of the
entropy-power inequality~\cite{Stam,Blachman} does not seem to be
sufficient. To substantiate this claim, let us consider the
objective function in (\ref{epi_is_not_sufficient_obj})
\begin{align}
I(X;Y_1|U_2)-I(X;Z|U_2)&=h(Y_1|U_2)-h(Z|U_2)-\frac{1}{2}\log\frac{\sigma_1^2}{\sigma_Z^2}\\
&\leq \frac{1}{2} \log \Big(e^{2h(Z|U_2)}-2\pi e
\left(\sigma_Z^2-\sigma_1^2\right)\Big)-h(Z|U_2)-\frac{1}{2}\log\frac{\sigma_1^2}{\sigma_Z^2}
\label{epi_is_not_sufficient4}
\end{align}
where the inequality is obtained by using the entropy-power
inequality. Since the right-hand side of
(\ref{epi_is_not_sufficient4}) is monotonically increasing in
$h(Z|U_2)$, to show the optimality of Gaussian signalling, we need
\begin{align}
h(Z|U_2)\leq \frac{1}{2}\log 2\pi e (\alpha P+\sigma_Z^2)
\label{that_is_wut_we_need}
\end{align}
which will result in the desired bound on
(\ref{epi_is_not_sufficient_obj}), i.e., the following
\begin{align}
I(X;Y_1|U_2)-I(X;Z|U_2) \leq \frac{1}{2} \log \left(1+\frac{
\alpha P}{\sigma_1^2}\right) -\frac{1}{2} \log
\left(1+\frac{\alpha P}{\sigma_Z^2}\right)
\end{align}
which is the desired end-result in (\ref{siso_rate_user_1}).

We now check whether (\ref{that_is_wut_we_need}) holds under the
constraint given in (\ref{epi_is_not_sufficient_constraint}). To
this end, consider the difference of mutual informations in
(\ref{epi_is_not_sufficient_constraint})
\begin{align}
I(X;Y_2|U_2)-I(X;Z|U_2)&=h(Y_2|U_2)-h(Z|U_2)-\frac{1}{2}\log\frac{\sigma_2^2}{\sigma_Z^2}\\
&\leq \frac{1}{2} \log \Big(e^{2h(Z|U_2)}-2\pi e
\left(\sigma_Z^2-\sigma_2^2\right)\Big)-h(Z|U_2)-\frac{1}{2}\log\frac{\sigma_2^2}{\sigma_Z^2}
\label{epi_is_not_sufficient5}
\end{align}
where the inequality is obtained by using the entropy-power
inequality. Now, using the constraint given in
(\ref{epi_is_not_sufficient_constraint}) in
(\ref{epi_is_not_sufficient5}), we get
\begin{align}
\frac{1}{2}\log\left( \frac{ \alpha P+\sigma_2^2}{\alpha
P+\sigma_Z^2}\right) \leq \frac{1}{2}\log \Big(e^{2h(Z|U_2)}-2\pi
e \left(\sigma_Z^2-\sigma_2^2\right)\Big)-h(Z|U_2)
\end{align}
which implies
\begin{align}
\frac{1}{2}\log 2\pi e (\alpha P+\sigma_Z^2)\leq  h(Z|U_2)
\end{align}
Thus, as opposed to the inequality that we need to show the
optimality of Gaussian signalling via the entropy-power
inequality, i.e., the bound in (\ref{that_is_wut_we_need}), we
have an opposite inequality. This discussion reveals that if
Gaussian signalling is optimal, then its proof cannot be deduced
from a straightforward extension of the converse proofs for the
Gaussian scalar broadcast channel
in~\cite{Bergmans,El_Gamal_Converse}. Thus, we need a new
technique to provide the converse for
Theorem~\ref{theorem_epi_is_not_sufficient}. We now present two
different proofs. The first proof relies on the relationship
between the MMSE and the mutual information along with the
properties of the MMSE, and the second proof replaces the MMSE
with the Fisher information.

\subsection{Converse for Theorem~\ref{theorem_epi_is_not_sufficient} Using the MMSE}
We now provide a converse which uses the connection between the
MMSE and the mutual information established
in~\cite{Guo_ISIT_08,Guo_MMSE}. In~\cite{Guo_ISIT_08}, the authors
also give an alternative converse for the scalar Gaussian
broadcast channel. Our proof will follow this converse, and
generalize it to the context where there are secrecy constraints.

First, we briefly state the necessary background information. Let
$N$ be a zero-mean unit-variance Gaussian random variable, and
$(U,X)$ be a pair of arbitrarily correlated random variables which
are independent of $N$. The MMSE of $X$ when it is observed
through $U$ and $\sqrt{t}X+N$ is
\begin{align}
{\rm
mmse}(X,t|U)=E\left[\left(X-E\left[X|\sqrt{t}X+N,U\right]\right)^2\right]
\end{align}
As shown in~\cite{Guo_ISIT_08,Guo_MMSE}, the MMSE and the
conditional mutual information are related through
\begin{align}
I(X;\sqrt{t}X+N|U)=\frac{1}{2}\int_{0}^{t}{\rm mmse}(X,t|U)dt
\label{mmse_mutual_info}
\end{align}
For our converse, we need the following proposition which was
proved in~\cite{Guo_ISIT_08}. \vspace{-0.2cm}
\begin{Prop}[\!\!\cite{Guo_ISIT_08},Proposition 12]
\label{Prop_unique_zero} Let $U,X,N$ be as specified above. The
function
\begin{align}
f(t)=\frac{\sigma^2}{\sigma^2 t +1}-{\rm mmse}(X,t|U)
\end{align}
has at most one zero in $[0,\infty)$ unless $X$ is Gaussian
conditioned on $U$ with variance $\sigma^2$, in which case the
function is identically zero on $[0,\infty)$. In particular, if
$t_0 < \infty$ is the unique zero, then $f(t)$ is strictly
increasing on $[0,t_0]$, and strictly positive on $(t_0,\infty)$.
\vspace{-0.2cm}
\end{Prop}

We now give the converse. We use exactly the same steps from
(\ref{epi_is_not_sufficient1}) to (\ref{epi_is_not_sufficient3})
to establish the bound on the secrecy rate of the second user
given in (\ref{siso_rate_user_2}). To bound the secrecy rate of
the first user, we first restate (\ref{epi_is_not_sufficient3}) as
\begin{align}
I(X;Y_2|U_2)-I(X;Z|U_2)&=I(X;(1/\sigma_2)X+N|U_2)-I(X;(1/\sigma_Z)X+N|U_2)
\\
&=\frac{1}{2} \log \left(1+\frac{\alpha P}{\sigma_2^2}\right)
-\frac{1}{2} \log
\left(1+\frac{\alpha P}{\sigma_Z^2}\right)\\
&=\frac{1}{2} \int_{1/\sigma_Z^2}^{1/\sigma_2^2} \frac{\alpha P}{t
\alpha P+1}dt \label{mmse_based_converse1}
\end{align}
Furthermore, due to (\ref{mmse_mutual_info}), we also have
\begin{align}
I(X;Y_2|U_2)-I(X;Z|U_2)&=I(X;(1/\sigma_2)X+N|U_2)-I(X;(1/\sigma_Z)X+N|U_2)\nonumber\\
&=\frac{1}{2} \int_{1/\sigma_Z^2}^{1/\sigma_2^2} {\rm
mmse}(X,t|U_2) dt \label{mmse_based_converse2}
\end{align}
Comparing (\ref{mmse_based_converse1}) and
(\ref{mmse_based_converse2}) reveals that either we have
\begin{align}
{\rm mmse}(X,t|U_2)=\frac{ \alpha P}{t \alpha P+1}
\end{align}
for all $t\in[1/\sigma_Z^2,1/\sigma_2^2]$, or there exists a
unique $t_0\in(1/\sigma_Z^2,1/\sigma_2^2)$ such that
\begin{align}
{\rm mmse}(X,t_0|U_2)=\frac{\alpha P}{t_0 \alpha P+1}
\label{epi_is_not_sufficient6}
\end{align}
and
\begin{align}
{\rm mmse}(X,t|U_2)\leq \frac{ \alpha P}{t \alpha P+1}
\label{epi_is_not_sufficient7}
\end{align}
for $t>t_0$, because of Proposition~\ref{Prop_unique_zero}. The
former case occurs if $X$ is Gaussian conditioned on $U_2$ with
variance $\alpha P$, in which case we arrive at the desired bound
on the secrecy rate of the first user given in
(\ref{siso_rate_user_1}). If we assume that the latter case in
(\ref{epi_is_not_sufficient6})-(\ref{epi_is_not_sufficient7})
occurs, then, we can use the following sequence of derivations to
bound the first user's secrecy rate
\begin{align}
I(X;Y_1|U_2)-I(X;Z|U_2)&=
I(X;(1/\sqrt{\sigma_1})X+N|U_2)-I(X;(1/\sqrt{\sigma_Z})X+N|U_2)\\
&=\frac{1}{2}\int_{1/\sigma_Z^2}^{1/\sigma_1^2}{\rm
mmse}(X,t|U_2)dt\\
&=\frac{1}{2}\int_{1/\sigma_Z^2}^{1/\sigma_2^2}{\rm
mmse}(X,t|U_2)dt+\frac{1}{2}\int_{1/\sigma_2^2}^{1/\sigma_1^2}{\rm
mmse}(X,t|U_2)dt\\
&=\frac{1}{2} \log \left(1+\frac{ \alpha P}{\sigma_2^2}\right)
-\frac{1}{2} \log
\left(1+\frac{\alpha P}{\sigma_Z^2}\right)\nonumber\\
&\quad +\frac{1}{2}\int_{1/\sigma_2^2}^{1/\sigma_1^2}{\rm
mmse}(X,t|U_2)dt \label{epi_is_not_sufficient8}\\
&\leq \frac{1}{2} \log \left(1+\frac{\alpha P}{\sigma_2^2}\right)
-\frac{1}{2} \log \left(1+\frac{\alpha
P}{\sigma_Z^2}\right)+\frac{1}{2}\int_{1/\sigma_2^2}^{1/\sigma_1^2}
\frac{ \alpha P}{t \alpha P+1} dt \label{epi_is_not_sufficient9}\\
&=\frac{1}{2} \log \left(1+\frac{\alpha P}{\sigma_1^2}\right)
-\frac{1}{2} \log \left(1+\frac{\alpha P}{\sigma_Z^2}\right)
\label{epi_is_not_sufficient10}
\end{align}
where (\ref{epi_is_not_sufficient8}) follows from
(\ref{mmse_based_converse1}) and (\ref{mmse_based_converse2}), and
(\ref{epi_is_not_sufficient9}) is due to
(\ref{epi_is_not_sufficient7}). Since
(\ref{epi_is_not_sufficient10}) is the desired bound on the
secrecy rate of the first user given in (\ref{siso_rate_user_1}),
this completes the converse proof.

\subsection{Converse for Theorem~\ref{theorem_epi_is_not_sufficient} Using the Fisher Information}
\label{sec:tgi_fisher}

We now provide an alternative converse which replaces the MMSE
with the Fisher information in the above proof. We first provide
some basic definitions. The unconditional versions of the
following definition and the upcoming results regarding the Fisher
information can be found in standard detection-estimation texts;
to note one, \cite{oliver_johnson} is a good reference for a
detailed treatment of the subject.

\begin{Def}
Let $X,U$ be arbitrarily correlated random variables with
well-defined densities, and $f(x|u)$ be the corresponding
conditional density. The conditional Fisher information of $X$ is
defined by
\begin{align}
J(X|U)=E\left[\left(\frac{\partial \log f(x|u)}{\partial
x}\right)^2\right]
\end{align}
where the expectation is over $(U,X)$.
\end{Def}

The vector generalization of the following conditional form of the
Fisher information inequality will be given in
Lemma~\ref{lemma_cond_matrix_fii} in
Section~\ref{secr:proof_of_thm_conditional_vec_generalization},
thus its proof is omitted here.

\begin{Lem}
Let $U,X,Y$ be random variables, and let the density for any
combination of them exist. Moreover, let us assume that given $U$,
$X$ and $Y$ are independent. Then, we have
\begin{align}
J(X+Y|U) \leq \beta^2 J(X|U)+(1-\beta)^2 J(Y|U)
\end{align}
for any $\beta \in [0,1]$.
\end{Lem}

\begin{Cor}
\label{corollary_cond_fii} Let $X,Y,U$ be as specified above.
Then, we have
\begin{align}
\frac{1}{J(X+Y|U)}\geq \frac{1}{J(X|U)}+\frac{1}{J(Y|U)}
\end{align}
\end{Cor}
\begin{proof}
Select
\begin{align}
\beta=\frac{J(Y|U)}{J(X|U)+J(Y|U)}
\end{align}
in the previous lemma.
\end{proof}

Similarly, the vector generalization of the following conditional
form of the Cramer-Rao inequality will be given in
Lemma~\ref{lemma_conditional_crb_vector} in
Section~\ref{secr:proof_of_thm_conditional_vec_generalization},
and hence, its proof is omitted here.

\begin{Lem}
\label{lemma_conditional_crb} Let $X,U$ be arbitrarily correlated
random variables with well-defined densities. Then, we have
\begin{align}
J(X|U)\geq \frac{1}{{\rm Var}(X|U)}
\end{align}
with equality if $(U,X)$ is jointly Gaussian.
\end{Lem}

We now provide the conditional form of the de Bruin
identity~\cite{Stam,Blachman}. The vector generalization of this
lemma will be provided in Lemma~\ref{gradient_fisher_conditional}
in Section~\ref{secr:proof_of_thm_conditional_vec_generalization},
and hence, its proof is omitted here.

\begin{Lem}
\label{lemma_de_bruin} Let $X,U$ be arbitrarily correlated random
variables with finite second order moments. Moreover, assume that
they are independent of $N$ which is a zero-mean unit-variance
Gaussian random variable. Then, we have
\begin{align}
\frac{d h(X+\sqrt{t}N|U)}{d t}=\frac{1}{2}J(X+\sqrt{t}N|U)
\end{align}
\end{Lem}

We now note the following complementary relationship between the
MMSE and the Fisher information~\cite{Olivier_Rioul,Guo_MMSE}
\begin{align}
J(\sqrt{t}X+N)=1-t\cdot{\rm mmse}(X,t)
\label{complementary_relationship_between_Fisher_mmse}
\end{align}
which itself suggests the existence of an alternative converse
which uses the Fisher information instead of the MMSE. We now
provide the alternative converse based on the Fisher information.
We first bound the secrecy rate of the second user as in the
previous section, by following the exact steps from
(\ref{epi_is_not_sufficient1}) to (\ref{epi_is_not_sufficient3}).
To bound the secrecy rate of the first user, we first rewrite
(\ref{epi_is_not_sufficient3}) as follows
\begin{align}
I(X;Y_2|U_2)-I(X;Z|U_2)&=h(X+\sigma_2 N|U_2)-h(X+\sigma_Z
N|U_2)-\frac{1}{2}\log\frac{\sigma_2^2}{\sigma_Z^2}\\
&=-\frac{1}{2} \int_{\sigma_2^2}^{\sigma_Z^2}
J(X+\sqrt{t}N|U_2)dt-\frac{1}{2}\log\frac{\sigma_2^2}{\sigma_Z^2}\label{de_bruin_implies}\\
&=-\frac{1}{2} \int_{\sigma_2^2}^{\sigma_Z^2}
J(X+\sqrt{t-t^{*}}N^{\prime}+\sqrt{t^{*}}N^{\prime\prime}|U_2)dt-\frac{1}{2}\log\frac{\sigma_2^2}{\sigma_Z^2}
\label{stability_gaussian_2}
\end{align}
where (\ref{de_bruin_implies}) follows from
Lemma~\ref{lemma_de_bruin}, and in (\ref{stability_gaussian_2}),
we used the stability of Gaussian random variables where,
$N^{\prime}, N^{\prime\prime}$ are two independent zero-mean
unit-variance Gaussian random variables. Moreover, $t^*$ is
selected in the range of $\left(0,\sigma_2^2\right)$. We now use
Corollary~\ref{corollary_cond_fii} to bound the conditional Fisher
information in (\ref{stability_gaussian_2}) as follows
\begin{align}
\frac{1}{J(X+\sqrt{t-t^{*}}N^{\prime}+\sqrt{t^{*}}N^{\prime\prime}|U_2)}
&\geq \frac{1}{J(X+\sqrt{t^{*}}N^{\prime\prime}|U_2)}+\frac{1}{J(\sqrt{t-t^{*}}N^{\prime}|U_2)}\\
&= \frac{1}{J(X+\sqrt{t^{*}}N^{\prime\prime}|U_2)}+(t-t^*)
\label{cond_fii_implies}
\end{align}
where the equality follows from Lemma~\ref{lemma_conditional_crb}.
The inequality in (\ref{cond_fii_implies}) is equivalent to
\begin{align}
J(X+\sqrt{t-t^{*}}N^{\prime}+\sqrt{t^{*}}N^{\prime\prime}|U_2)
&\leq
\frac{J(X+\sqrt{t^{*}}N^{\prime\prime}|U_2)}{1+J(X+\sqrt{t^{*}}N^{\prime\prime}|U_2)(t-t^*)}
\end{align}
using which in (\ref{stability_gaussian_2}) yields
\begin{align}
I(X;Y_2|U_2)-I(X;Z|U_2)&\geq -\frac{1}{2}
\int_{\sigma_2^2}^{\sigma_Z^2}\frac{J(X+\sqrt{t^{*}}N^{\prime\prime}|U_2)}{1+J(X+\sqrt{t^{*}}N^{\prime\prime}|U_2)(t-t^*)}dt
-\frac{1}{2}\log\frac{\sigma_2^2}{\sigma_Z^2}\\
&=-\frac{1}{2} \log
\frac{1+J(X+\sqrt{t^{*}}N^{\prime\prime}|U_2)(\sigma_Z^2-t^*)}{1+J(X+\sqrt{t^{*}}N^{\prime\prime}|U_2)(\sigma_2^2-t^*)}
-\frac{1}{2}\log\frac{\sigma_2^2}{\sigma_Z^2}
\label{alt_conv_fisher_compare_1}
\end{align}
We remind that we had already fixed the left-hand side of this
inequality as
\begin{align}
I(X;Y_2|U_2)-I(X;Z|U_2)=\frac{1}{2} \log \left(1+\frac{\alpha
P}{\sigma_2^2}\right) -\frac{1}{2} \log \left(1+\frac{\alpha
P}{\sigma_Z^2}\right) \label{alt_conv_fisher_compare_2}
\end{align}
in (\ref{epi_is_not_sufficient3}). Comparison of
(\ref{alt_conv_fisher_compare_1}) and
(\ref{alt_conv_fisher_compare_2}) results in
\begin{align}
J(X+\sqrt{t^*}N^{\prime\prime}|U_2)\geq \frac{1}{\alpha
P+t^{*}},\quad 0<t^* \leq \sigma_2^2
\label{epi_was_never_sufficient}
\end{align}
At this point, we compare the inequalities in
(\ref{epi_is_not_sufficient7}) and
(\ref{epi_was_never_sufficient}). These two inequalities imply
each other through the complementary relationship between the MMSE
and the Fisher information given in
(\ref{complementary_relationship_between_Fisher_mmse}) after
appropriate change of variables and by noting that
$J(aX)=(1/a^2)J(X)$~\cite{oliver_johnson}. We now find the desired
bound on the secrecy rate of the first user via using the
inequality in (\ref{epi_was_never_sufficient})
\begin{align}
I(X;Y_1|U_2)-I(X;Z|U_2)&=h(X+\sigma_1 N|U_2)-h(X+\sigma_Z N|U_2)-\frac{1}{2}\log\frac{\sigma_1^2}{\sigma_Z^2}\\
&=-\frac{1}{2}\int_{\sigma_1^2}^{\sigma_Z^2}J(X+\sqrt{t}N|U_2)dt-\frac{1}{2}\log\frac{\sigma_1^2}{\sigma_Z^2}\\
&=-\frac{1}{2}\int_{\sigma_1^2}^{\sigma_2^2}J(X+\sqrt{t}N|U_2)dt
-\frac{1}{2}\int_{\sigma_2^2}^{\sigma_Z^2}J(X+\sqrt{t}N|U_2)dt\nonumber\\
&\quad -\frac{1}{2}\log\frac{\sigma_1^2}{\sigma_Z^2}\\
&=-\frac{1}{2}\int_{\sigma_1^2}^{\sigma_2^2}J(X+\sqrt{t}N|U_2)dt -
\frac{1}{2} \log \left(\frac{\alpha P+\sigma_Z^2}{ \alpha
P+\sigma_2^2}\right)
\nonumber\\
& \quad -\frac{1}{2}\log\frac{\sigma_1^2}{\sigma_Z^2} \label{epi_was_never_sufficient_1}\\
&\leq -\frac{1}{2}\int_{\sigma_1^2}^{\sigma_2^2}\frac{1}{\alpha
P+t}dt - \frac{1}{2} \log \left(\frac{\alpha P+\sigma_Z^2}{\alpha
P+\sigma_2^2}\right)
-\frac{1}{2}\log\frac{\sigma_1^2}{\sigma_Z^2} \label{epi_was_never_sufficient_2} \\
&= - \frac{1}{2} \log \left(\frac{ \alpha P+\sigma_2^2}{\alpha
P+\sigma_1^2}\right) - \frac{1}{2} \log \left(\frac{\alpha
P+\sigma_Z^2}{\alpha P+\sigma_2^2}\right)
-\frac{1}{2}\log\frac{\sigma_1^2}{\sigma_Z^2}\\
&=\frac{1}{2} \log\left(1+\frac{\alpha P}{\sigma_1^2}\right)
-\frac{1}{2} \log\left(1+\frac{\alpha P}{\sigma_Z^2}\right)
\label{epi_was_never_sufficient_3}
\end{align}
where (\ref{epi_was_never_sufficient_1}) follows from
(\ref{de_bruin_implies}) and (\ref{alt_conv_fisher_compare_2}),
and (\ref{epi_was_never_sufficient_2}) is due to
(\ref{epi_was_never_sufficient}). Since
(\ref{epi_was_never_sufficient_3}) provides the desired bound on
the secrecy rate of the first user given
in~(\ref{siso_rate_user_1}), this completes the converse proof.

\subsection{Summary of the SISO Case, Outlook for the MIMO Case}
In this section, we first revisited the standard converse
proofs~\cite{Bergmans,El_Gamal_Converse} of the Gaussian scalar
broadcast channel, and showed that a straightforward extension of
these proofs will not be able to provide a converse proof for the
Gaussian SISO multi-receiver wiretap channel. Basically, a
stand-alone use of the entropy-power
inequality~\cite{Stam,Blachman} falls short of resolving the
ambiguity on the auxiliary random variables. We showed that, in
this secrecy context, either the connection between the mutual
information and the MMSE or the connection between the
differential entropy and the Fisher information can be used, along
with their properties, to come up with a converse.

In the next section, we will generalize this converse proof
technique to the degraded MIMO channel. One way of generalizing
this converse technique to the MIMO case might be to use the
channel enhancement technique, which was successfully used in
extending Bergmans' converse proof from the scalar Gaussian
broadcast channel to the degraded vector Gaussian broadcast
channel. We note that such an extension will not work in this
secrecy context. In the degraded Gaussian MIMO broadcast channel,
the non-trivial part of the converse proof was to extend Bergmans'
converse to a vector case, and this was accomplished by the
invention of the channel enhancement technique. However, as we
have shown in Section~\ref{sec:epi_insufficient}, even in the
Gaussian SISO multi-receiver wiretap channel, a Bergmans type
converse does not work. Therefore, we will not pursue a channel
enhancement approach to extend our proof from the SISO channel to
the degraded MIMO channel. Instead, we will use the connections
between the Fisher information and the differential entropy, as we
did in Section~\ref{sec:tgi_fisher}, to come up with a converse
proof for the degraded MIMO channel. We will then use the channel
enhancement technique to extend our converse proof to the aligned
MIMO channel. Finally, we will use some limiting arguments, as
in~\cite{Tie_Liu_MIMO_WT,Shamai_MIMO}, to come up with a converse
proof for the most general MIMO channel.

\section{Degraded Gaussian MIMO Multi-receiver Wiretap Channel}

\label{sec:proof_of_main_result}

In this section, we establish the secrecy capacity region of the
degraded Gaussian MIMO multi-receiver wiretap channel. We state
the main result of this section in the following theorem.

\begin{Theo}
\label{main_result} The secrecy capacity region of the degraded
Gaussian MIMO multi-receiver wiretap channel is given by the union
of the rate tuples $R_1,\ldots,R_K$ satisfying
\begin{align}
R_k &\leq \frac{1}{2} \log
\frac{\left|\sum_{i=1}^{k}\bbk_i+\bbsigma_k\right|}{\left|\sum_{i=1}^{k-1}\bbk_i+\bbsigma_k\right|}
-\frac{1}{2} \log
\frac{\left|\sum_{i=1}^{k}\bbk_i+\bbsigma_Z\right|}{\left|\sum_{i=1}^{k-1}\bbk_i+\bbsigma_Z\right|},\quad
k=1,\ldots,K \label{sec_rates_degraded_general}
\end{align}
where the union is over all positive semi-definite matrices
$\{\bbk_i\}_{i=1}^K$ that satisfy
\begin{align}
\sum_{i=1}^K \bbk_i =\bbs
\end{align}
\end{Theo}

The achievability of these rates follows from
Theorem~\ref{degraded_multiuser_wiretap} by selecting
$(U_K,\ldots,U_2,\bbx)$ to be jointly Gaussian. Thus, to prove the
theorem, we only need to provide a converse. Since the converse
proof is rather long and involves technical digressions, we first
present the converse proof for $K=2$. In this process, we will
develop all necessary tools which we will use to provide the
converse proof for arbitrary $K$ in
Section~\ref{sec:proof_main_result_arbitrary_K}.

The secrecy capacity region of the two-user degraded MIMO channel,
from (\ref{sec_rates_degraded_general}), is the union of the rate
pairs $(R_1,R_2)$ satisfying
\begin{align}
R_1&\leq \frac{1}{2} \log
\frac{\left|\bbk_1+\bbsigma_1\right|}{\left|\bbsigma_1\right|}
-\frac{1}{2} \log
\frac{\left|\bbk_1+\bbsigma_Z\right|}{\left|\bbsigma_Z\right|} \label{sec_rates_degraded_two_user_1}\\
R_2&\leq \frac{1}{2} \log
\frac{\left|\bbs+\bbsigma_2\right|}{\left|\bbk_1+\bbsigma_2\right|}
-\frac{1}{2} \log
\frac{\left|\bbs+\bbsigma_Z\right|}{\left|\bbk_1+\bbsigma_Z\right|}
\label{sec_rates_degraded_two_user_2}
\end{align}
where the union is over all selections of $\bbk_1$ that satisfies
$\bzero \preceq \bbk_1 \preceq \bbs $. We note that these rates
are achievable by choosing $\bbx=\bbu_2+\bbv$ in
Theorem~\ref{degraded_multiuser_wiretap}, where $\bbu_2$ and
$\bbv$ are independent Gaussian random vectors with covariance
matrices $\bbs-\bbk_1$ and $\bbk_1$, respectively. Next, we prove
that the union of the rate pairs in
(\ref{sec_rates_degraded_two_user_1}) and
(\ref{sec_rates_degraded_two_user_2}) constitute the secrecy
capacity region of the two-user degraded MIMO channel.

\subsection{Proof of Theorem~\ref{main_result} for $K=2$}

To prove that (\ref{sec_rates_degraded_two_user_1}) and
(\ref{sec_rates_degraded_two_user_2}) give the secrecy capacity
region, we need the results of some intermediate optimization
problems. The first one is the so-called worst additive noise
lemma~\cite{Ihara,Diggavi_Cover}.
\begin{Lem}
\label{worst_additive} Let $\bbn$ be a Gaussian random vector with
covariance matrix $\bbsigma$, and $\bbk_X$ be a positive
semi-definite matrix. Consider the following optimization problem,
\begin{align}
\min_{p(\bx)}&\quad I(\bbn;\bbn+\bbx) \nonumber\\
{\rm s.t.}&\quad  {\rm Cov}(\bbx)=\bbk_X
\end{align}
where $\bbx$ and $\bbn$ are independent. A Gaussian $\bbx$ is the
minimizer of this optimization problem.
\end{Lem}

The second optimization problem that will be useful in the
upcoming proof is the conditional version of the following
theorem.

\begin{Theo}
\label{thm_vector_generalization} Let $\bbx,\bbn_1,\bbn_2,\bbn_Z$
be independent random vectors, where $\bbn_1,\bbn_2,\bbn_Z$ are
zero-mean Gaussian random vectors with covariance matrices
$\bzero\prec\bbsigma_1\preceq\bbsigma_2\preceq\bbsigma_Z$,
respectively. Moreover, assume that the second moment of $\bbx$ is
constrained as
\begin{align}
E\left[\bbx \bbx^\top\right]\preceq\bbs
\end{align}
where $\bbs$ is a positive definite matrix. Then, for any
admissible $\bbx$, there exists a matrix $\bbk^*$ such that
$\bzero \preceq \bbk^* \preceq \bbs$, and
\begin{align}
h(\bbx+\bbn_Z)-h(\bbx+\bbn_2)&=\frac{1}{2}\log\frac{\left|\bbk^*
+\bbsigma_Z\right|}{\left|\bbk^*+\bbsigma_2\right|}\\
h(\bbx+\bbn_Z)-h(\bbx+\bbn_1)&\geq\frac{1}{2}\log\frac{\left|\bbk^*
+\bbsigma_Z\right|}{\left|\bbk^*+\bbsigma_1\right|}
\end{align}
\end{Theo}

The conditional version of Theorem~\ref{thm_vector_generalization}
is given as follows.

\begin{Theo}
\label{thm_conditional_vector_generalization} Let $\bbu,\bbx$ be
arbitrarily correlated random vectors which are independent of
$\bbn_1,\bbn_2,\bbn_Z$, where $\bbn_1,\bbn_2,\bbn_Z$ are zero-mean
Gaussian random vectors with covariance matrices
$\bzero\prec\bbsigma_1\preceq\bbsigma_2\preceq\bbsigma_Z$,
respectively. Moreover, assume that the second moment of $\bbx$ is
constrained as
\begin{align}
E\left[\bbx \bbx^\top\right]\preceq \bbs
\end{align}
where $\bbs $ is a positive definite matrix. Then, for any
admissible $(\bbu,\bbx)$ pair, there exists a matrix $\bbk^*$ such
that $\bzero \preceq \bbk^* \preceq \bbs$, and
\begin{align}
h(\bbx+\bbn_Z|\bbu)-h(\bbx+\bbn_2|\bbu)&=\frac{1}{2}\log\frac{\left|\bbk^*
+\bbsigma_Z\right|}{\left|\bbk^*+\bbsigma_2\right|}\\
h(\bbx+\bbn_Z|\bbu)-h(\bbx+\bbn_1|\bbu)&\geq\frac{1}{2}\log\frac{\left|\bbk^*
+\bbsigma_Z\right|}{\left|\bbk^*+\bbsigma_1\right|}
\end{align}
\end{Theo}
Theorem~\ref{thm_vector_generalization} serves as a step towards
the proof of Theorem~\ref{thm_conditional_vector_generalization}.
Proofs of these two theorems are deferred to
Sections~\ref{sec:proof_of_thm_vector_generalization}
and~\ref{secr:proof_of_thm_conditional_vec_generalization}.

We are now ready to show that the secrecy capacity region of the
two-user degraded MIMO channel is given by
(\ref{sec_rates_degraded_two_user_1})-(\ref{sec_rates_degraded_two_user_2}).
We first consider $R_2$, and bound it using
Theorem~\ref{degraded_multiuser_wiretap} as follows
\begin{align}
R_2 &\leq I(U_2;\bby_2)-I(U_2;\bbz)\\
&=[I(\bbx;\bby_2)-I(\bbx;\bbz)]-[I(\bbx;\bby_2|U_2)-I(\bbx;\bbz|U_2)]
\label{secrecy_rate_degraded_user_2}
\end{align}
where the equality is obtained by using the chain rule and the
Markov chain $U_2\rightarrow \bbx \rightarrow (\bby_2,\bbz)$. We
now consider the expression in the first bracket of
(\ref{secrecy_rate_degraded_user_2})
\begin{align}
I(\bbx;\bby_2)-I(\bbx;\bbz)&=h(\bby_2)-h(\bby_2|\bbx)-h(\bbz)+h(\bbz|\bbx)\\
&=h(\bby_2)-h(\bbz)-\frac{1}{2}\log
\frac{|\bbsigma_2|}{|\bbsigma_Z|} \label{worst_case_will_imply}
\end{align}
where the second equality follows from the facts that
$h(\bby_2|\bbx)=h(\bbn_2)$ and $h(\bbz|\bbx)=h(\bbn_Z)$. We now
consider the difference of differential entropies in
(\ref{worst_case_will_imply}). To this end, consider the Gaussian
random vector $\tilde{\bbn}_2$ with covariance matrix
$\bbsigma_Z-\bbsigma_2,$ which is chosen to be independent of
$\bbx,\bbn_2$. Using the Markov chain in
(\ref{degraded_Markov_chain}), we get
\begin{align}
h(\bby_2)-h(\bbz)&= h(\bby_2)-h(\bby_2+\tilde{\bbn}_2) \\
&=-I(\tilde{\bbn}_2;\bby_2+\tilde{\bbn}_2) \\
&\leq \max_{\bzero \preceq \bbk \preceq \bbs} \frac{1}{2}
\log\frac{|\bbk+\bbsigma_2|}{|\bbk+\bbsigma_Z|}
\label{worst_additive_implies} \\
&= \frac{1}{2} \log \frac{|\bbs+\bbsigma_2|}{|\bbs+\bbsigma_Z|}
\label{monotonicity_shamai}
\end{align}
where (\ref{worst_additive_implies}) follows from
Lemma~\ref{worst_additive} and (\ref{monotonicity_shamai}) is a
consequence of the fact that
\begin{align}
\frac{|\bbb|}{|\bba+\bbb|}\leq
\frac{|\bbb+\bbdelta|}{|\bba+\bbb+\bbdelta|}
\label{shamao_inequality_without_proof}
\end{align}
when $\bba,\bbb,\bbdelta \succeq 0,$ and $\bba+\bbb\succ
0$~\cite{Shamai_MIMO}. Plugging (\ref{monotonicity_shamai}) into
(\ref{worst_case_will_imply}) yields
\begin{align}
I(\bbx;\bby_2)-I(\bbx;\bbz)\leq \frac{1}{2} \log
\frac{|\bbs+\bbsigma_2|}{|\bbsigma_2|}-\frac{1}{2} \log
\frac{|\bbs+\bbsigma_Z|}{|\bbsigma_Z|}
\label{worst_case_will_imply_1}
\end{align}
We now consider the expression in the second bracket of
(\ref{secrecy_rate_degraded_user_2}). For that purpose, we use
Theorem~\ref{thm_conditional_vector_generalization}. According to
Theorem~\ref{thm_conditional_vector_generalization}, for any
admissible pair $(U_2,\bbx)$, there exists a $\bbk^*$ such that
\begin{align}
h(\bbx+\bbn_Z|U_2)-h(\bbx+\bbn_2|U_2)=\frac{1}{2}\log\frac{|\bbk^*+\bbsigma_Z|}{|\bbk^*+\bbsigma_2|}
\label{thm4_implies_0}
\end{align}
which is equivalent to
\begin{align}
I(\bbx;\bbz|U_2)-I(\bbx;\bby_2|U_2)=\frac{1}{2}\log\frac{|\bbk^*+\bbsigma_Z|}{|\bbsigma_Z|}
-\frac{1}{2}\log\frac{|\bbk^*+\bbsigma_2|}{|\bbsigma_2|}
\label{thm4_implies}
\end{align}
Thus, using (\ref{worst_case_will_imply_1}) and
(\ref{thm4_implies}) in (\ref{secrecy_rate_degraded_user_2}), we
get
\begin{align}
R_2 \leq \frac{1}{2} \log
\frac{|\bbs+\bbsigma_2|}{|\bbk^*+\bbsigma_2
|}-\frac{|\bbs+\bbsigma_Z|}{|\bbk^*+\bbsigma_Z|}
\end{align}
which is the desired bound on $R_2$ given in
(\ref{sec_rates_degraded_two_user_2}). We now obtain the desired
bound on $R_1$ given in (\ref{sec_rates_degraded_two_user_1}). To
this end, we first bound $R_1$ using
Theorem~\ref{degraded_multiuser_wiretap}
\begin{align}
R_1 &\leq I(\bbx;\bby_1|U_2)-I(\bbx;\bbz|U_2)\\
&= h(\bby_1|U_2)-h(\bby_1|U_2,\bbx)-h(\bbz|U_2)+h(\bbz|U_2,\bbx)\\
&=h(\bby_1|U_2)-h(\bbz|U_2)-\frac{1}{2}\log\frac{|\bbsigma_1|}{|\bbsigma_Z|}
\label{secrecy_rate_degraded_user_1}
\end{align}
where the second equality follows from the facts that
$h(\bby_1|U_2,\bbx)=h(\bbn_1)$ and $h(\bbz|U_2,\bbx)=h(\bbn_Z)$.
To bound the difference of conditional differential entropies in
(\ref{secrecy_rate_degraded_user_1}), we use
Theorem~\ref{thm_conditional_vector_generalization}.
Theorem~\ref{thm_conditional_vector_generalization} states that
for any admissible pair $(U_2,\bbx)$, there exists a matrix
$\bbk^*$ such that it satisfies (\ref{thm4_implies_0}) and also
\begin{align}
h(\bbz|U_2)-h(\bby_1|U_2) \geq \frac{1}{2} \log
\frac{|\bbk^*+\bbsigma_Z|}{|\bbk^*+\bbsigma_1|}
\label{thm4_implies_1}
\end{align}
Thus, using (\ref{thm4_implies_1}) in
(\ref{secrecy_rate_degraded_user_1}), we get
\begin{align}
R_1 \leq \frac{1}{2}\log
\frac{|\bbk^*+\bbsigma_1|}{|\bbsigma_1|}-\frac{1}{2}
\log\frac{|\bbk^*+\bbsigma_Z|}{|\bbsigma_Z|}
\end{align}
which is the desired bound on $R_1$ given in
(\ref{sec_rates_degraded_two_user_1}), completing the converse
proof for $K=2$.

As we have seen, the main ingredient in the above proof was
Theorem~\ref{thm_conditional_vector_generalization}. Therefore, to
complete the converse proof for the degraded channel for $K=2$,
from this point on, we will focus on the proof of
Theorem~\ref{thm_conditional_vector_generalization}. We will give
the proof of Theorem~\ref{thm_conditional_vector_generalization}
in Section~\ref{secr:proof_of_thm_conditional_vec_generalization}.
In preparation to that, we will give the proof of
Theorem~\ref{thm_vector_generalization}, which is the
unconditional version of
Theorem~\ref{thm_conditional_vector_generalization}, in
Section~\ref{sec:proof_of_thm_vector_generalization}. The proof of
Theorem~\ref{thm_vector_generalization} involves the use of
properties of the Fisher information, and its connection to the
differential entropy, which are provided next.

\subsection{The Fisher Information Matrix}

We start with the definition~\cite{oliver_johnson}.
\begin{Def}
Let $\bbu$ be a length-$n$ random vector with differentiable
density $f_{U}(\bu)$. The Fisher information matrix of $\bbu$,
$\bbj (\bbu)$, is defined as
\begin{align}
\bbj(\bbu)=E\left[\brho(\bbu)\brho(\bbu)^\top\right]
\end{align}
where $\brho(\bu)$ is the score function which is given by
\begin{align}
\brho(\bu)=\nabla\log f_{U}(\bu)=\left[~\frac{\partial \log
f_{U}(\bu)}{\partial u_1}~~\ldots~~\frac{\partial\log
f_{U}(\bu)}{\partial u_n}~\right]^\top
\end{align}
\end{Def}
Since we are mainly interested in the additive Gaussian channel,
how the Fisher information matrix behaves under the addition of
two independent random vectors is crucial. Regarding this, we have
the following lemma which is due
to~\cite{Liu_Extremal_Inequality}.
\begin{Lem}[\!\!\cite{Liu_Extremal_Inequality}]
\label{Lemma_Matrix_FII} Let $\bbu$ be a random vector with
differentiable density, and let $\bbsigma_U\succ \bzero$ be its
covariance matrix. Moreover, let $\bbv$ be another random vector
with differentiable density, and be independent of $\bbu$. Then,
we have the following facts:
\begin{enumerate}
\item Matrix form of the Cramer-Rao inequality
\begin{align}
\bbj(\bbu)\succeq \bbsigma_U^{-1}
\end{align}
which is satisfied with equality if $\bbu$ is Gaussian. \item For
any square matrix $\bba$,
\begin{align}
\bbj(\bbu+\bbv)\preceq \bba \bbj(\bbu)\bba^\top+(\bbi-\bba)
\bbj(\bbv) (\bbi-\bba)^\top
\end{align}
\end{enumerate}
\end{Lem}

We will use the following consequences of this lemma.
\begin{Cor}
\label{matrix_fisher_information_ineq} Let $\bbu,\bbv$ be as
specified before. Then,
\begin{enumerate}
\item $\bbj(\bbu+\bbv)\preceq \bbj(\bbu)$ \item
$\bbj(\bbu+\bbv)\preceq
\big[\bbj(\bbu)^{-1}+\bbj(\bbv)^{-1}\big]^{-1}$
\end{enumerate}
\end{Cor}
\begin{proof}
The first part of the corollary is obtained by choosing
$\bba=\bbi$, and the second part is obtained by choosing
\begin{align}
\bba
=\big[\bbj(\bbu)^{-1}+\bbj(\bbv)^{-1}\big]^{-1}\bbj(\bbu)^{-1}
\end{align}
and also by noting that $\bbj(\cdot)$ is always a symmetric
matrix.
\end{proof}

The following lemma regarding the Fisher information matrix is
also useful in the proof of
Theorem~\ref{thm_vector_generalization}.

\begin{Lem}
\label{Lemma_change_Fisher} Let $\bbu,\bbv_1,\bbv_2$ be random
vectors such that $\bbu$ and $(\bbv_1,\bbv_2)$ are independent.
Moreover, let $\bbv_1,\bbv_2$ be Gaussian random vectors with
covariance matrices $\bzero \prec \bbsigma_1 \preceq \bbsigma_2$.
Then, we have
\begin{align}
\bbj(\bbu+\bbv_2)^{-1}-\bbsigma_2\succeq
\bbj(\bbu+\bbv_1)^{-1}-\bbsigma_1
\end{align}
\end{Lem}
\begin{proof}
Without loss of generality, let $\bbv_2=\bbv_1+\tilde{\bbv}_1$
such that $\tilde{\bbv}_1$ is a Gaussian random vector with
covariance matrix $\bbsigma_2-\bbsigma_1$, and independent of
$\bbv_1$. Due to the second part of
Corollary~\ref{matrix_fisher_information_ineq}, we have
\begin{align}
\bbj(\bbu+\bbv_2)=\bbj(\bbu+\bbv_1+\tilde{\bbv}_1)&\preceq
\big[\bbj(\bbu+\bbv_1)^{-1}+\bbj(\tilde{\bbv}_1)^{-1}\big]^{-1} \\
&= \big[\bbj(\bbu+\bbv_1)^{-1}+\bbsigma_{2}-\bbsigma_1\big]^{-1}
\end{align}
which is equivalent to
\begin{align}
\bbj(\bbu+\bbv_2)^{-1} \succeq
\bbj(\bbu+\bbv_1)^{-1}+\bbsigma_2-\bbsigma_1
\end{align}
which proves the lemma.
\end{proof}

Moreover, we need the relationship between the Fisher information
matrix and the differential entropy, which is due
to~\cite{Palomar_Gradient}.

\begin{Lem}[\!\!\cite{Palomar_Gradient}]
\label{gradient_fisher} Let $\bbx$ and $\bbn$ be independent
random vectors, where $\bbn$ is zero-mean Gaussian with covariance
matrix $\bbsigma_N\succ\bzero$, and $\bbx$ has a finite second
order moment. Then, we have
\begin{align}
\nabla_{\bbsigma_N} h(\bbx+\bbn)=\frac{1}{2} \bbj(\bbx+\bbn)
\end{align}
\end{Lem}

\subsection{Proof of Theorem~\ref{thm_vector_generalization}}

\label{sec:proof_of_thm_vector_generalization}

To prove Theorem~\ref{thm_vector_generalization}, we first
consider the following expression
\begin{align}
h(\bbx+\bbn_Z)-h(\bbx+\bbn_2) \label{minus_optimization}
\end{align}
which is bounded due to the covariance constraint on $\bbx$. In
particular, we have
\begin{align}
\frac{1}{2}\log \frac{| \bbs+\bbsigma_Z |}{|\bbs+\bbsigma_2|}\leq
h(\bbx+\bbn_Z)-h(\bbx+\bbn_2)\leq \frac{1}{2}\log
\frac{|\bbsigma_Z|}{|\bbsigma_2 |} \label{dummy_bounds}
\end{align}
To see this, define $\tilde{\bbn}$ which is Gaussian with
covariance matrix $\bbsigma_Z-\bbsigma_2$, and is independent of
$\bbn_2$ and $\bbx$. Thus, without loss of generality, we can
assume $\bbz=\bbx+\bbn_2+\tilde{\bbn}$. Then, the left-hand side
of (\ref{dummy_bounds}) can be verified by noting that
\begin{align}
h(\bbx+\bbn_Z)-h(\bbx+\bbn_2)=I(\tilde{\bbn};\bbx+\bbn_2+\tilde{\bbn})
\end{align}
and then by using Lemma~\ref{worst_additive}. The right-hand side
of (\ref{dummy_bounds}) follows from
\begin{align}
h(\bbx+\bbn_Z)-h(\bbx+\bbn_2)&=I(\tilde{\bbn};\bbx+\bbn_Z)\label{rhs_start}\\
&=h(\tilde{\bbn})-h(\tilde{\bbn}|\bbx+\bbn_Z)\\
& \leq h(\tilde{\bbn})-h(\tilde{\bbn}|\bbx+\bbn_Z,\bbx)\label{rhs_cond_entropy}\\
& = h(\tilde{\bbn})-h(\tilde{\bbn}|\bbn_Z)\label{rhs_independence}\\
&=I(\tilde{\bbn};\bbn_Z)\\
&= \frac{1}{2}\log \frac{|\bbsigma_Z|}{|\bbsigma_2
|}\label{rhs_end}
\end{align}
where (\ref{rhs_cond_entropy}) comes from the fact that
conditioning cannot increase entropy, and (\ref{rhs_independence})
is due to the fact that $\bbx$ and $(\bbn_2,\tilde{\bbn})$ are
independent. Thus, we can fix the difference of the differential
entropies in (\ref{dummy_bounds}) to an $\alpha$ in this range,
i.e., we can set
\begin{align}
h(\bbx+\bbn_Z)-h(\bbx+\bbn_2)=\alpha \label{fix_alpha}
\end{align}
where $\alpha\in\left[\frac{1}{2}\log | \bbs+\bbsigma_Z
|/|\bbs+\bbsigma_2|,\frac{1}{2}\log | \bbsigma_Z
|/|\bbsigma_2|\right]$. We now would like to understand how the
constraint in (\ref{fix_alpha}) affects the set of admissible
random vectors. For that purpose, we use
Lemma~\ref{gradient_fisher}, and express this difference of
entropies as an integral of the Fisher information
matrix\footnote{The integration in (\ref{integral_of_fixed}),
i.e., $\int_{\bbsigma_2}^{\bbsigma_Z} \bbj(\cdot) d\bbsigma$, is a
line integral of the vector-valued function $\bbj(\cdot)$.
Moreover, since $\bbj(\cdot)$ is the gradient of a scalar field,
the integration expressed in $\int_{\bbsigma_2}^{\bbsigma_Z}
\bbj(\cdot) d\bbsigma$ is path-free, i.e., it yields the same
value for any path from $\bbsigma_2$ to $\bbsigma_Z$. This remark
applies to all upcoming integrals of $\bbj(\cdot)$.}
\begin{align}
\alpha=h(\bbx+\bbn_Z)-h(\bbx+\bbn_2)=\frac{1}{2}\int_{\bbsigma_2}^{\bbsigma_Z}
\bbj(\bbx+\bbn)d\bbsigma_N \label{integral_of_fixed}
\end{align}
Using the stability of Gaussian random vectors, we can express
$\bbj(\bbx+\bbn)$ as
\begin{align}
\bbj(\bbx+\bbn)=\bbj(\bbx+\bbn_2+\tilde{\bbn})
\label{stability_gaussian}
\end{align}
where $\tilde{\bbn}$ is a zero-mean Gaussian random vector with
covariance matrix $\bbsigma_N-\bbsigma_2 \succeq \bzero $, and is
independent of $\bbn_2$. Using the second part of
Corollary~\ref{matrix_fisher_information_ineq} in
(\ref{stability_gaussian}), we get
\begin{align}
\bbj(\bbx+\bbn)=\bbj(\bbx+\bbn_2+\tilde{\bbn})& \preceq
\big[\bbj(\bbx+\bbn_2)^{-1}+\bbj(\tilde{\bbn})^{-1}\big]^{-1}\\
&= \big[\bbj(\bbx+\bbn_2)^{-1}+\bbsigma_N-\bbsigma_2\big]^{-1}
\label{Matrix_FII_implies_1}
\end{align}
where we used the fact that
$\bbj(\tilde{\bbn})=(\bbsigma_N-\bbsigma_2)^{-1}$ which is a
consequence of the first part of Lemma~\ref{Lemma_Matrix_FII} by
noting that $\tilde{\bbn}$ is Gaussian. We now bound the integral
in (\ref{integral_of_fixed}) by using
(\ref{Matrix_FII_implies_1}). For that purpose, we introduce the
following lemma.
\begin{Lem}
\label{Shamai_s_lemma} Let $\bbk_1,\bbk_2$ be positive
semi-definite matrices satisfying
$\bzero\preceq\bbk_1\preceq\bbk_2$, and $\mathbf{f}(\bbk)$ be a
matrix-valued function such that $\mathbf{f}(\bbk)\succeq\bzero$
for $\bbk_1\preceq\bbk\preceq \bbk_2$. Then, we have
\begin{align}
\int_{\bbk_1}^{\bbk_2}\mathbf{f}(\bbk)d\bbk \geq 0
\end{align}
\end{Lem}
\begin{proof}
The integral is equivalent to
\begin{align}
\int_{\bbk_1}^{\bbk_2}\mathbf{f}(\bbk)d\bbk =\int_{0}^1 \bone^\top
\left[\mathbf{f}\big(\bbk_1+t(\bbk_2-\bbk_1)\big)\odot
(\bbk_2-\bbk_1)\right] \bone dt
\end{align}
where $\odot$ denotes the Schur (Hadamard) product, and
$\bone=[1~\ldots~1]^\top$ with appropriate size. Since the Schur
product of two positive semi-definite matrices is positive
semi-definite~\cite{horn_johnson_book1}, the integrand is
non-negative implying the non-negativity of the integral.
\end{proof}

In light of this lemma, using (\ref{Matrix_FII_implies_1}) in
(\ref{integral_of_fixed}), we get
\begin{align}
\alpha &\leq
\frac{1}{2}\int_{\bbsigma_2}^{\bbsigma_Z}\big[\bbj(\bbx+\bbn_2)^{-1}+\bbsigma_N-\bbsigma_2\big]^{-1}d\bbsigma_N
\\
&=  \frac{1}{2}
\log\frac{|\bbj(\bbx+\bbn_2)^{-1}+\bbsigma_Z-\bbsigma_2|}{|\bbj(\bbx+\bbn_2)^{-1}|}
\label{bound_on_alpha}
\end{align}
where we used the well-known fact that $\nabla_{\bbsigma}\log
|\bbsigma|=\bbsigma^{-\top}$ for $\bbsigma\succ\bzero$. We also
note that the denominator in (\ref{bound_on_alpha}) is strictly
positive because
\begin{align}
\bbj(\bbx+\bbn_2)^{-1}\succeq \bbj(\bbn_2)^{-1}=\bbsigma_2 \succ
\bzero
\end{align}
which implies $|\bbj(\bbx+\bbn_2)^{-1}|>0$.

Following similar steps, we can also find a lower bound on
$\alpha$. Again, using the stability of Gaussian random vectors,
we have
\begin{align}
\bbj(\bbx+\bbn_Z)=\bbj(\bbx+\bbn+\tilde{\bbn})
\label{stability_gaussian_1}
\end{align}
where $\bbn, \tilde{\bbn}$ are zero-mean Gaussian random vectors
with covariance matrices $\bbsigma_N,\bbsigma_Z-\bbsigma_N$,
respectively, $\bbsigma_2\preceq \bbsigma_N\preceq \bbsigma_Z$,
and they are independent. Using the second part of
Corollary~\ref{matrix_fisher_information_ineq}
in~(\ref{stability_gaussian_1}) yields
\begin{align}
\bbj(\bbx+\bbn_Z)=\bbj(\bbx+\bbn+\tilde{\bbn})&\preceq\big[\bbj(\bbx+\bbn)^{-1}+\bbj(\tilde{\bbn})^{-1}\big]^{-1}\\
& =\big[\bbj(\bbx+\bbn)^{-1}+\bbsigma_Z-\bbsigma_N\big]^{-1}
\label{bound_on_Fisher_alt}
\end{align}
where we used the fact that
$\bbj(\tilde{\bbn})=(\bbsigma_Z-\bbsigma_N)^{-1}$ which follows
from the first part of Lemma~\ref{Lemma_Matrix_FII} due to the
Gaussianity of $\tilde{\bbn}$. Then, (\ref{bound_on_Fisher_alt})
is equivalent to
\begin{align}
\bbj(\bbx+\bbn_Z)^{-1}\succeq
\bbj(\bbx+\bbn)^{-1}+\bbsigma_Z-\bbsigma_N
\end{align}
and that implies
\begin{align}
\big[\bbj(\bbx+\bbn_Z)^{-1}+\bbsigma_N-\bbsigma_Z\big]^{-1}\preceq\bbj(\bbx+\bbn)
\label{bound_on_Fisher_alt_1}
\end{align}
Use of Lemma~\ref{Shamai_s_lemma} and
(\ref{bound_on_Fisher_alt_1}) in (\ref{integral_of_fixed}) yields
\begin{align}
\alpha &\geq \int_{\bbsigma_2}^{\bbsigma_Z} \big[\bbj(\bbx+\bbn_Z)^{-1}+\bbsigma_N-\bbsigma_Z\big]^{-1} d\bbsigma_N \\
&=\frac{1}{2} \log
\frac{|\bbj(\bbx+\bbn_Z)^{-1}|}{|\bbj(\bbx+\bbn_Z)^{-1}+\bbsigma_2-\bbsigma_Z|}
\label{bound_on_alpha_1}
\end{align}
where we again used $\nabla_{\bbsigma}\log
|\bbsigma|=\bbsigma^{-\top}$ for $\bbsigma\succ\bzero$. Here also,
the denominator is strictly positive because
\begin{align}
\bbj(\bbx+\bbn_Z)^{-1}+\bbsigma_2-\bbsigma_Z \succeq
\bbj(\bbn_Z)^{-1}+\bbsigma_2-\bbsigma_Z =\bbsigma_2\succ \bzero
\end{align}
which implies $|\bbj(\bbx+\bbn_Z)^{-1}+\bbsigma_2-\bbsigma_Z|>0$.
Combining the two bounds on $\alpha$ given in
(\ref{bound_on_alpha}) and (\ref{bound_on_alpha_1}) yields
\begin{align}
\frac{1}{2}
\log\frac{|\bbj(\bbx+\bbn_Z)^{-1}|}{|\bbj(\bbx+\bbn_Z)^{-1}+\bbsigma_2-\bbsigma_Z|}
\leq \alpha \leq  \frac{1}{2} \log
\frac{|\bbj(\bbx+\bbn_2)^{-1}+\bbsigma_Z-\bbsigma_2|}{|\bbj(\bbx+\bbn_2)^{-1}|}
\label{bound_on_alpha_2}
\end{align}

Next, we will discuss the implications of
(\ref{bound_on_alpha_2}). First, we have a digression of technical
nature to provide the necessary information for such a discussion.
We present the following lemma from~\cite{horn_johnson_book1}.

\begin{Lem}[\!\!\cite{horn_johnson_book1}, Theorem 7.6.4, page 465]
\label{horn_johnson_lemma_1} Let $\bba,\bbb\in M_n$, where $M_n$
is the set of all square matrices of size $n\times n$ over the
complex numbers, be two Hermitian matrices and suppose that there
is a real linear combination of $\bba$ and $\bbb$ that is positive
definite. Then there exists a non-singular matrix $\bbc$ such that
both $\bbc^H \bba \bbc$ and $\bbc^H \bbb \bbc$ are diagonal, where
$(\cdot)^H$ denotes the conjugate transpose.
\end{Lem}

\begin{Lem}
\label{lemma_heart_of_the_contr} Consider the function
\begin{align}
r(t)=\frac{1}{2}\log
\frac{|\bba+\bbb+t\bbdelta|}{|\bba+t\bbdelta|},\qquad 0\leq t\leq1
\end{align}
where $\bba,\bbb,\bbdelta$ are real, symmetric matrices, and $\bba
\succ\bzero$, $\bbb\succeq  \bzero,\bbdelta \succeq \bzero$. The
function $r(t)$ is continuous and monotonically decreasing in $t$.
\end{Lem}
\begin{proof}
We first define the function inside the $\log(\cdot)$ as
\begin{align}
f(t)=\frac{|\bba+\bbb+t\bbdelta|}{|\bba+t\bbdelta|},\qquad 0\leq
t\leq1
\end{align}
We first prove the continuity of $r(t)$. To this end, consider the
function
\begin{align}
g(t)=|\bbe+t\bbdelta |,\qquad 0\leq t\leq 1
\end{align}
where $\bbe\succ\bzero$ is a real, symmetric matrix. By
Lemma~\ref{horn_johnson_lemma_1}, there exists a non-singular
matrix $\bbc$ such that both $\bbc^\top \bbe \bbc $ and $\bbc^\top
\bbdelta \bbc $ are diagonal. Thus, using this fact, we get
\begin{align}
g(t)&=\left|\bbc^{-\top}\bbc^\top\bbe\bbc\bbc^{-1}+t
\bbc^{-\top}\bbc^\top\bbdelta\bbc\bbc^{-1}\right| \\
&= \left|\bbc^{-\top}\right|\left|\bbc^\top\bbe\bbc+t
\bbc^\top\bbdelta\bbc\right|\left|\bbc^{-1}\right| \label{det_product} \\
&= \frac{1}{|\bbc|^2}\left|\bbc^\top\bbe\bbc+t
\bbc^\top\bbdelta\bbc\right| \label{det_inverse}\\
&= \frac{1}{|\bbc|^2}\left|\bbd_E+t
\bbd_{\Delta}\right|\label{diagonal_def}
\end{align}
where (\ref{det_product}) follows from the fact that $|\bba
\bbb|=|\bba||\bbb|$, (\ref{det_inverse}) comes from the fact that
$\left|\bbc^{-\top}\right|=\left|\bbc^{-1}\right|=1/|\bbc|$, and
in (\ref{diagonal_def}), we defined the diagonal matrices
$\bbd_E=\bbc^\top\bbe\bbc$, $\bbd_\Delta=\bbc^\top\bbdelta\bbc$.
Let the diagonal elements of $\bbd_E$ and $\bbd_{\Delta}$ be
$\{d_{E,i}\}_{i=1}^n$ and $\{d_{\Delta,i}\}_{i=1}^n$,
respectively. Then, $g(t)$ can be expressed as
\begin{align}
g(t)=\frac{1}{|\bbc|^2}\prod_{i=1}^n (d_{E,i}+td_{\Delta,i})
\end{align}
which is polynomial in $t$, thus $g(t)$ is continuous in $t$.
Being the ratio of two non-zero continuous functions, $f(t)$ is
continuous as well. Then, continuity of $r(t)$ follows from the
fact that composition of two continuous functions is also
continuous.

We now show the monotonicity of $r(t)$. To this end, consider the
derivative of $r(t)$
\begin{align}
\frac{dr(t)}{dt}=\frac{1}{2f(t)}\frac{df(t)}{dt}
\end{align}
where we have $f(t)>0$ because of the facts that $\bba\succ
\bzero$, $\bbb\succeq \bzero, \bbdelta \succeq \bzero$, and $0\leq
t\leq 1$. Moreover, $f(t)$ is monotonically decreasing in $t$,
which can be deduced from (\ref{shamao_inequality_without_proof}),
implying $df(t)/dt \leq 0$. Thus, we have $dr(t)/dt \leq 0$,
completing the proof.
\end{proof}

After this digression, we are ready to investigate the
implications of (\ref{bound_on_alpha_2}). For that purpose, let us
select $\bba,\bbb,\bbdelta$ in $r(t)$ in
Lemma~\ref{lemma_heart_of_the_contr} as follows
\begin{align}
\bba&=\bbj(\bbx+\bbn_2)^{-1}\\
\bbb&= \bbsigma_Z-\bbsigma_2 \\
\bbdelta &=
\bbj(\bbx+\bbn_Z)^{-1}+\bbsigma_2-\bbsigma_Z-\bbj(\bbx+\bbn_2)^{-1}
\end{align}
where clearly $\bba\succ\bzero $, $\bbb \succeq \bzero$, and also
$\bbdelta \succeq \bzero$ due to Lemma~\ref{Lemma_change_Fisher}.
With these selections, we have
\begin{align}
r(0)&=\frac{1}{2}\log \frac{|\bbj(\bbx+\bbn_2)^{-1}+\bbsigma_Z-\bbsigma_2|}{|\bbj(\bbx+\bbn_2)^{-1}|}\\
r(1)&=\frac{1}{2}\log\frac{|\bbj(\bbx+\bbn_Z)^{-1}|}{|\bbj(\bbx+\bbn_Z)^{-1}+\bbsigma_2-\bbsigma_Z|}
\end{align}
Thus, (\ref{bound_on_alpha_2}) can be expressed as
\begin{align}
r(1)\leq\alpha \leq r(0)
\end{align}
We know from Lemma~\ref{lemma_heart_of_the_contr} that $r(t)$ is
continuous in $t$. Then, from the intermediate value theorem,
there exists a $t^*$ such that $r(t^*)=\alpha$. Thus, we have
\begin{align}
\alpha=r(t^*)&=\frac{1}{2} \log \frac{|\bba+t^*
\bbdelta+\bbsigma_Z-\bbsigma_2|}{|\bba+t^* \bbdelta|}\\
&=\frac{1}{2} \log \frac{|\bbk^*+\bbsigma_Z|}{|\bbk^*+
\bbsigma_2|}
\end{align}
where $\bbk^*=\bba+t^* \bbdelta-\bbsigma_2$. Since $0\leq t^*\leq
1$, $\bbk^*$ satisfies the following orderings,
\begin{align}
 \bbj(\bbx+\bbn_2)^{-1}-\bbsigma_2\preceq \bbk^*\preceq
\bbj(\bbx+\bbn_Z)^{-1}-\bbsigma_Z\label{bound_on_Fisher}
\end{align}
which in turn, by using Lemma~\ref{Lemma_Matrix_FII} and
Corollary~\ref{matrix_fisher_information_ineq}, imply the
following orderings,
\begin{align}
\bbk^{*} & \succeq \bbj(\bbx+\bbn_2)^{-1}-\bbsigma_2 \succeq
\bbj(\bbn_2)^{-1}-\bbsigma_2=\bbsigma_2-\bbsigma_2=\bzero \\
\bbk^\star & \preceq \bbj(\bbx+\bbn_Z)^{-1}-\bbsigma_Z \preceq
{\rm Cov}(\bbx)+\bbsigma_Z-\bbsigma_Z={\rm Cov}(\bbx)\preceq \bbs
\end{align}
which can be summarized as follows,
\begin{align}
\bzero \preceq \bbk^{\star} \preceq \bbs
\end{align}
In addition, using Lemma~\ref{Lemma_change_Fisher} in
(\ref{bound_on_Fisher}), we get
\begin{align}
\bbk^{*} \succeq \bbj(\bbx+\bbn)^{-1}-\bbsigma_N
\label{bound_on_Fisher_1}
\end{align}
for any Gaussian random vector $\bbn$ such that its covariance
matrix satisfies $\bbsigma_N\preceq\bbsigma_2$. The inequality in
(\ref{bound_on_Fisher_1}) is equivalent to
\begin{align}
\bbj(\bbx+\bbn)\succeq \left(\bbk^{*} +\bbsigma_N
\right)^{-1},\quad\textrm{ for } \quad\bbsigma_N \preceq\bbsigma_2
\label{bound_on_Fisher_2}
\end{align}
where $\bbn$ is a Gaussian random vector with covariance matrix
$\bbsigma_N$.

Returning to the proof of Theorem~\ref{thm_vector_generalization},
we now lower bound
\begin{align}
h(\bbx+\bbn_Z)-(\bbx+\bbn_1) \label{lower_bound_1}
\end{align}
while keeping
\begin{align}
h(\bbx+\bbn_Z)-(\bbx+\bbn_2)=\alpha=\frac{1}{2}\log\frac{|\bbk^{*}+\bbsigma_Z|}{|\bbk^{*}+\bbsigma_2|}
\end{align}
The lower bound on (\ref{lower_bound_1}) can be obtained as
follows
\begin{align}
h(\bbx+\bbn_Z)-h(\bbx+\bbn_1)&=\frac{1}{2}
\int_{\bbsigma_1}^{\bbsigma_Z}\bbj(\bbx+\bbn)d\bbsigma_N \label{path_independence_0} \\
&=\frac{1}{2}
\int_{\bbsigma_1}^{\bbsigma_2}\bbj(\bbx+\bbn)d\bbsigma_N
+\frac{1}{2}
\int_{\bbsigma_2}^{\bbsigma_Z}\bbj(\bbx+\bbn)d\bbsigma_N \label{path_independence}\\
&=\frac{1}{2}
\int_{\bbsigma_1}^{\bbsigma_2}\bbj(\bbx+\bbn)d\bbsigma_N
+\frac{1}{2}
\log \frac{|\bbk^{*}+\bbsigma_Z|}{|\bbk^{*}+\bbsigma_2|} \\
&\geq \frac{1}{2} \int_{\bbsigma_1}^{\bbsigma_2}
(\bbk^{*}+\bbsigma_N)^{-1} d\bbsigma_N +\frac{1}{2}
\log\frac{|\bbk^{\star}+\bbsigma_Z|}{|\bbk^{*}+\bbsigma_2|} \label{constraint_implies}\\
&=\frac{1}{2} \log
\frac{|\bbk^{*}+\bbsigma_2|}{|\bbk^{*}+\bbsigma_1|} +\frac{1}{2}
\log \frac{|\bbk^{*}+\bbsigma_Z|}{|\bbk^{*}+\bbsigma_2|}\\
&=\frac{1}{2} \log
\frac{|\bbk^{*}+\bbsigma_Z|}{|\bbk^{\star}+\bbsigma_1|}
\end{align}
where (\ref{path_independence}) follows from the fact that the
integral in (\ref{path_independence_0}) is path-independent, and
(\ref{constraint_implies}) is due to Lemma~\ref{Shamai_s_lemma}
and (\ref{bound_on_Fisher_2}).

Thus, we have shown the following: For any admissible random
vector $\bbx$, we can find a positive semi-definite matrix
$\bbk^*$ such that $\bzero \preceq \bbk^* \preceq \bbs$, and
\begin{align}
h(\bbx+\bbn_Z)-(\bbx+\bbn_2)&=\frac{1}{2}\log
\frac{|\bbk^{*}+\bbsigma_Z|}{|\bbk^{*}+\bbsigma_2|}
\label{optimization_part_1}
\end{align}
and
\begin{align}
h(\bbx+\bbn_Z)-h(\bbx+\bbn_1)&\geq \frac{1}{2} \log
\frac{|\bbk^{*}+\bbsigma_Z|}{|\bbk^{*}+\bbsigma_1|}
\label{optimization_part_2}
\end{align}
which completes the proof of
Theorem~\ref{thm_vector_generalization}.

\subsection{Proof of
Theorem~\ref{thm_conditional_vector_generalization}}

\label{secr:proof_of_thm_conditional_vec_generalization}

We now adopt the proof of Theorem~\ref{thm_vector_generalization}
to the setting of
Theorem~\ref{thm_conditional_vector_generalization} by providing
the conditional versions of the tools we have used in the proof of
Theorem~\ref{thm_vector_generalization}. Main ingredients of the
proof of Theorem~\ref{thm_vector_generalization} are: the
relationship between the differential entropy and the Fisher
information matrix given in Lemma~\ref{gradient_fisher}, and the
properties of the Fisher information matrix given in
Lemmas~\ref{Lemma_Matrix_FII},~\ref{Lemma_change_Fisher} and
Corollary~\ref{matrix_fisher_information_ineq}. Thus, in this
section, we basically provide the extensions of
Lemmas~\ref{Lemma_Matrix_FII},~\ref{Lemma_change_Fisher},~\ref{gradient_fisher}
and Corollary~\ref{matrix_fisher_information_ineq} to the
conditional setting. From another point of view, the material that
we present in this section can be regarded as extending some
well-known results on the Fisher information
matrix~\cite{Liu_Extremal_Inequality,oliver_johnson} to a
conditional setting.

We start with the definition of the conditional Fisher information
matrix.

\begin{Def}
Let $(\bbu,\bbx)$ be an arbitrarily correlated length-$n$ random
vector pair with well-defined densities. The conditional Fisher
information matrix of $\bbx$ given $\bbu$ is defined as
\begin{align}
\bbj(\bbx|\bbu)=E\left[\brho(\bbx|\bbu)\brho(\bbx|\bbu)^\top\right]
\end{align}
where the expectation is over the joint density $f(\bu,\bx)$, and
the conditional score function $\brho(\bx|\bu)$ is
\begin{align}
\brho(\bx|\bu)=\nabla \log f(\bx|\bu)=\left[~\frac{\partial\log
f(\bx|\bu)}{\partial x_1}~~\ldots~~\frac{\partial\log
f(\bx|\bu)}{\partial x_n}~\right]^\top
\end{align}
\end{Def}

The following lemma extends Stein
identity~\cite{Liu_Extremal_Inequality,oliver_johnson} to a
conditional setting. We provide its proof in
Appendix~\ref{proof_of_lemma_conditional_stein_identity}.

\begin{Lem}[Conditional Stein Identity]
\label{lemma_conditional_stein_identity} Let $\bbu,\bbx$ be as
specified above. Consider a \break smooth scalar-valued function
of $\bx$, $g(\bx)$, which well-behaves at infinity in the sense
that
\begin{align}
\lim_{x_i\rightarrow \pm \infty} g(\bx) f(\bx|\bu)=0,\qquad
i=1,\ldots,n \label{assumption_of_the_lemma}
\end{align}
For such a $g(\bx)$, we have
\begin{align}
E\left[g(\bbx)\brho(\bbx|\bbu)\right]=-E\left[\nabla
g(\bbx)\right]
\end{align}
\end{Lem}

The following implications of this lemma are important for the
upcoming proofs.

\begin{Cor}
\label{cor_cond_stein_implies} Let $\bbu,\bbx$ be as specified
above.
\begin{enumerate}
\item $E\left[\brho(\bbx|\bbu)\right]=\bzero$ \item $E\left[\bbx
\brho(\bbx|\bbu)^{\top}\right]=-\bbi$
\end{enumerate}
\end{Cor}
\begin{proof}
The first and the second parts of the corollary follow from the
previous lemma by selecting $g(\bx)=1$ and $g(\bx)=x_i$,
respectively.
\end{proof}

We also need the following variation of this corollary whose proof
is given in Appendix~\ref{proof_of_lemma_cond_stein_implies}.

\begin{Lem}
\label{lemma_cond_stein_implies} Let $\bbu,\bbx$ be as specified
above. Then, we have
\begin{enumerate}
\item $E\left[\brho(\bbx|\bbu)|\bbu\right]=\bzero$.
\item Let
$g(\bu)$ be a finite, scalar-valued function of $\bu$. For such a
$g(\bu)$, we have
\begin{align}
E\left[g(\bbu)\brho(\bbx|\bbu)\right]=\bzero
\end{align}
\item Let $E\left[\bbx|\bbu\right]$ be finite, then we have
\begin{align}
E\left[E\left[\bbx|\bbu\right]\brho(\bbx|\bbu)^\top\right]=\bzero
\end{align}
\end{enumerate}
\end{Lem}

We are now ready to prove the conditional version of the
Cramer-Rao inequality, i.e., the generalization of the first part
of Lemma~\ref{Lemma_Matrix_FII} to a conditional setting.

\begin{Lem}[Conditional Cramer-Rao Inequality]
\label{lemma_conditional_crb_vector} Let $\bbu,\bbx$ be
arbitrarily correlated random vectors with well-defined densities.
Let the conditional covariance matrix of $\bbx$ be ${\rm
Cov}(\bbx|\bbu)\succ \bzero$, then we have
\begin{align}
\bbj(\bbx|\bbu)\succeq {\rm Cov}(\bbx|\bbu)^{-1}
\end{align}
which is satisfied with equality if $(\bbu,\bbx)$ is jointly
Gaussian with conditional covariance matrix ${\rm
Cov}(\bbx|\bbu)$.
\end{Lem}
\begin{proof}
We first prove the inequality
\begin{align}
\bzero & \preceq E\bigg[ \Big(\brho(\bbx|\bbu)+{\rm
Cov}(\bbx|\bbu)^{-1}
\big(\bbx-E\left[\bbx|\bbu\right]\big)\Big)\nonumber\\
&\hspace{2cm}\left. \Big(\brho(\bbx|\bbu)+{\rm
Cov}(\bbx|\bbu)^{-1}
\big(\bbx-E\left[\bbx|\bbu\right]\big)\Big)^\top \right] \\
&= E\left[\brho(\bbx|\bbu)\brho(\bbx|\bbu)^\top \right] +
E\left[\brho(\bbx|\bbu)
\big(\bbx-E\left[\bbx|\bbu\right]\big)^\top \right] {\rm
Cov}(\bbx|\bbu)^{-1}\nonumber\\
&\quad  + {\rm Cov}(\bbx|\bbu)^{-1} E\left[
\big(\bbx-E\left[\bbx|\bbu\right]\big)\brho(\bbx|\bbu)^\top
\right] \nonumber\\
&\quad + {\rm Cov}(\bbx|\bbu)^{-1} E\left[
\big(\bbx-E\left[\bbx|\bbu\right]\big)\big(\bbx-E\left[\bbx|\bbu\right]\big)^\top
\right] {\rm Cov}(\bbx|\bbu)^{-1}\\
&= \bbj(\bbx|\bbu) + E\left[\brho(\bbx|\bbu)
\big(\bbx-E\left[\bbx|\bbu\right]\big)^\top \right] {\rm
Cov}(\bbx|\bbu)^{-1}\nonumber\\
&\quad  + {\rm Cov}(\bbx|\bbu)^{-1} E\left[
\big(\bbx-E\left[\bbx|\bbu\right]\big)\brho(\bbx|\bbu)^\top
\right] +{\rm Cov}(\bbx|\bbu)^{-1} \label{cond_crb_proof_1}
\end{align}
where for the second equality, we used the definition of the
conditional Fisher information matrix, and the conditional
covariance matrix. We note that
\begin{align}
\left(E\left[
\big(\bbx-E\left[\bbx|\bbu\right]\big)\brho(\bbx|\bbu)^\top
\right]\right)^\top&= E\left[\brho(\bbx|\bbu)
\big(\bbx-E\left[\bbx|\bbu\right]\big)^\top \right] \\
&= E\left[\brho(\bbx|\bbu) \bbx^\top
\right]-E\left[\brho(\bbx|\bbu)
E\left[\bbx|\bbu\right]^\top \right] \\
&= E\left[\brho(\bbx|\bbu) \bbx^\top \right] \label{cor_stein_implies_1}\\
&=-\bbi \label{cor_stein_implies_2}
\end{align}
where (\ref{cor_stein_implies_1}) is due to the third part of
Lemma~\ref{lemma_cond_stein_implies}, and
(\ref{cor_stein_implies_2}) is a result of the second part of
Corollary~\ref{cor_cond_stein_implies}. Using
(\ref{cor_stein_implies_2}) in (\ref{cond_crb_proof_1}) gives
\begin{align}
\bzero & \preceq \bbj(\bbx|\bbu) - {\rm Cov}(\bbx|\bbu)^{-1} -{\rm
Cov}(\bbx|\bbu)^{-1}+{\rm Cov}(\bbx|\bbu)^{-1}
\end{align}
which concludes the proof.

For the equality case, consider the conditional Gaussian
distribution
\begin{align}
f(\bx|\bu)=C \exp
\left(-\frac{1}{2}\big(\bx-E\left[\bbx|\bbu=\bu\right]\big)^\top{\rm
Cov}(\bbx|\bbu)^{-1}\big(\bx-E\left[\bbx|\bbu=\bu\right]\big)\right)
\end{align}
where $C$ is the normalizing factor. The conditional score
function is
\begin{align}
\brho(\bx|\bu)=-{\rm
Cov}(\bbx|\bbu)^{-1}\big(\bx-E\left[\bbx|\bbu=\bu\right]\big)
\end{align}
which implies $\bbj(\bbx|\bbu)={\rm Cov}(\bbx|\bbu)^{-1}$.
\end{proof}

We now present the conditional convolution identity which is
crucial to extend the second part of Lemma~\ref{Lemma_Matrix_FII}
to a conditional setting.

\begin{Lem}[Conditional Convolution Identity]
\label{lemma_cond_conv_identity} Let $\bbx,\bby,\bbu$ be
length-$n$ random vectors and let the density for any combination
of these random vectors exist. Moreover, let $\bbx$ and $\bby$ be
conditionally independent given $\bbu$, and let $\bbw$ be defined
as $\bbw=\bbx+\bby$. Then, we have
\begin{align}
\brho
(\bw|\bu)=E\left[\brho(\bbx|\bbu=\bu)|\bbw=\bw,\bbu=\bu\right]
=E\left[\brho(\bby|\bbu=\bu)|\bbw=\bw,\bbu=\bu\right]
\end{align}
\end{Lem}

The proof of this lemma is given in
Appendix~\ref{proof_of_lemma_cond_conv_identity}. We will use this
lemma to prove the conditional Fisher information matrix
inequality, i.e., the generalization of the second part of
Lemma~\ref{Lemma_Matrix_FII}.

\begin{Lem}[Conditional Fisher Information Matrix Inequality]
\label{lemma_cond_matrix_fii} Let $\bbx,\bby,\bbu$ be as specified
in the previous lemma. For any square matrix $\bba$, we have
\begin{align}
\bbj(\bbx+\bby|\bbu)\preceq \bba
\bbj(\bbx|\bbu)\bba^\top+\left(\bbi-\bba\right)
\bbj(\bby|\bbu)\left(\bbi-\bba\right)^\top
\end{align}
\end{Lem}

The proof of this lemma is given in
Appendix~\ref{proof_of_lemma_cond_matrix_fii}. The following
implications of Lemma~\ref{lemma_cond_matrix_fii} correspond to
the conditional version of
Corollary~\ref{matrix_fisher_information_ineq}.

\begin{Cor}
Let $\bbx,\bby,\bbu$ be as specified in the previous lemma. Then,
we have
\begin{enumerate}
\item $\bbj(\bbx+\bby|\bbu)\preceq \bbj(\bbx|\bbu)$ \item
$\bbj(\bbx+\bby|\bbu)\preceq
\left[\bbj(\bbx|\bbu)^{-1}+\bbj(\bby|\bbu)^{-1}\right]^{-1}$
\end{enumerate}
\end{Cor}
\begin{proof}
The first part of the corollary can be obtained by setting
$\bba=\bbi$ in the previous lemma. For the second part, the
selection
$\bba=\left[\bbj(\bbx|\bbu)^{-1}+\bbj(\bby|\bbu)^{-1}\right]^{-1}\bbj(\bbx|\bbu)^{-1}$
yields the desired result.
\end{proof}

Using this corollary, one can prove the conditional version of
Lemma~\ref{Lemma_change_Fisher} as well, which is omitted. So far,
we have proved the conditional versions of the inequalities
related to the Fisher information matrix, that were used in the
proof of Theorem~\ref{thm_vector_generalization}. To claim that
the proof of Theorem~\ref{thm_vector_generalization} can be
adapted for Theorem~\ref{thm_conditional_vector_generalization},
we only need the conditional version of
Lemma~\ref{gradient_fisher}. In~\cite{Palomar_Gradient}, the
following result is implicity present.
\begin{Lem}
\label{gradient_fisher_conditional} Let $(\bbu,\bbx)$ be an
arbitrarily correlated random vector pair with finite second order
moments, and be independent of the random vector $\bbn$ which is
zero-mean Gaussian with covariance matrix $\bbsigma_N\succ\bzero$.
Then, we have
\begin{align}
\nabla_{\bbsigma_N} h(\bbx+\bbn|\bbu)=\frac{1}{2}
\bbj(\bbx+\bbn|\bbu)
\end{align}
\end{Lem}
\begin{proof}
Let $F_{U}(\bu)$ be the cumulative distribution function of
$\bbu$, and $f(\bx +\bn |\bbu =\bu)$ be the conditional density of
$\bbx +\bbn $ which is guaranteed to exist because $\bbn $ is
Gaussian. We have
\begin{align}
\nabla_{\bbsigma_N} h(\bbx+\bbn|\bbu)&= \nabla_{\bbsigma_N} \int h(\bbx+\bbn|\bbu=\bu) dF_{U}(\bu)\\
&= \int \nabla_{\bbsigma_N} h(\bbx+\bbn|\bbu=\bu) dF_{U}(\bu)\label{interchange}\\
&= \frac{1}{2}\int E\left[ \nabla \log f(\bbx+\bbn|\bbu=\bu)
\nabla \log f(\bbx+\bbn|\bbu=\bu) ^\top \right]dF_{U}(\bu)
\label{gradient_fisher_implies}\\
&= \frac{1}{2}E\left[ \nabla \log f(\bbx+\bbn|\bbu) \nabla \log f(\bbx+\bbn|\bbu) ^\top \right]\\
&= \frac{1}{2} \bbj(\bbx+\bbn|\bbu)
\label{def_of_cond_Fisher_imply_4}
\end{align}
where in (\ref{interchange}), we changed the order of integration
and differentiation, which can be done due to the finiteness of
the conditional differential entropy, which in turn is ensured by
the finite second-order moments of $(\bbu,\bbx)$,
(\ref{gradient_fisher_implies}) is a consequence of
Lemma~\ref{gradient_fisher}, and
(\ref{def_of_cond_Fisher_imply_4}) follows from the definition of
the conditional Fisher information matrix.
\end{proof}

Since we have derived all necessary tools, namely conditional
counterparts of
Lemmas~\ref{Lemma_Matrix_FII},~\ref{Lemma_change_Fisher},~\ref{gradient_fisher}
and Corollary~\ref{matrix_fisher_information_ineq}, the proof
Theorem~\ref{thm_vector_generalization} can be adapted to prove
Theorem~\ref{thm_conditional_vector_generalization}.

\subsection{Proof of Theorem~\ref{main_result} for Arbitrary $K$}
\label{sec:proof_main_result_arbitrary_K}

We now prove Theorem~\ref{main_result} for arbitrary $K$. To this
end, we will mainly use the intuition gained in the proof of
Theorem~\ref{thm_vector_generalization} and the tools developed in
the previous section. The only new ingredient that is needed is
the following lemma whose proof is given in
Appendix~\ref{proof_of_lemma_conditioning_increases_Fisher}.

\begin{Lem}
\label{lemma_conditioning_increases_Fisher} Let $(\bbv,\bbu,\bbx)$
be length-$n$ random vectors with well-defined densities.
Moreover, assume that the partial derivatives of $f(\bu|\bv,\bx)$
with respect to $x_i,~i=1,\ldots,n,$ exist and satisfy
\begin{align}
\max_{1\leq i\leq n}\left|\frac{\partial f(\bu |\bx,\bv)}{\partial
x_i}\right|\leq g(\bu)
\label{assumption_of_the_lemma_cond_inc_Fisher}
\end{align}
for some integrable function $g(\bu)$. Then, if $(\bbv,\bbu,\bbx)$
satisfy the Markov chain $\bbv\rightarrow \bbu \rightarrow \bbx $,
we have
\begin{align}
\bbj(\bbx |\bbu)\succeq \bbj(\bbx |\bbv)
\end{align}
\end{Lem}

We now start the proof of Theorem~\ref{main_result} for arbitrary
$K$. First, we rewrite the bound given in
Theorem~\ref{degraded_multiuser_wiretap} for the $K$th user's
secrecy rate as follows
\begin{align}
I(U_K;\bby_K)-I(U_K;\bbz)&=I(\bbx;\bby_K)-I(\bbx;\bbz)-\left[I(\bbx;\bby_K|U_K)-I(\bbx;\bbz|U_K)\right] \label{Markov_chain_implies_1} \\
&\leq \frac{1}{2}\log
\frac{|\bbs+\bbsigma_K|}{|\bbsigma_K|}-\frac{1}{2}\log
\frac{|\bbs+\bbsigma_Z|}{|\bbsigma_Z|}-\left[I(\bbx;\bby_K|U_K)-I(\bbx;\bbz|U_K)\right]
\label{worst_additive_implies_1}
\end{align}
where in (\ref{Markov_chain_implies_1}), we used the Markov chain
$U_K\rightarrow \bbx\rightarrow (\bby_K,\bbz )$, and obtained
(\ref{worst_additive_implies_1}) using the worst additive noise
lemma given in Lemma~\ref{worst_additive}. Moreover, using the
Markov chain $U_K\rightarrow \bbx \rightarrow \bby_K\rightarrow
\bbz$, the other difference term in
(\ref{worst_additive_implies_1}) can be bounded as follows.
\begin{align}
0\leq I(\bbx;\bby_K|U_K)-I(\bbx;\bbz|U_K) &\leq I(\bbx;\bby_K)-I(\bbx;\bbz) \\
&\leq \frac{1}{2}\log
\frac{|\bbs+\bbsigma_K|}{|\bbsigma_K|}-\frac{1}{2}\log
\frac{|\bbs+\bbsigma_Z|}{|\bbsigma_Z|} \label{dummy_range}
\end{align}
The proofs of Theorems~\ref{thm_vector_generalization}
and~\ref{thm_conditional_vector_generalization} reveal that for
any value of $I(\bbx;\bby_K|U_K)-I(\bbx;\bbz|U_K)$ in the range
given in (\ref{dummy_range}), there exists positive semi-definite
matrix $\tilde{\bbk}_K$ such that
\begin{align}
\bbj(\bbx+\bbn_K|U_K)^{-1}-\bbsigma_K \preceq \tilde{\bbk}_K
\preceq \bbs \label{finding_ordered_matrices_1}
\end{align}
and
\begin{align}
I(\bbx;\bby_K|U_K)-I(\bbx;\bbz|U_K) &=\frac{1}{2}\log
\frac{|\tilde{\bbk}_K+\bbsigma_K|}{|\bbsigma_K|}-\frac{1}{2}\log
\frac{|\tilde{\bbk}_K+\bbsigma_Z|}{|\bbsigma_Z|}
\label{equality_Kth_user}
\\
I(\bbx;\bby_{K-1}|U_K)-I(\bbx;\bbz|U_K) &\leq
\frac{1}{2}\log\frac{|\tilde{\bbk}_K+\bbsigma_{K-1}|}{|\bbsigma_{K-1}|}-\frac{1}{2}\log
\frac{|\tilde{\bbk}_K+\bbsigma_Z|}{|\bbsigma_Z|}
\label{fixing_K_th_user_implies}
\end{align}
Using (\ref{equality_Kth_user}) in
(\ref{worst_additive_implies_1}) yields the desired bound on the
$K$th user's secrecy rate as follows
\begin{align}
R_K&\leq \frac{1}{2}
\log\frac{|\bbs+\bbsigma_{K}|}{|\tilde{\bbk}_K+\bbsigma_{K}|}-
\frac{1}{2}
\log\frac{|\bbs+\bbsigma_{Z}|}{|\tilde{\bbk}_K+\bbsigma_{Z}|}
\end{align}

We now bound the $(K-1)$th user's secrecy rate. To this end, first
note that
\begin{align}
R_{K-1}&\leq I(U_{K-1};\bby_{K-1}|U_K)-I(U_{K-1};\bbz|U_K)\\
&=I(\bbx;\bby_{K-1}|U_K)-I(\bbx;\bbz|U_K)-\left[I(\bbx;\bby_{K-1}|U_{K-1})-I(\bbx;\bbz|U_{K-1})\right]\label{Markov_chain_implies_2}\\
&\leq \frac{1}{2}\log
\frac{|\tilde{\bbk}_K+\bbsigma_{K-1}|}{|\bbsigma_{K-1}|}
-\frac{1}{2}\log \frac{|\tilde{\bbk}_K+\bbsigma_Z|}{|\bbsigma_Z|}
- \left[I(\bbx;\bby_{K-1}|U_{K-1})-I(\bbx;\bbz|U_{K-1})\right]
\label{fixing_K_th_user_implies_1}
\end{align}
where in order to obtain (\ref{Markov_chain_implies_2}), we used
the Markov chain $U_K\rightarrow U_{K-1}\rightarrow \bbx
\rightarrow (\bby_{K-1},\bbz)$, and
(\ref{fixing_K_th_user_implies_1}) comes from
(\ref{fixing_K_th_user_implies}). Using the Markov chain
$U_K\rightarrow U_{K-1}\rightarrow \bbx \rightarrow
\bby_{K-1}\rightarrow \bbz$, the mutual information difference in
(\ref{fixing_K_th_user_implies_1}) is bounded as
\begin{align}
0 \leq I(\bbx;\bby_{K-1}|U_{K-1})-I(\bbx;\bbz|U_{K-1})& \leq I(\bbx;\bby_{K-1}|U_K)-I(\bbx;\bbz|U_K) \\
&\leq \frac{1}{2}\log
\frac{|\tilde{\bbk}_K+\bbsigma_{K-1}|}{|\bbsigma_{K-1}|}
-\frac{1}{2}\log \frac{|\tilde{\bbk}_K+\bbsigma_Z|}{|\bbsigma_Z|}
\label{towards_fixing_(K-1)_1}
\end{align}
Using the analysis carried out in the proof of
Theorem~\ref{thm_vector_generalization}, we can get a more refined
lower bound as follows
\begin{align}
I(\bbx;\bby_{K-1}|U_{K-1})-I(\bbx;\bbz|U_{K-1})&\geq \frac{1}{2}
\log \frac{|\bbj(\bbx+\bbn_{K-1}|U_{K-1})^{-1}|}{|\bbsigma_{K-1}|}\nonumber\\
&\quad -\frac{1}{2} \log\
\frac{|\bbj(\bbx+\bbn_{K-1}|U_{K-1})^{-1}+\bbsigma_Z-\bbsigma_{K-1}|}{|\bbsigma_{Z}|}
\label{towards_fixing_(K-1)_2}
\end{align}
Combining (\ref{towards_fixing_(K-1)_1}) and
(\ref{towards_fixing_(K-1)_2}) yields
\begin{align}
\lefteqn{\frac{1}{2} \log \frac{|\bbj(\bbx+\bbn_{K-1}|U_{K-1})^{-1}|}{|\bbj(\bbx+\bbn_{K-1}|U_{K-1})^{-1}+\bbsigma_Z-\bbsigma_{K-1}|}}&\nonumber\\
&\hspace{2cm}\leq
I(\bbx;\bby_{K-1}|U_{K-1})-I(\bbx;\bbz|U_{K-1})+\frac{1}{2}\log \frac{|\bbsigma_{K-1}|}{|\bbsigma_Z|}\nonumber\\
&\hspace{2cm}\leq \frac{1}{2}\log
\frac{|\tilde{\bbk}_K+\bbsigma_{K-1}|}{|\tilde{\bbk}_K+\bbsigma_Z|}
\label{towards_fixing_(K-1)_3}
\end{align}
Now, using the lower bound on $\tilde{\bbk}_{K}$ given in
(\ref{finding_ordered_matrices_1}), we get
\begin{align}
\tilde{\bbk}_{K}&\succeq \bbj(\bbx+\bbn_K|U_K)^{-1}-\bbsigma_K \\
& \succeq \bbj(\bbx+\bbn_{K-1}|U_K)^{-1}-\bbsigma_{K-1}
\label{towards_fixing_(K-1)_4}
\end{align}
where (\ref{towards_fixing_(K-1)_4}) is obtained using
Lemma~\ref{Lemma_change_Fisher}. Moreover, since we have
$U_{K}\rightarrow U_{K-1}\rightarrow \bbx+\bbn_{K-1}$, the
following order exists
\begin{align}
\bbj(\bbx+\bbn_{K-1}|U_{K-1})\succeq \bbj(\bbx+\bbn_{K-1}|U_K)
\label{conditioning_decreases_Fisher_step5}
\end{align}
due to Lemma~\ref{lemma_conditioning_increases_Fisher}. Equation
(\ref{conditioning_decreases_Fisher_step5}) is equivalent to
\begin{align}
\bbj(\bbx+\bbn_{K-1}|U_{K-1})^{-1}\preceq
\bbj(\bbx+\bbn_{K-1}|U_K)^{-1}
\end{align}
using which in (\ref{towards_fixing_(K-1)_4}), we get
\begin{align}
\tilde{\bbk}_{K}&\succeq
\bbj(\bbx+\bbn_{K-1}|U_{K-1})^{-1}-\bbsigma_{K-1}
\label{towards_fixing_(K-1)_5}
\end{align}
We now consider the function
\begin{align}
r(t)=\frac{1}{2} \log
\frac{|\bba+\bbb+t\bbdelta|}{|\bba+t\bbdelta|},\quad 0\leq t\leq 1
\end{align}
with the following parameters
\begin{align}
\bba&=\bbj(\bbx+\bbn_{K-1}|U_{K-1})^{-1}\\
\bbb&=\bbsigma_Z-\bbsigma_{K-1}\\
\bbdelta&=
\tilde{\bbk}_{K}+\bbsigma_{K-1}-\bbj(\bbx+\bbn_{K-1}|U_{K-1})^{-1}
\end{align}
where $\bbdelta\succeq \bzero$ due to
(\ref{towards_fixing_(K-1)_5}). Using this function, we can
paraphrase the bound in (\ref{towards_fixing_(K-1)_3}) as
\begin{align}
-r(0)\leq
I(\bbx;\bby_{K-1}|U_{K-1})-I(\bbx;\bbz|U_{K-1})+\frac{1}{2}\log
\frac{|\bbsigma_{K-1}|}{|\bbsigma_Z|} \leq -r(1)
\end{align}
As shown in Lemma~\ref{lemma_heart_of_the_contr}, $r(t)$ is
continuous and monotonically decreasing in $t$. Thus, there exists
a $t^*$ such that
\begin{align}
-r(t^*)=I(\bbx;\bby_{K-1}|U_{K-1})-I(\bbx;\bbz|U_{K-1})+\frac{1}{2}\log
\frac{|\bbsigma_{K-1}|}{|\bbsigma_Z|}
\end{align}
due to the intermediate value theorem. Let
$\tilde{\bbk}_{K-1}=\bba+t^* \bbdelta-\bbsigma_{K-1}$, then we get
\begin{align}
I(\bbx;\bby_{K-1}|U_{K-1})-I(\bbx;\bbz|U_{K-1})=\frac{1}{2}\log
\frac{|\tilde{\bbk}_{K-1}+\bbsigma_{K-1}|}{|\bbsigma_{K-1}|}
-\frac{1}{2}\log
\frac{|\tilde{\bbk}_{K-1}+\bbsigma_Z|}{|\bbsigma_Z|}
\label{towards_fixing_(K-1)_6}
\end{align}
We note that using (\ref{towards_fixing_(K-1)_6}) in
(\ref{fixing_K_th_user_implies_1}) yields the desired bound on the
$(K-1)$th user's secrecy rate as follows
\begin{align}
R_{K-1}&\leq
 \frac{1}{2}\log \frac{|\tilde{\bbk}_K+\bbsigma_{K-1}|}{|\tilde{\bbk}_{K-1}+\bbsigma_{K-1}|}
-\frac{1}{2}\log
\frac{|\tilde{\bbk}_K+\bbsigma_Z|}{|\tilde{\bbk}_{K-1}+\bbsigma_Z|}
\end{align}
Moreover, since $\bbdelta\succeq \bzero$ and $0\leq t\leq 1$,
$\tilde{\bbk}_{K-1}=\bba+t^* \bbdelta-\bbsigma_{K-1}$ satisfies
the following orderings
\begin{align}
\bbj(\bbx+\bbn_{K-1}|U_{K-1})^{-1}-\bbsigma_{K-1}\preceq
\tilde{\bbk}_{K-1}\preceq \tilde{\bbk}_{K}
\label{towards_fixing_(K-1)_7}
\end{align}
Furthermore, the lower bound in (\ref{towards_fixing_(K-1)_7})
implies the following order
\begin{align}
\tilde{\bbk}_{K-1} \succeq
\bbj(\bbx+\bbn|U_{K-1})^{-1}-\bbsigma_{N} \label{dummy_ordering}
\end{align}
for any Gaussian random vector $\bbn$ such that $\bbsigma_N\preceq
\bbsigma_{K-1}$, and is independent of $U_{K-1},\bbx$, which is a
consequence of Lemma~\ref{Lemma_change_Fisher}. Using
(\ref{dummy_ordering}), and following the proof of
Theorem~\ref{thm_vector_generalization}, we can show that
\begin{align}
I(\bbx;\bby_{K-2}|U_{K-1})-I(\bbx;\bbz|U_{K-1})\leq
\frac{1}{2}\log\frac{|\tilde{\bbk}_{K-1}+\bbsigma_{K-2}|}{|\bbsigma_{K-2}|}
-\frac{1}{2}\log\frac{|\tilde{\bbk}_{K-1}+\bbsigma_Z|}{|\bbsigma_Z|}
\end{align}

Thus, as a recap, we have showed that there exists
$\tilde{\bbk}_{K-1}$ such that
\begin{align}
\bbj(\bbx+\bbn_{K-1}|U_{K-1})^{-1}-\bbsigma_{K-1}\preceq
\tilde{\bbk}_{K-1}\preceq \tilde{\bbk}_{K}
\end{align}
and
\begin{align}
I(\bbx;\bby_{K-1}|U_{K-1})-I(\bbx;\bbz|U_{K-1}) &=\frac{1}{2}\log
\frac{|\tilde{\bbk}_{K-1}+\bbsigma_{K-1}|}{|\bbsigma_{K-1}|}
-\frac{1}{2}\log \frac{|\tilde{\bbk}_{K-1}+\bbsigma_Z|}{|\bbsigma_Z|}\\
I(\bbx;\bby_{K-2}|U_{K-1})-I(\bbx;\bbz|U_{K-1})&\leq
\frac{1}{2}\log
\frac{|\tilde{\bbk}_{K-1}+\bbsigma_{K-2}|}{|\bbsigma_{K-2}|}
-\frac{1}{2}\log
\frac{|\tilde{\bbk}_{K-1}+\bbsigma_Z|}{|\bbsigma_Z|}
\end{align}
which are analogous to (\ref{finding_ordered_matrices_1}),
(\ref{equality_Kth_user}), (\ref{fixing_K_th_user_implies}). Thus,
proceeding in the same manner, for any selection of the joint
distribution $p(u_{K})p(u_{K-1}|u_{K})\ldots p(\bx|u_2)$, we can
show the existence of matrices
$\big\{\tilde{\bbk}_{k}\big\}_{k=1}^{K+1}$ such that
\begin{align}
\bzero=\tilde{\bbk}_1\preceq \tilde{\bbk}_2\preceq \quad \ldots
\quad\preceq \tilde{\bbk}_K \preceq \tilde{\bbk}_{K+1}=\bbs
\label{psd_matrices_found_1}
\end{align}
and
\begin{align}
I(\bbx;\bby_{k}|U_k)-I(\bbx;\bbz|U_k)&=\frac{1}{2}
\log\frac{|\tilde{\bbk}_k+\bbsigma_k|}{|\bbsigma_k|}
-\frac{1}{2} \log\frac{|\tilde{\bbk}_k+\bbsigma_Z|}{|\bbsigma_Z|},\quad k=2,\ldots,K \label{psd_matrices_found_2} \\
I(\bbx;\bby_{k-1}|U_k)-I(\bbx;\bbz|U_k)&\leq \frac{1}{2}
\log\frac{|\tilde{\bbk}_k+\bbsigma_{k-1}|}{|\bbsigma_{k-1}|}
-\frac{1}{2} \log
\frac{|\tilde{\bbk}_k+\bbsigma_Z|}{|\bbsigma_Z|},\quad
k=2,\ldots,K+1 \label{psd_matrices_found_3}
\end{align}
where $U_{K+1}=\phi$. We now define
$\bbk_k=\tilde{\bbk}_{k+1}-\tilde{\bbk}_k,~k=1,\ldots,K,$ which
yields $\tilde{\bbk}_{k+1}=\sum_{i=1}^k \bbk_i$, and in
particular, $\bbs=\sum_{i=1}^K \bbk_i$. Using these new variables
in conjunction with (\ref{psd_matrices_found_2}) and
(\ref{psd_matrices_found_3}) results in
\begin{align}
R_k&\leq I(U_k;\bby_k|U_{k+1})-I(U_k;\bbz|U_{k+1})\\
&=
I(\bbx;\bby_k|U_{k+1})-I(\bbx;\bbz|U_{k+1})-\left[I(\bbx;\bby_k|U_{k})
-I(\bbx;\bbz|U_{k})\right]\\
&\leq \frac{1}{2}
\log\frac{\big|\tilde{\bbk}_{k+1}+\bbsigma_{k}\big|}{\left|\bbsigma_{k}\right|}
-\frac{1}{2} \log \frac{\big|\tilde{\bbk}_{k+1}+\bbsigma_Z\big|}{|\bbsigma_Z|}\nonumber\\
&\quad -\frac{1}{2}
\log\frac{\big|\tilde{\bbk}_k+\bbsigma_{k}\big|}{|\bbsigma_{k}|}
+\frac{1}{2} \log\frac{\big|\tilde{\bbk}_k+\bbsigma_Z\big|}{|\bbsigma_Z|} \\
&=\frac{1}{2}
\log\frac{\big|\tilde{\bbk}_{k+1}+\bbsigma_{k}\big|}{\big|\tilde{\bbk}_{k}+\bbsigma_{k}\big|}
-\frac{1}{2} \log \frac{\big|\tilde{\bbk}_{k+1}+\bbsigma_Z\big|}{\big|\tilde{\bbk}_{k}+\bbsigma_Z\big|}\\
&=\frac{1}{2}
\log\frac{\big|\sum_{i=1}^k\bbk_{i}+\bbsigma_{k}\big|}{\big|\sum_{i=1}^{k-1}\bbk_{i}+\bbsigma_{k}\big|}
-\frac{1}{2} \log
\frac{\big|\sum_{i=1}^k\bbk_{i}+\bbsigma_Z\big|}{\big|\sum_{i=1}^{k-1}\bbk_{i}+\bbsigma_Z\big|}
\label{bounds_found_1}
\end{align}
for $k=2,\ldots,K$. For $k=1$, the bound in
(\ref{psd_matrices_found_3}), by setting $k=2$ in the
corresponding expression, yields the desired bound on the first
user's secrecy rate
\begin{align}
R_1&\leq I(\bbx;\bby_1|U_2)-I(\bbx;\bbz|U_2)\\
&\leq \frac{1}{2}
\log\frac{|\bbk_{1}+\bbsigma_{1}|}{|\bbsigma_{1}|} -\frac{1}{2}
\log\frac{|\bbk_{1}+\bbsigma_Z|}{|\bbsigma_Z|}
\label{bounds_found_2}
\end{align}
Since for any selection of the joint distribution
$p(u_K)p(u_{K-1}|u_K)\ldots p(\bx|u_2)$, we can establish the
bounds in (\ref{bounds_found_1}) and (\ref{bounds_found_2}) with
positive semi-definite matrices $\left\{\bbk_i\right\}_{i=1}^K$
such that $\bbs=\sum_{i=1}^K\bbk_{i}$, the union of these bounds
over such matrices would be an outer bound for the secrecy
capacity region, completing the converse proof of
Theorem~\ref{main_result} for an arbitrary $K$.

\section{Aligned Gaussian MIMO Multi-receiver Wiretap \newline Channel}
\label{proof_of_theorem_main_result_aligned}

We now consider the aligned Gaussian MIMO multi-receiver wiretap
channel, and prove its secrecy capacity region. To that end, we
basically use our capacity result for the degraded Gaussian MIMO
multi-receiver wiretap channel in
Section~\ref{sec:proof_of_main_result} in conjunction with the
channel enhancement technique~\cite{Shamai_MIMO}. Due to the
presence of an eavesdropper in our channel model, there are some
differences between the way we invoke the channel enhancement
technique and the way it was used in its original version that
appeared in~\cite{Shamai_MIMO}. These differences will be pointed
out during our proof.

Given the covariance matrices $\left\{\bbk_i\right\}_{i=1}^K$ such
that $\sum_{i=1}^K \bbk_i \preceq \bbs$, let us define the
following rates,
\begin{align}
R_k^{\rm
DPC}\left(\pi,\left\{\bbk_i\right\}_{i=1}^K,\left\{\bbsigma_i\right\}_{i=1}^K,\bbsigma_Z\right)&=
\frac{1}{2}\log\frac{\left|\sum_{i=1}^k
\bbk_{\pi(i)}+\bbsigma_{\pi(k)}\right|}{\left|\sum_{i=1}^{k-1}
\bbk_{\pi(i)}+\bbsigma_{\pi(k)}\right|}\nonumber\\
&\quad - \frac{1}{2}\log\frac{\left|\sum_{i=1}^k
\bbk_{\pi(i)}+\bbsigma_Z\right|}{\left|\sum_{i=1}^{k-1}
\bbk_{\pi(i)}+\bbsigma_Z\right|},\quad k=1,\ldots,K
\end{align}
where $\pi(\cdot)$ is a one-to-one permutation on
$\{1,\ldots,K\}$. We also note that the subscript of $R_k^{\rm
DPC}\left(\pi,\left\{\bbk_i\right\}_{i=1}^K,\left\{\bbsigma_i\right\}_{i=1}^K,\bbsigma_Z\right)$
does not denote the $k$th user, instead it denotes the $(K-k+1)$th
user in line to be encoded. Rather, the secrecy rate of the $k$th
user is given by
\begin{align}
R_k=R_{\pi^{-1}(k)}^{\rm
DPC}\left(\pi,\left\{\bbk_i\right\}_{i=1}^K,\left\{\bbsigma_i\right\}_{i=1}^K,\bbsigma_Z\right)
\end{align}
when dirty-paper coding with stochastic encoding is used with an
encoding order of $\pi$. We define the following region:
\begin{eqnarray}
\hspace{-0.5cm}&&\mathcal{R}^{\rm
DPC}\left(\pi,\bbs,\left\{\bbsigma_i\right\}_{i=1}^K,\bbsigma_Z\right)\nonumber
\\
&&\hspace{0.0cm}=\left\{(R_1,\ldots,R_K)\left|
\begin{array}{rcl}
&R_k=R_{\pi^{-1}(k)}^{\rm
DPC}\left(\pi,\left\{\bbk_i\right\}_{i=1}^K,\left\{\bbsigma_i\right\}_{i=1}^K,\bbsigma_Z\right),~k=1,\ldots,K,\\
&\textrm{for some }\left\{\bbk_i\right\}_{i=1}^K \textrm{ such
that }\bbk_i\succeq 0,~i=1,\ldots,K,\\
&\textrm{and }\sum_{i=1}^K\bbk_i \preceq\bbs
\end{array}\right.\right\}\nonumber\\
\end{eqnarray}
The secrecy capacity region of the aligned Gaussian MIMO broadcast
channel is given by the following theorem.
\begin{Theo}
\label{theorem_main_result_aligned} The secrecy capacity region of
the aligned Gaussian MIMO multi-receiver wiretap channel is given
by the convex closure of the following union
\begin{align}
\bigcup_{\pi\in\Pi} \mathcal{R}^{\rm
DPC}\left(\pi,\bbs,\left\{\bbsigma_i\right\}_{i=1}^K,\bbsigma_Z\right)
\end{align}
where $\Pi$ is the set of all possible one-to-one permutations on
$\{1,\ldots,K\}$.
\end{Theo}

We will show the achievability of the secrecy rates in
Theorem~\ref{theorem_main_result_aligned} by extending Marton's
achievable scheme for broadcast channels~\cite{Marton} to
multi-receiver wiretap channels. For that purpose, we will use
Theorem~1 of~\cite{Wei_Yu}, where the authors provided an
achievable region for Gaussian vector broadcast channels using
Marton's achievable scheme in~\cite{Marton}. While using this
result, we will combine it with a stochastic encoding scheme for
secrecy purposes. To provide a converse proof for
Theorem~\ref{theorem_main_result_aligned}, we will follow the
channel enhancement technique~\cite{Shamai_MIMO}. We will show
that for any point on the boundary of the secrecy capacity region,
there exists a degraded channel such that its secrecy capacity
region includes the secrecy capacity region of the original
channel, and furthermore, the boundaries of these two regions
intersect at this specific point.

\subsection{Achievability}
\label{sec:ach_theorem_main_result_aligned}

To show the achievability of the secrecy rates in
Theorem~\ref{theorem_main_result_aligned}, we mostly rely on the
derivation of the dirty-paper coding region for the Gaussian MIMO
broadcast channel in Theorem~1 of \cite{Wei_Yu}. We employ the
achievable scheme in~\cite{Wei_Yu} in conjunction with a
stochastic encoding scheme due to secrecy concerns. Without loss
of generality, we consider the identity permutation, i.e.,
$\pi(k)=k,~k=1,\ldots,K$. Let $(\bbv_1,\ldots,\bbv_K)$ be
arbitrarily correlated random vectors such that
\begin{align}
(\bbv_1,\ldots,\bbv_K)\rightarrow \bbx\rightarrow
(\bby_1,\ldots,\bby_K,\bbz)
\label{ach_proof_dpc_dummy_Markov_chain}
\end{align}
Using these correlated random vectors, we can construct codebooks
$\left\{\bbv_{k,1}^n(W_k,\tilde{W}_k)\right\}_{k=1}^K$, where
$W_k\in\left\{1,\ldots,2^{nR_k}\right\}$,
$\tilde{W}_k\in\big\{1,\ldots,2^{n\tilde{R}_k}\big\}
,~k=1,\ldots,K$, such that each legitimate receiver can decode the
following rates
\begin{align}
R_k+\tilde{R}_k=\frac{1}{2}\log
\frac{\left|\sum_{i=1}^{k}\bbk_i+\bbsigma_k\right|}{\left|\sum_{i=1}^{k-1}\bbk_i+\bbsigma_k\right|},
\quad k=1,\ldots,K \label{dpc_rates}
\end{align}
for some positive semi-definite matrices
$\left\{\bbk_i\right\}_{i=1}^K$ such that
$\sum_{k=1}^K\bbk_k\preceq \bbs$~\cite{Wei_Yu}. The messages
$\big\{\tilde{W}_k\big\}_{k=1}^K$ do not carry any information,
and their sole purpose is to confuse the eavesdropper. In other
words, the purpose of these messages is to make the eavesdropper
spend its decoding capability on them, preventing the eavesdropper
to decode the confidential messages $\left\{W_k\right\}_{k=1}^K$.
Thus, we need to select the rates of these dummy messages
$\big\{\tilde{R}_k\big\}_{k=1}^K$ as follows
\begin{align}
\tilde{R}_k=\frac{1}{2}\log
\frac{\left|\sum_{i=1}^{k}\bbk_i+\bbsigma_Z\right|}{\left|\sum_{i=1}^{k-1}\bbk_i+\bbsigma_Z\right|},
\quad k=1,\ldots,K \label{dummy_message_rates}
\end{align}
To achieve the rates given in~(\ref{dpc_rates}), $\left\{\bbv_k
\right\}_{k=1}^K$ should be taken as jointly Gaussian with
appropriate covariance matrices. Moreover, it is sufficient to
choose $\bbx$ as a deterministic function of $\left\{\bbv_{k}
\right\}_{k=1}^K$, and the resulting unconditional distribution of
$\bbx$ is also Gaussian with covariance matrix $\sum_{k=1}^K
\bbk_{k}$~\cite{Wei_Yu}.

To complete the proof, we need to show that the above codebook
structure fulfills all of the secrecy constraints
in~(\ref{perfect_secrecy}). To this end, we take a shortcut, by
using the fact that, if a codebook satisfies
\begin{align}
\lim_{n\rightarrow\infty} \frac{1}{n}H(W_1,\ldots,W_K|\bbz^n)\geq
\sum_{k=1}^{K}R_k \label{sufficient_condition_for_secrecy}
\end{align}
then it also satisfies all of the remaining secrecy constraints
in~(\ref{perfect_secrecy})~\cite{Ekrem_Ulukus_BC_Secrecy}. Thus,
we only check~(\ref{sufficient_condition_for_secrecy})
\begin{align}
\lefteqn{\frac{1}{n}H(W_1,\ldots,W_K|\bbz^n)=\frac{1}{n}H(W_1,\ldots,W_K,\bbz^n)
-\frac{1}{n}H(\bbz^n)}\\
&=\frac{1}{n}H(\bbv_{1,1}^n,\ldots,\bbv_{K,1}^n,W_1,\ldots,W_K,\bbz^n)
-\frac{1}{n}H(\bbv_{1,1}^n,\ldots,\bbv_{K,1}^n|W_1,\ldots,W_K,\bbz^n)\nonumber\\
&\quad -\frac{1}{n}H(\bbz^n)\\
&=\frac{1}{n}H(\bbv_{1,1}^n,\ldots,\bbv_{K,1}^n)
+\frac{1}{n}H(W_1,\ldots,W_K,\bbz^n|\bbv_{1,1}^n,\ldots,\bbv_{K,1}^n)\nonumber\\
&\quad -\frac{1}{n}H(\bbv_{1,1}^n,\ldots,\bbv_{K,1}^n|W_1,\ldots,W_K,\bbz^n)-\frac{1}{n}H(\bbz^n)\\
&\geq \frac{1}{n}H(\bbv_{1,1}^n,\ldots,\bbv_{K,1}^n)
-\frac{1}{n}I(\bbv_{1,1}^n,\ldots,\bbv_{K,1}^n;\bbz^n)
-\frac{1}{n}H(\bbv_{1,1}^n,\ldots,\bbv_{K,1}^n|W_1,\ldots,W_K,\bbz^n)
\label{ach_proof_dpc_step_1}
\end{align}
We will treat each of the three terms in
(\ref{ach_proof_dpc_step_1}) separately. Since
$(\bbv_{1,1}^n,\ldots,\bbv_{K,1}^n)$ can take $2^{n\sum_{k=1}^K
\big(R_k+\tilde{R}_k\big)}$ values uniformly, for the first term
in (\ref{ach_proof_dpc_step_1}), we have
\begin{align}
\frac{1}{n}H(\bbv_{1,1}^n,\ldots,\bbv_{K,1}^n)=\sum_{k=1}^K
R_k+\sum_{k=1}^K \tilde{R}_k \label{ach_proof_dpc_step_2}
\end{align}
The second term in (\ref{ach_proof_dpc_step_1}) can be bounded as
\begin{align}
\frac{1}{n}I(\bbv_{1,1}^n,\ldots,\bbv_{K,1}^n;\bbz^n)&\leq
I(\bbv_{1,1},\ldots,\bbv_{K,1};\bbz)+\epsilon_n
\label{ach_proof_dpc_step_3}\\
&\leq I(\bbx;\bbz)+\epsilon_n \label{ach_proof_dpc_step_4}\\
& =\frac{1}{2} \log
\frac{\left|\sum_{k=1}^K\bbk_k+\bbsigma_Z\right|}{|\bbsigma_Z|}+\epsilon_n
\label{ach_proof_dpc_step_5}
\end{align}
where $\epsilon_n\rightarrow 0$ as $n\rightarrow \infty$. The
first inequality can be shown following Lemma~8 of~\cite{Wyner},
the second inequality follows from the Markov chain
in~(\ref{ach_proof_dpc_dummy_Markov_chain}), and the equality in
(\ref{ach_proof_dpc_step_5}) comes from our choice of $\bbx$,
which is Gaussian with covariance matrix $\sum_{k=1}^K \bbk_k$. We
now consider the third term in~(\ref{ach_proof_dpc_step_1}).
First, we note that given $(W_1=w_1,\ldots,W_K=w_K)$,
$\big(\bbv_{1,1}^n,\ldots,\bbv_{K,1}^n \big)$ can take
$2^{n\sum_{k=1}^K\tilde{R}_k}$ values, where
$\sum_{k=1}^K\tilde{R}_k$ is given by
\begin{align}
\sum_{k=1}^K\tilde{R}_k=\frac{1}{2} \log
\frac{\left|\sum_{k=1}^K\bbk_k+\bbsigma_Z\right|}{|\bbsigma_Z|}
\label{ach_proof_dpc_step_6}
\end{align}
using our selection in (\ref{dummy_message_rates}). Thus,
(\ref{ach_proof_dpc_step_6}) implies that given
$(W_1=w_1,\ldots,W_K=w_K)$, the eavesdropper can decode
$\big(\bbv_{1,1}^n,\ldots,\bbv_{K,1}^n \big)$ with vanishingly
small probability of error. Hence, using Fano's lemma, we get
\begin{align}
\frac{1}{n}H(\bbv_{1,1}^n,\ldots,\bbv_{K,1}^n|W_1,\ldots,W_K,\bbz^n)\leq
\frac{1}{n}\left[1+\gamma_n \left(\sum_{k=1}^K
\tilde{R}_k\right)\right] \label{ach_proof_dpc_step_7}
\end{align}
where $\gamma_n\rightarrow 0$ as $n\rightarrow \infty$. Thus,
plugging (\ref{ach_proof_dpc_step_2}),
(\ref{ach_proof_dpc_step_5}) and (\ref{ach_proof_dpc_step_7}) into
(\ref{ach_proof_dpc_step_1}) yields
\begin{align}
\lim_{n\rightarrow\infty} \frac{1}{n}H(W_1,\ldots,W_K|\bbz^n) \geq
\sum_{k=1}^K R_k
\end{align}
which ensures that the rates
\begin{align}
R_k=\frac{1}{2}\log
\frac{\left|\sum_{i=1}^{k}\bbk_i+\bbsigma_k\right|}{\left|\sum_{i=1}^{k-1}\bbk_i+\bbsigma_k\right|}
-\frac{1}{2}\log
\frac{\left|\sum_{i=1}^{k}\bbk_i+\bbsigma_Z\right|}{\left|\sum_{i=1}^{k-1}\bbk_i+\bbsigma_Z\right|},
\quad k=1,\ldots,K
\end{align}
can be transmitted in perfect secrecy.

\subsection{Converse}
To show the converse, we consider the maximization of the
following expression
\begin{align}
\sum_{k=1}^K \mu_k R_k \label{converse_proof_aligned_step_1}
\end{align}
where $\mu_k\geq 0,~k=1,\ldots,K$. We note that the maximum value
of (\ref{converse_proof_aligned_step_1}) traces the boundary of
the secrecy capacity region, i.e., its maximum value for any
non-negative vector $[\mu_1~\ldots~\mu_K]$ will give us a point on
the boundary of the secrecy capacity region. Let us define
$\pi(\cdot)$ to be a one-to-one permutation on $\{1,\ldots,K\}$
such that
\begin{align}
0\leq \mu_{\pi(1)}\leq \ldots \leq \mu_{\pi(K)}
\end{align}
Furthermore, let $0< m\leq K$ of $\left\{\mu_k\right\}_{k=1}^K$ be
strictly positive, i.e., $\mu_{\pi(1)}=\ldots=\mu_{\pi(K-m)}=0,$
and $\mu_{\pi(K-m+1)}>0$. We now define another permutation
$\pi^{\prime}(\cdot)$ on the strictly positive elements of
$\left\{\mu_k\right\}_{k=1}^K$ such that
$\pi^{\prime}(l)=\pi(K-m+l),~\l=1,\ldots,m$. Then,
(\ref{converse_proof_aligned_step_1}) can be expressed as
\begin{align}
\sum_{k=1}^K \mu_k R_k=\sum_{k=1}^K \mu_{\pi(k)}
R_{\pi(k)}=\sum_{k=1}^m \mu_{\pi^{\prime}(k)} R_{\pi^{\prime}(k)}
\label{converse_proof_aligned_step_2}
\end{align}
We will show that
\begin{align}
\max \sum_{k=1}^K \mu_k R_k &=\max \sum_{k=1}^m \mu_{\pi^{\prime}(k)} R_{\pi^{\prime}(k)}\\
& \leq \max \sum_{k=1}^m \frac{\mu_{\pi^{\prime}(k)}}{2}\log
\frac{\left|\sum_{i=1}^k
\bbk_{\pi^{\prime}(i)}+\bbsigma_{\pi^{\prime}(k)}\right|}{\left|\sum_{i=1}^{k-1}
\bbk_{\pi^{\prime}(i)}+\bbsigma_{\pi^{\prime}(k)}\right|}
\nonumber\\
&\qquad\qquad -\sum_{k=1}^m \frac{\mu_{\pi^{\prime}(k)}}{2}\log
\frac{\left|\sum_{i=1}^k
\bbk_{\pi^{\prime}(i)}+\bbsigma_{Z}\right|}{\left|\sum_{i=1}^{k-1}
\bbk_{\pi^{\prime}(i)}+\bbsigma_{Z}\right|}
\label{converse_proof_aligned_step_3}
\end{align}
where the last maximization is over all positive semi-definite
matrices $\left\{ \bbk_{\pi^{\prime}(k)}\right\}_{k=1}^m$ such
that $\sum_{k=1}^m\bbk_{\pi^{\prime}(k)} \preceq \bbs$. Since the
right hand side of (\ref{converse_proof_aligned_step_3}) is
achievable, if we can show that
(\ref{converse_proof_aligned_step_3}) holds for any non-negative
vector $[\mu_1~\ldots~\mu_K]$, this will complete the proof of
Theorem~\ref{theorem_main_result_aligned}. To simplify the
notation, without loss of generality, we assume that
$\pi^{\prime}(k)=k,~k=1,\ldots,m$. This assumption is equivalent
to the assumption that $0<\mu_1\leq \ldots \leq \mu_m,\textrm{ and
}\mu_k=0,~k=m+1,\ldots,K$.

We now investigate the maximization in
(\ref{converse_proof_aligned_step_3}). The objective function in
(\ref{converse_proof_aligned_step_3}) is generally non-convex in
the covariance matrices $\left\{
\bbk_{\pi^{\prime}(k)}\right\}_{k=1}^m$ implying that the KKT
conditions for this problem are necessary, but not sufficient. Let
us construct the Lagrangian for this optimization problem
\begin{align}
L\left(\left\{\bbm_i\right\}_{i=1}^m,\bbm_Z\right)=\sum_{k=1}^m
\mu_k R_k^{G}+\sum_{k=1}^m {\rm tr}(\bbk_k\bbm_k)+ {\rm
tr}\left(\left(\bbs-\sum_{k=1}^m \bbk_k\right)\bbm_Z\right)
\label{converse_proof_aligned_step_4}
\end{align}
where the Lagrange multipliers
$\left\{\bbm_i\right\}_{i=1}^m,\bbm_Z$ are positive semi-definite
matrices, and we defined $\left\{R_k^{G}\right\}_{k=1}^m$ as
follows,
\begin{align}
R_k^{G}=\frac{1}{2}\log \frac{\left|\sum_{i=1}^k
\bbk_i+\bbsigma_k\right|}{\left|\sum_{i=1}^{k-1}
\bbk_i+\bbsigma_k\right|} -\frac{1}{2}\log
\frac{\left|\sum_{i=1}^k
\bbk_i+\bbsigma_Z\right|}{\left|\sum_{i=1}^{k-1}
\bbk_i+\bbsigma_Z\right|},\quad k=1,\ldots,m
\end{align}
The gradient of
$L\left(\left\{\bbm_i\right\}_{i=1}^m,\bbm_Z\right)$ with respect
to $\bbk_j$ for any $j=1,\ldots,m-1$, is given by
\begin{align}
\nabla_{\bbk_j}L\left(\left\{\bbm_i\right\}_{i=1}^m,\bbm_Z\right)
&=\sum_{k=j}^m
\frac{\mu_k}{2}\left(\sum_{i=1}^k\bbk_i+\bbsigma_k\right)^{-1}
-\sum_{k=j+1}^m \frac{\mu_k}{2}\left(\sum_{i=1}^{k-1}\bbk_i+\bbsigma_k\right)^{-1}\nonumber\\
&\quad -\sum_{k=j}^m
\frac{\mu_k}{2}\left(\sum_{i=1}^k\bbk_i+\bbsigma_Z\right)^{-1}
+\sum_{k=j+1}^m \frac{\mu_k}{2}\left(\sum_{i=1}^{k-1}\bbk_i+\bbsigma_Z\right)^{-1}\nonumber\\
&\quad+\bbm_j-\bbm_Z \label{converse_proof_aligned_step_5}
\end{align}
and the gradient of
$L\left(\left\{\bbm_i\right\}_{i=1}^m,\bbm_Z\right)$ with respect
to $\bbk_m$ is given by
\begin{align}
\nabla_{\bbk_m}L\left(\left\{\bbm_i\right\}_{i=1}^m,\bbm_Z\right)
&=\frac{\mu_m}{2}\left(\sum_{i=1}^m\bbk_i+\bbsigma_m\right)^{-1} -
\frac{\mu_m}{2}\left(\sum_{i=1}^m\bbk_i+\bbsigma_Z\right)^{-1}
+\bbm_m-\bbm_Z \label{converse_proof_aligned_step_6}
\end{align}
The KKT conditions are given by
\begin{align}
\nabla_{\bbk_j}L\left(\left\{\bbm_i\right\}_{i=1}^m,\bbm_Z\right)&=\bzero,\quad j=1,\ldots,m \label{converse_proof_aligned_step_7}\\
{\rm tr}(\bbk_j\bbm_j)&=0,\quad j=1,\ldots,m \label{converse_proof_aligned_step_8}\\
{\rm tr}\left(\left(\bbs-\sum_{k=1}^m
\bbk_k\right)\bbm_Z\right)&=0
\label{converse_proof_aligned_step_9}
\end{align}
We note that since ${\rm tr}(\bbk_j\bbm_j)={\rm
tr}(\bbm_j\bbk_j)$, and $\bbm_j\succeq \bzero,\bbk_j\succeq
\bzero$, we have $\bbm_j\bbk_j=\bbk_j\bbm_j=\bzero$. Thus, the KKT
conditions in (\ref{converse_proof_aligned_step_8}) are equivalent
to
\begin{align}
\bbm_j\bbk_j=\bbk_j\bbm_j&=\bzero,\quad j=1,\ldots,m
\label{converse_proof_aligned_step_10}
\end{align}
Similarly, we also have
\begin{align}
\bbm_Z \left(\bbs-\sum_{k=1}^m \bbk_k\right)=
\left(\bbs-\sum_{k=1}^m \bbk_k\right)\bbm_Z =\bzero
\label{converse_proof_aligned_step_11}
\end{align}
Subtracting the gradient of the Lagrangian with respect to
$\bbk_{j+1}$ from the one with respect to $\bbk_j$, for
$j=1,\ldots,m-1$, we get
\begin{align}
\lefteqn{\nabla_{\bbk_j}L\left(\left\{\bbm_i\right\}_{i=1}^m,\bbm_Z\right)-\nabla_{\bbk_{j+1}}L\left(\left\{\bbm_i\right\}_{i=1}^m,\bbm_Z\right)}\nonumber\\
&=\sum_{k=j}^m
\frac{\mu_k}{2}\left(\sum_{i=1}^k\bbk_i+\bbsigma_k\right)^{-1}
-\sum_{k=j+1}^m \frac{\mu_k}{2}\left(\sum_{i=1}^{k-1}\bbk_i+\bbsigma_k\right)^{-1}\nonumber\\
&\quad -\sum_{k=j}^m
\frac{\mu_k}{2}\left(\sum_{i=1}^k\bbk_i+\bbsigma_Z\right)^{-1}
+\sum_{k=j+1}^m \frac{\mu_k}{2}\left(\sum_{i=1}^{k-1}\bbk_i+\bbsigma_Z\right)^{-1}+\bbm_j-\bbm_Z\nonumber\\
&\quad -\sum_{k=j+1}^m
\frac{\mu_k}{2}\left(\sum_{i=1}^k\bbk_i+\bbsigma_k\right)^{-1}
+\sum_{k=j+2}^m \frac{\mu_k}{2}\left(\sum_{i=1}^{k-1}\bbk_i+\bbsigma_k\right)^{-1}\nonumber\\
&\quad +\sum_{k=j+1}^m
\frac{\mu_k}{2}\left(\sum_{i=1}^k\bbk_i+\bbsigma_Z\right)^{-1}
-\sum_{k=j+2}^m \frac{\mu_k}{2}\left(\sum_{i=1}^{k-1}\bbk_i+\bbsigma_Z\right)^{-1}-\bbm_{j+1}+\bbm_Z\\
&= \frac{\mu_j}{2}\left(\sum_{i=1}^j\bbk_i+\bbsigma_j\right)^{-1}
- \frac{\mu_{j+1}}{2}\left(\sum_{i=1}^{j}\bbk_i+\bbsigma_{j+1}\right)^{-1}\nonumber\\
&\quad -
\frac{\mu_j}{2}\left(\sum_{i=1}^j\bbk_i+\bbsigma_Z\right)^{-1} +
\frac{\mu_{j+1}}{2}\left(\sum_{i=1}^{j}\bbk_i+\bbsigma_Z\right)^{-1}+\bbm_j-\bbm_{j+1}
\label{converse_proof_aligned_step_12}
\end{align}
Thus, using (\ref{converse_proof_aligned_step_10}),
(\ref{converse_proof_aligned_step_11}),
(\ref{converse_proof_aligned_step_12}), we can express the KKT
conditions in (\ref{converse_proof_aligned_step_7}),
(\ref{converse_proof_aligned_step_8}),
(\ref{converse_proof_aligned_step_9}) as follows
\begin{align}
\mu_j\left(\sum_{i=1}^j\bbk_i+\bbsigma_j\right)^{-1}
+(\mu_{j+1}-\mu_j)\left(\sum_{i=1}^j\bbk_i+\bbsigma_Z\right)^{-1}+\bbm_j&=
\mu_{j+1}\left(\sum_{i=1}^{j}\bbk_i+\bbsigma_{j+1}\right)^{-1}\nonumber\\
&\quad +\bbm_{j+1},\quad j=1,\ldots,m-1
\label{converse_proof_aligned_step_13}\\
\mu_m \left(\sum_{i=1}^m\bbk_i+\bbsigma_m\right)^{-1}+\bbm_m&=
 \mu_m \left(\sum_{i=1}^m\bbk_i+\bbsigma_Z\right)^{-1}
+\bbm_Z \label{converse_proof_aligned_step_14}\\
\bbk_j\bbm_j=\bbm_j\bbk_j&=\bzero,\quad j=1,\ldots,m
\label{converse_proof_aligned_step_15}\\
\bbm_Z \left(\bbs-\sum_{k=1}^m \bbk_k\right)
=\left(\bbs-\sum_{k=1}^m \bbk_k\right)\bbm_Z &=\bzero
\label{converse_proof_aligned_step_16}
\end{align}
where we also embed the multiplications by 2 into the Lagrange
multipliers.

We now present a lemma which will be instrumental in constructing
a degraded Gaussian MIMO multi-receiver wiretap channel, such that
the secrecy capacity region of the constructed channel includes
the secrecy capacity region of the original channel, and the
boundary of the secrecy capacity region of this constructed
channel coincides with the boundary of the secrecy capacity region
of the original channel at a certain point for a given
non-negative vector $[\mu_1~\ldots~\mu_K]$.

\begin{Lem}
\label{lemma_construct_degraded_channel} Given the covariance
matrices $\{\bbk_j\}_{j=1}^m$ satisfying the KKT conditions given
in
(\ref{converse_proof_aligned_step_13})-(\ref{converse_proof_aligned_step_16}),
there exist noise covariance matrices
$\big\{\tilde{\bbsigma}_j\big\}_{j=1}^m$ such that
\begin{enumerate}
\item $\tilde{\bbsigma}_j \preceq \bbsigma_j,~j=1,\ldots,m$. \item
$\bzero \prec \tilde{\bbsigma}_1  \preceq \ldots
\preceq\tilde{\bbsigma}_m \preceq \bbsigma_Z  $ \item
$\mu_j\left(\sum_{i=1}^j\bbk_i+\tilde{\bbsigma}_j\right)^{-1}
+(\mu_{j+1}-\mu_j)\left(\sum_{i=1}^j\bbk_i+\bbsigma_Z\right)^{-1}=
\mu_{j+1}\left(\sum_{i=1}^{j}\bbk_i+\tilde{\bbsigma}_{j+1}\right)^{-1},\vspace{0.32cm}\newline
\vspace{0.32cm} \textrm{for }j=1,\ldots,m-1, \textrm{ and }
\newline
\mu_m\left(\sum_{i=1}^m\bbk_i+\tilde{\bbsigma}_m\right)^{-1}=
\mu_{m}\big(\sum_{i=1}^m\bbk_i+\bbsigma_Z\big)^{-1}+\bbm_Z$ \item
$\left(\sum_{i=1}^j\bbk_i+\tilde{\bbsigma}_j\right)^{-1}\left(\sum_{i=1}^{j-1}\bbk_i+\tilde{\bbsigma}_j\right)
=\left(\sum_{i=1}^j\bbk_i+\bbsigma_j\right)^{-1}\left(\sum_{i=1}^{j-1}\bbk_i+\bbsigma_j\right)
\vspace{0.32cm}
\newline \textrm{for }j=1,\ldots,m$
\item
$\left(\bbs+\tilde{\bbsigma}_m\right)\left(\sum_{i=1}^m\bbk_i+\tilde{\bbsigma}_m\right)^{-1}
=\left(\bbs+\bbsigma_Z\right)\left(\sum_{i=1}^m\bbk_i+\bbsigma_Z\right)^{-1}$
\end{enumerate}
\end{Lem}

The proof of this lemma is given in
Appendix~\ref{proof_of_lemma_construct_degraded_channel}.


Without loss of generality, we have already fixed
$[\mu_1~\ldots~\mu_K]$ such that $0<\mu_1\leq\ldots\leq \mu_m$,
and $\mu_k=0,~k=m+1,\ldots,K$ for some $0<m\leq K$. For this fixed
$[\mu_1~\ldots~\mu_K]$, assume that $\{\bbk_k^{*}\}_{k=1}^m$
achieves the maximum of (\ref{converse_proof_aligned_step_3}).
Since these covariance matrices need to satisfy the KKT conditions
given in
(\ref{converse_proof_aligned_step_13})-(\ref{converse_proof_aligned_step_16}),
Lemma~\ref{lemma_construct_degraded_channel} ensures the existence
of the covariance matrices
$\big\{\tilde{\bbsigma}_j\big\}_{j=1}^m$ that have the properties
listed in Lemma~\ref{lemma_construct_degraded_channel}. Thus, we
can define a degraded Gaussian MIMO  multi-receiver wiretap
channel that has the following noise covariance matrices
\begin{align}
\hat{\bbsigma}_k=\left\{
\begin{array}{rcl}
\tilde{\bbsigma}_k,&\quad 1\leq k\leq m \\
\alpha_{k-m}\tilde{\bbsigma}_1,&\quad m+1 \leq k\leq K
\end{array}
\right.
\end{align}
where $0 < \alpha_{k-m} \leq 1 $ are chosen to satisfy
$\alpha_{k-m}\tilde{\bbsigma}_1\preceq \bbsigma_k$ for
$k=m+1,\ldots,K$, where the existence of such
$\{\alpha_{k-m}\}_{k=m+1}^K$ are ensured by the positive
definiteness of $\left\{\bbsigma_{k}\right\}_{k=1}^K$. The noise
covariance matrix of the eavesdropper is the same as in the
original channel, i.e., $\bbsigma_Z$. Since this channel is
degraded, its secrecy capacity region is given by
Theorem~\ref{main_result}. Moreover, since $\hat{\bbsigma}_k
\preceq \bbsigma_k,~k=1,\ldots,K,$ and the noise covariance
matrices in the constructed degraded channel and the original
channel are the same, the secrecy capacity region of this degraded
channel outer bounds that of the original channel. Next, we show
that for the so-far fixed $[\mu_1~\ldots~\mu_K]$, the boundaries
of these two regions intersect at this point. For this purpose,
reconsider the maximization problem in
(\ref{converse_proof_aligned_step_1})
\begin{align}
\max \sum_{k=1}^K \mu_k R_k&=\max \sum_{k=1}^m \mu_k R_k \label{fixed_gamma_implies}\\
&\leq \max_{\substack{\bbk_i\succeq 0,~i=1,\ldots,K\\
\sum_{i=1}^K\bbk_i \preceq \bbs} }\sum_{k=1}^m \frac{\mu_k }{2}
\left[\log\frac{\left|\sum_{i=1}^k \bbk_i +\sum_{i=m+1}^K \bbk_i
+\tilde{\bbsigma}_k\right|}
{\left|\sum_{i=1}^{k-1} \bbk_i +\sum_{i=m+1}^K \bbk_i +\tilde{\bbsigma}_k\right|}\right. \nonumber \\
&\hspace{4.8cm}\left.-\log\frac{\left|\sum_{i=1}^k \bbk_i
+\sum_{i=m+1}^K \bbk_i +\bbsigma_Z\right|}
{\left|\sum_{i=1}^{k-1} \bbk_i +\sum_{i=m+1}^K \bbk_i +\bbsigma_Z\right|}\right]\label{constructed_channel_outer_bounds}\\
&= \max_{\substack{\bbk_i\succeq 0,~i=1,\ldots,m\\
\sum_{i=1}^m\bbk_i \preceq \bbs} }\sum_{k=1}^m \frac{\mu_k }{2}
\left[\log \frac{\left|\sum_{i=1}^k \bbk_i
+\tilde{\bbsigma}_k\right|} {\left|\sum_{i=1}^{k-1} \bbk_i
+\tilde{\bbsigma}_k\right|}-\log\frac{\left|\sum_{i=1}^k \bbk_i
+\bbsigma_Z\right|} {\left|\sum_{i=1}^{k-1} \bbk_i
+\bbsigma_Z\right|}\right] \label{no_rate_to_zero_gammas}
\end{align}
where (\ref{fixed_gamma_implies}) is implied by the fact that for
the fixed $[\mu_1~\ldots~\mu_K]$, we assumed that
$\mu_k=0,~k=m+1,\ldots,K$ and $0<\mu_1\leq \ldots\leq \mu_m$,
(\ref{constructed_channel_outer_bounds}) follows from the facts
that the constructed degraded channel includes the secrecy
capacity region of the original channel, and the secrecy capacity
region of the degraded channel is given by
Theorem~\ref{main_result}. The last equation, i.e.,
(\ref{no_rate_to_zero_gammas}), comes from the fact that, since
$\mu_k=0,~k=m+1,\ldots,K$, there is no loss of optimality in
choosing $\bbk_k=\bzero,~k=m+1,\ldots,K$. We now claim that the
maximum in (\ref{no_rate_to_zero_gammas}) is achieved by
$\{\bbk^{*}_k\}_{k=1}^m$. To prove this claim, we first define
\begin{align}
R_k^{*}& =\frac{1}{2}\log\frac{\left|\sum_{i=1}^k
\bbk_i^{*}+\tilde{\bbsigma}_k\right|}{\left|\sum_{i=1}^{k-1}\bbk_i^{*}+\tilde{\bbsigma}_k\right|}
-\frac{1}{2}\log \frac{\left|\sum_{i=1}^k
\bbk_i^{*}+\bbsigma_Z\right|}{\left|\sum_{i=1}^{k-1}\bbk_i^{*}+\bbsigma_Z\right|},\quad
k=1,\ldots,m
\end{align}
and
\begin{align}
\hat{R}_k& =\frac{1}{2}\log
\frac{\left|\sum_{i=1}^k
\bbk_i+\tilde{\bbsigma}_k\right|}{\left|\sum_{i=1}^{k-1}\bbk_i+\tilde{\bbsigma}_k\right|}
-\frac{1}{2}\log \frac{\left|\sum_{i=1}^k
\bbk_i+\bbsigma_Z\right|}{\left|\sum_{i=1}^{k-1}\bbk_i+\bbsigma_Z\right|},\quad
k=1,\ldots,m
\end{align}
for some arbitrary positive semi-definite matrices
$\{\bbk_i\}_{i=1}^m$ such that $\sum_{i=1}^m\bbk_i \preceq \bbs$.
To prove that the maximum in (\ref{no_rate_to_zero_gammas}) is
achieved by $\{\bbk^{*}_k\}_{k=1}^m$, we will show that
\begin{align}
\sum_{k=1}^m \mu_k R^{*}_k-\sum_{k=1}^m \mu_k\hat{R}_k \geq 0
\label{dummy_equations_forever}
\end{align}
To this end, consider the first summation in
(\ref{dummy_equations_forever})
\begin{align}
\sum_{k=1}^m \mu_k R^{*}_k &=\sum_{k=1}^m
\frac{\mu_k}{2}\left(\log\left| \sum_{i=1}^k
\bbk_i^{*}+\tilde{\bbsigma}_k\right|
-\log\left| \sum_{i=1}^k \bbk_i^{*}+\bbsigma_Z\right|\right)\nonumber \\
&\quad -\sum_{k=2}^m \frac{\mu_k}{2}\left(\log\left| \sum_{i=1}^{k-1} \bbk_i^{*}+\tilde{\bbsigma}_k\right|
-\log \left| \sum_{i=1}^{k-1} \bbk_i^{*}+\bbsigma_Z\right|\right)\nonumber\\
&\quad -\frac{\mu_1}{2}\log\frac{\big|\tilde{\bbsigma}_1\big|}{\left|\bbsigma_Z\right|}\\
&=\sum_{k=1}^m \frac{\mu_k}{2}\left(\log\left| \sum_{i=1}^k \bbk_i^{*}+\tilde{\bbsigma}_k\right|
-\log \left| \sum_{i=1}^k \bbk_i^{*}+\bbsigma_Z\right|\right)\nonumber \\
&\quad -\sum_{k=1}^{m-1} \frac{\mu_{k+1}}{2}\left(\log \left|
\sum_{i=1}^{k} \bbk_i^{*}+\tilde{\bbsigma}_{k+1}\right|
-\log \left| \sum_{i=1}^{k} \bbk_i^{*}+\bbsigma_Z\right|\right)\nonumber\\
&\quad -\frac{\mu_1}{2}\log\frac{\big|\tilde{\bbsigma}_1\big|}{\left|\bbsigma_Z\right|} \\
&=\frac{\mu_m}{2} \log\frac{\left| \sum_{i=1}^m \bbk_i^{*}+\tilde{\bbsigma}_m\right|}{\left| \sum_{i=1}^m \bbk_i^{*}+\bbsigma_Z\right|}\nonumber\\
&\quad +
\sum_{k=1}^{m-1} \frac{\mu_k}{2}\left(\log\left| \sum_{i=1}^k \bbk_i^{*}+\tilde{\bbsigma}_k\right|
-\log\left| \sum_{i=1}^k \bbk_i^{*}+\bbsigma_Z\right|\right)\nonumber \\
&\quad -\sum_{k=1}^{m-1} \frac{\mu_{k+1}}{2}\left(\log\left| \sum_{i=1}^{k} \bbk_i^{*}+\tilde{\bbsigma}_{k+1}\right|
-\log \left| \sum_{i=1}^{k} \bbk_i^{*}+\bbsigma_Z\right|\right)\nonumber\\
&\quad -\frac{\mu_1}{2}\log \frac{\big|\tilde{\bbsigma}_1\big|}{\left|\bbsigma_Z\right|} \\
&=\frac{\mu_m}{2} \log\frac{\left| \sum_{i=1}^m
\bbk_i^{*}+\tilde{\bbsigma}_m\right|}{\left| \sum_{i=1}^m
\bbk_i^{*}+\bbsigma_Z\right|} +
\sum_{k=1}^{m-1} \frac{\mu_k}{2}\log \left| \sum_{i=1}^k \bbk_i^{*}+\tilde{\bbsigma}_k\right|\nonumber\\
&\quad +\sum_{k=1}^{m-1} \frac{\mu_{k+1}-\mu_k}{2}\log \left|
\sum_{i=1}^k \bbk_i^{*}+\bbsigma_Z\right| -\sum_{k=1}^{m-1}
\frac{\mu_{k+1}}{2}\log\left| \sum_{i=1}^{k}
\bbk_i^{*}+\tilde{\bbsigma}_{k+1}\right|\nonumber\\
&\quad
-\frac{\mu_1}{2}\log\frac{\big|\tilde{\bbsigma}_1\big|}{\left|\bbsigma_Z\right|}
\label{dummy_equations_forever_1}
\end{align}
Similarly, we have
\begin{align}
\sum_{k=1}^m \mu_k \hat{R}_k &=\frac{\mu_m}{2} \log\frac{\left|
\sum_{i=1}^m \bbk_i+\tilde{\bbsigma}_m\right|}{\left| \sum_{i=1}^m
\bbk_i+ \bbsigma_Z\right|} +
\sum_{k=1}^{m-1} \frac{\mu_k}{2}\log\left| \sum_{i=1}^k \bbk_i+\tilde{\bbsigma}_k\right|\nonumber\\
&\quad +\sum_{k=1}^{m-1} \frac{\mu_{k+1}-\mu_k}{2}\log\left|
\sum_{i=1}^k \bbk_i+\bbsigma_Z\right| -\sum_{k=1}^{m-1}
\frac{\mu_{k+1}}{2}\log\left| \sum_{i=1}^{k}
\bbk_i+\tilde{\bbsigma}_{k+1}\right|\nonumber\\ &\quad
-\frac{\mu_1}{2}\log\frac{\big|\tilde{\bbsigma}_1\big|}{\left|\bbsigma_Z\right|}
\label{dummy_equations_forever_2}
\end{align}
We define the following matrices
\begin{align}
\bbdelta_k=\sum_{i=1}^{k}\bbk_i-\sum_{i=1}^k\bbk^{*}_i,\quad
k=1,\ldots,m \label{dummy_equations_forever_3}
\end{align}
Using (\ref{dummy_equations_forever_1}),
(\ref{dummy_equations_forever_2}) and
(\ref{dummy_equations_forever_3}), the difference in
(\ref{dummy_equations_forever}) can be expressed as
\begin{align}
\sum_{k=1}^m \mu_k R^{*}_k-\sum_{k=1}^m \mu_k\hat{R}_k&=
\frac{\mu_m}{2} \log\frac{\left| \sum_{i=1}^m
\bbk_i^{*}+\tilde{\bbsigma}_m\right|}{\left| \sum_{i=1}^m
\bbk_i^{*}+\bbsigma_Z\right|} -\frac{\mu_m}{2} \log\frac{\left|
\sum_{i=1}^m \bbk_i+\tilde{\bbsigma}_m\right|}{\left| \sum_{i=1}^m
\bbk_i+\bbsigma_Z\right|}
\nonumber \\
&\quad-
\sum_{k=1}^{m-1} \frac{\mu_k}{2}
\log\left|\bbi+\left( \sum_{i=1}^k \bbk_i^{*}+\tilde{\bbsigma}_k\right)^{-1}\bbdelta_k\right|\nonumber\\
&\quad -\sum_{k=1}^{m-1} \frac{\mu_{k+1}-\mu_k}{2}
\log\left|\bbi+\left( \sum_{i=1}^k \bbk_i^{*}+\bbsigma_Z\right)^{-1}\bbdelta_k\right|\nonumber\\
&\quad +\sum_{k=1}^{m-1} \frac{\mu_{k+1}}{2}\log \left|
\bbi+\left(\sum_{i=1}^{k}
\bbk_i^{*}+\tilde{\bbsigma}_{k+1}\right)^{-1}\bbdelta_k\right|
\label{dummy_equations_forever_4}
\end{align}
We first note that
\begin{align}
\frac{\left| \sum_{i=1}^m
\bbk_i^{*}+\tilde{\bbsigma}_m\right|}{\left| \sum_{i=1}^m
\bbk_i^{*}+\bbsigma_Z\right|} =\frac{\left|
\bbs+\tilde{\bbsigma}_m\right|}{\left| \bbs+\bbsigma_Z\right|}\geq
\frac{\left| \sum_{i=1}^m \bbk_i+\tilde{\bbsigma}_m\right|}{\left|
\sum_{i=1}^m \bbk_i+\bbsigma_Z\right|}
\label{dummy_equations_forever_5}
\end{align}
where the equality is due to the fifth part of
Lemma~\ref{lemma_construct_degraded_channel}, and the inequality
follows from the fact that the function
\begin{align}
\frac{\big|\bba+\tilde{\bbsigma}_m\big|}{|\bba+\bbsigma_Z|}
\end{align}
is monotonically increasing in the positive semi-definite matrix
$\bba$ as can be deduced from
(\ref{shamao_inequality_without_proof}), and that
$\sum_{i=1}^m\bbk_i \preceq \bbs$. Furthermore, we have
\begin{align}
\lefteqn{\frac{\mu_k}{\mu_{k+1}}\log\left|\bbi+\left( \sum_{i=1}^k
\bbk_i^{*}+\tilde{\bbsigma}_k\right)^{-1}\bbdelta_k\right|
+\frac{\mu_{k+1}-\mu_k}{\mu_{k+1}}
\log\left|\bbi+\left( \sum_{i=1}^k \bbk_i^{*}+\bbsigma_Z\right)^{-1}\bbdelta_k\right|}\nonumber\\
&\leq  \log\left|\bbi+\frac{\mu_k}{\mu_{k+1}}\left( \sum_{i=1}^k
\bbk_i^{*}+\tilde{\bbsigma}_k\right)^{-1}\bbdelta_k
+\frac{\mu_{k+1}-\mu_k}{\mu_{k+1}}\left( \sum_{i=1}^k
\bbk_i^{*}+\bbsigma_Z\right)^{-1}\bbdelta_k\right|
\label{log_det_is_concave} \\
& =  \log\left|\bbi+\left( \sum_{i=1}^k
\bbk_i^{*}+\tilde{\bbsigma}_{k+1}\right)^{-1}\bbdelta_k \right|
\label{kkt_imply_5}
\end{align}
where the inequality in (\ref{log_det_is_concave}) follows from
the concavity of $\log |\cdot|$ in positive semi-definite
matrices, and (\ref{kkt_imply_5}) follows from the third part of
Lemma~\ref{lemma_construct_degraded_channel}. Using
(\ref{dummy_equations_forever_5}) and (\ref{kkt_imply_5}) in
(\ref{dummy_equations_forever_4}) yields
\begin{align}
\sum_{k=1}^m \mu_k R^{*}_k-\sum_{k=1}^m \mu_k\hat{R}_k \geq 0
\end{align}
which implies that the maximum in (\ref{no_rate_to_zero_gammas})
is achieved by $\{\bbk^{*}_k\}_{k=1}^m$. Thus, using this fact in
(\ref{no_rate_to_zero_gammas}), we get
\begin{align}
\max \sum_{k=1}^K \mu_k R_k& \leq \sum_{k=1}^m \frac{\mu_k }{2}
\left[\log\frac{\left|\sum_{i=1}^k \bbk_i^{*}
+\tilde{\bbsigma}_k\right|} {\left|\sum_{i=1}^{k-1} \bbk_i^{*}
+\tilde{\bbsigma}_k\right|}-\log\frac{\left|\sum_{i=1}^k
\bbk_i^{*} +\bbsigma_Z\right|}
{\left|\sum_{i=1}^{k-1} \bbk_i^{*} +\bbsigma_Z\right|}\right] \\
&=\sum_{k=1}^m \frac{\mu_k }{2} \left[\log\frac{\left|\sum_{i=1}^k
\bbk_i^{*} +\bbsigma_k\right|} {\left|\sum_{i=1}^{k-1} \bbk_i^{*}
+ \bbsigma_k\right|}-\log\frac{\left|\sum_{i=1}^k \bbk_i^{*}
+\bbsigma_Z\right|} {\left|\sum_{i=1}^{k-1} \bbk_i^{*}
+\bbsigma_Z\right|}\right] \label{dummy_equations_forever_6}
\end{align}
where the equality follows from the fourth part of
Lemma~\ref{lemma_construct_degraded_channel}. Since the right hand
side of (\ref{dummy_equations_forever_6}) is achievable, and we
can get a similar outer bound for any non-negative vector
$[\mu_1~\ldots~\mu_K]$, this completes the converse proof for the
aligned Gaussian MIMO channel.

\section{General Gaussian MIMO Multi-receiver Wiretap \break Channel}

In this final part of the paper, we consider the general Gaussian
multi-receiver wiretap channel and prove its secrecy capacity
region. The main idea in this section is to construct an aligned
channel that is indexed by a scalar variable, and then show that
this aligned channel has the same secrecy capacity region as the
original channel in the limit of this indexing parameter on the
constructed aligned channel. This argument was previously used
in~\cite{Shamai_MIMO,Tie_Liu_MIMO_WT}. The way we use this
argument here is different from~\cite{Shamai_MIMO} because there
are no secrecy constraints in~\cite{Shamai_MIMO}, and it is
different from~\cite{Tie_Liu_MIMO_WT} because there are multiple
legitimate receivers here.

Given the covariance matrices $\left\{\bbk_k\right\}_{k=1}^K$ such
that $\sum_{k=1}^K \bbk_k \preceq \bbs$, we define the following
rates
\begin{align}
\lefteqn{R_k^{\rm
DPC}\left(\pi,\left\{\bbk_i\right\}_{i=1}^K,\left\{\bbsigma_i\right\}_{i=1}^K,\bbsigma_Z,\left\{\bbh_i\right\}_{i=1}^K,\bbh_Z\right)}\nonumber\\
&=\frac{1}{2}\log \frac{\left|\bbh_{\pi(k)}\left(\sum_{i=1}^k
\bbk_{\pi(i)}\right)\bbh_{\pi(k)}^\top+\bbsigma_{\pi(k)}\right|}{\left|\bbh_{\pi(k)}\left(\sum_{i=1}^{k-1}
\bbk_{\pi(i)}\right)\bbh_{\pi(k)}^\top+\bbsigma_{\pi(k)}\right|} -
\frac{1}{2}\log\frac{\left|\bbh_{Z}\left(\sum_{i=1}^k
\bbk_{\pi(i)}\right)\bbh_Z^\top+\bbsigma_Z\right|}{\left|\bbh_{Z}\left(\sum_{i=1}^{k-1}
\bbk_{\pi(i)}\right)\bbh_{Z}^\top+\bbsigma_Z\right|},\nonumber\\
&\hspace{12cm} k=1,\ldots,K
\end{align}
where $\pi(\cdot)$ is a one-to-one permutation on
$\{1,\ldots,K\}$. We also note that the subscript of $R_k^{\rm
DPC}\left(\pi,\left\{\bbk_i\right\}_{i=1}^K,\left\{\bbsigma_i\right\}_{i=1}^K,\bbsigma_Z,\left\{\bbh_i\right\}_{i=1}^K,\bbh_Z\right)$
does not denote the $k$th user, instead it denotes the $(K-k+1)$th
user in line to be encoded. Rather, the secrecy rate of the $k$th
user is given by
\begin{align}
R_k=R_{\pi^{-1}(k)}^{\rm
DPC}\left(\pi,\left\{\bbk_i\right\}_{i=1}^K,\left\{\bbsigma_i\right\}_{i=1}^K,\bbsigma_Z,\left\{\bbh_i\right\}_{i=1}^K,\bbh_Z\right)
\end{align}
when dirty-paper coding with stochastic encoding is used with an
encoding order of $\pi$.

We define the following region.
\begin{eqnarray}
&&\mathcal{R}^{\rm
DPC}\left(\pi,\bbs,\left\{\bbsigma_i\right\}_{i=1}^K,\bbsigma_Z,\left\{\bbh_i\right\}_{i=1}^K,\bbh_Z\right)\nonumber
\\
&&\hspace{0.0cm}=\left\{(R_1,\ldots,R_K)\left|
\begin{array}{rcl}
&R_k=R_{\pi^{-1}(k)}^{\rm
DPC}\left(\pi,\left\{\bbk_i\right\}_{i=1}^K,\left\{\bbsigma_i\right\}_{i=1}^K,\bbsigma_Z,\left\{\bbh_i\right\}_{i=1}^K,\bbh_Z\right),\\
&~k=1,\ldots,K,\textrm{ for some }\left\{\bbk_i\right\}_{i=1}^K
\textrm{ such
that }\bbk_i\succeq 0,\\
&~i=1,\ldots,K, \textrm{ and }\sum_{i=1}^K\bbk_i \preceq\bbs
\end{array}\right.\right\}\nonumber\\
\end{eqnarray}
The secrecy capacity region of the general Gaussian MIMO broadcast
channel is given by the following theorem.
\begin{Theo}
\label{theorem_main_result_general_mimo} The secrecy capacity
region of the general Gaussian MIMO multi-receiver wiretap channel
is given by the convex closure of the following union
\begin{align}
\bigcup_{\pi\in\Pi} \mathcal{R}^{\rm
DPC}\left(\pi,\bbs,\left\{\bbsigma_i\right\}_{i=1}^K,\bbsigma_Z,\left\{\bbh_i\right\}_{i=1}^K,\bbh_Z\right)
\end{align}
where $\Pi$ is the set of all possible one-to-one permutations on
$\{1,\ldots,K\}$.
\end{Theo}

\subsection{Proof of
Theorem~\ref{theorem_main_result_general_mimo}}
\label{proof_of_theorem_main_result_general_mimo}

Achievability of the region given in
Theorem~\ref{theorem_main_result_general_mimo} can be shown by
following the achievability proof of
Theorem~\ref{theorem_main_result_aligned} given in
Section~\ref{sec:ach_theorem_main_result_aligned}, hence it is
omitted. For the converse, we basically use the ideas presented in
\cite{Shamai_MIMO,Tie_Liu_MIMO_WT}. Following Section V-B
of~\cite{Shamai_MIMO}, we can construct an equivalent channel
which has the same secrecy capacity region as the original channel
defined in
(\ref{general_mimo_original_def1})-(\ref{general_mimo_original_def2}).
In this constructed equivalent channel, all receivers, including
the eavesdropper, and the transmitter have the same number of
antennas, which is $t$,
\begin{align}
\hat{\bby}_k&=\hat{\bbh}_k\bbx +\hat{\bbn}_k,\quad k=1,\ldots,K \label{general_mimo_equivalent_def1}\\
\hat{\bbz}&=\hat{\bbh}_Z \bbx +\hat{\bbn}_Z
\label{general_mimo_equivalent_def2}
\end{align}
where $\hat{\bbh}_k=\hat{\bblambda}_k \bbv_k$, $\bbv_k$ is a
$t\times t$ orthonormal matrix, and $\hat{\bblambda}_k$ is a
$t\times t$ diagonal matrix whose first $(t-\hat{r}_k)$ diagonal
entries are zero, and the rest of the diagonal entries are
strictly positive. Here, $\hat{r}_k$ is the rank of the original
channel gain matrix, $\bbh_k$. The noise covariance matrix of the
Gaussian random vector $\hat{\bbn}_k$ is given by
$\hat{\bbsigma}_k$ which has the following block diagonal form
\begin{align}
\hat{\bbsigma}_k= \left[
\begin{array}{cc}
\hat{\bbsigma}^A_k & \bzero \\
\bzero & \hat{\bbsigma}^B_k
\end{array}
\right] \label{noise_in_equivalent_of_general_users}
\end{align}
where $\hat{\bbsigma}^A_k$ is of size $(t-\hat{r}_k)\times
(t-\hat{r}_k)$, and $\hat{\bbsigma}^B_k$ is of size
$\hat{r}_k\times \hat{r}_k$.

Similar notations hold for the eavesdropper's observation
$\hat{\bbz}$ as well. In particular,
$\hat{\bbh}_Z=\hat{\bblambda}_Z\bbv_Z$ where $\bbv_Z$ is a
$t\times t$ orthonormal matrix, and  $\hat{\bblambda}_Z$ is a
$t\times t$ diagonal matrix whose first $(t-\hat{r}_Z)$ diagonal
entries are zero, and the rest of the diagonal entries are
strictly positive. Here, $\hat{r}_Z$ is the rank of the original
channel gain matrix of the eavesdropper, $\bbh_Z$. The covariance
matrix of the Gaussian random vector $\hat{\bbn}_Z$ is given by
$\hat{\bbsigma}_Z$ which has the following block diagonal form
\begin{align}
\hat{\bbsigma}_Z= \left[
\begin{array}{cc}
\hat{\bbsigma}^A_Z & \bzero \\
\bzero & \hat{\bbsigma}^B_Z
\end{array}
\right] \label{noise_in_equivalent_of_general_eavesdropper}
\end{align}
where $\hat{\bbsigma}_Z^A$ is of size $(t-\hat{r}_Z)\times
(t-\hat{r}_Z)$ and $\hat{\bbsigma}_Z^B$ is of size $\hat{r}_Z
\times \hat{r}_Z$. Since this new channel in
(\ref{general_mimo_equivalent_def1})-(\ref{general_mimo_equivalent_def2})
can be constructed from the original channel in
(\ref{general_mimo_original_def1})-(\ref{general_mimo_original_def2})
through invertible transformations~\cite{Shamai_MIMO}, both have
the same secrecy capacity region. Moreover, these transformations
preserve the dirty-paper coding region as well, i.e.,
\begin{align}
\lefteqn{R_k^{\rm
DPC}\left(\pi,\left\{\bbk_i\right\}_{i=1}^K,\left\{\bbsigma_i\right\}_{i=1}^K,\bbsigma_Z,\left\{\bbh_i\right\}_{i=1}^K,\bbh_Z\right)}\nonumber\\
&=\frac{1}{2}\log \frac{\left|\bbh_{\pi(k)}\left(\sum_{i=1}^k
\bbk_{\pi(i)}\right)\bbh_{\pi(k)}^\top+\bbsigma_{\pi(k)}\right|}{\left|\bbh_{\pi(k)}\left(\sum_{i=1}^{k-1}
\bbk_{\pi(i)}\right)\bbh_{\pi(k)}^\top+\bbsigma_{\pi(k)}\right|} -
\frac{1}{2}\log\frac{\left|\bbh_{Z}\left(\sum_{i=1}^k
\bbk_{\pi(i)}\right)\bbh_Z^\top+\bbsigma_Z\right|}{\left|\bbh_{Z}\left(\sum_{i=1}^{k-1}
\bbk_{\pi(i)}\right)\bbh_{Z}^\top+\bbsigma_Z\right|} \nonumber \\
&=\frac{1}{2}\log
\frac{\left|\hat{\bbh}_{\pi(k)}\left(\sum_{i=1}^k
\bbk_{\pi(i)}\right)\hat{\bbh}_{\pi(k)}^\top+\hat{\bbsigma}_{\pi(k)}\right|}{\left|\hat{\bbh}_{\pi(k)}\left(\sum_{i=1}^{k-1}
\bbk_{\pi(i)}\right)\hat{\bbh}_{\pi(k)}^\top+\hat{\bbsigma}_{\pi(k)}\right|}
- \frac{1}{2}\log\frac{\left|\hat{\bbh}_{Z}\left(\sum_{i=1}^k
\bbk_{\pi(i)}\right)\hat{\bbh}_Z^\top+\hat{\bbsigma}_Z\right|}{\left|\hat{\bbh}_{Z}\left(\sum_{i=1}^{k-1}
\bbk_{\pi(i)}\right)\hat{\bbh}_{Z}^\top+\hat{\bbsigma}_Z\right|}
,\nonumber\\
&\hspace{12cm} k=1,\ldots,K
\end{align}

We now define another channel which does not have the same secrecy
capacity region or the dirty paper coding region as the original
channel:
\begin{align}
\bar{\bby}_k&=\bar{\bbh}_k\bbx +\hat{\bbn}_k,\quad k=1,\ldots,K \label{general_mimo_alpha_def1}\\
\bar{\bbz}&=\bar{\bbh}_Z \bbx +\hat{\bbn}_Z
\label{general_mimo_alpha_def2}
\end{align}
where
$\bar{\bbh}_k=\left(\hat{\bblambda}_k+\alpha\hat{\bbi}_k\right)\bbv_k$
and $\alpha>0$, and $\hat{\bbi}_k$ is a $t\times t$ diagonal
matrix whose first $(t-\hat{r}_k)$ diagonal entries are 1, and the
rest of the diagonal entries are zero. Similarly,
$\bar{\bbh}_Z=\left(\hat{\bblambda}_Z+\alpha\hat{\bbi}_Z\right)\bbv_Z$,
where $\hat{\bbi}_Z$ is a $t\times t$ diagonal matrix whose first
$(t-\hat{r}_Z)$ diagonal entries are 1, and the rest are zero. We
note that $\{\bar{\bbh}_k\}_{k=1}^K,\bar{\bbh}_Z$ are invertible,
hence the channel defined by
(\ref{general_mimo_alpha_def1})-(\ref{general_mimo_alpha_def2})
can be considered as an aligned Gaussian MIMO multi-receiver
wiretap channel. Thus, since it is an aligned Gaussian MIMO
multi-receiver wiretap channel, its secrecy capacity region is
given by Theorem~\ref{theorem_main_result_aligned}.

We now show that as $\alpha\rightarrow 0$, the secrecy capacity
region of the channel described by
(\ref{general_mimo_alpha_def1})-(\ref{general_mimo_alpha_def2})
converges to a region that includes the secrecy capacity region of
the original channel in
(\ref{general_mimo_original_def1})-(\ref{general_mimo_original_def2}).
Since the original channel in
(\ref{general_mimo_original_def1})-(\ref{general_mimo_original_def2})
and the channel in
(\ref{general_mimo_equivalent_def1})-(\ref{general_mimo_equivalent_def2})
have the same secrecy capacity region and the dirty-paper coding
region, checking that the secrecy capacity region of the channel
described by
(\ref{general_mimo_alpha_def1})-(\ref{general_mimo_alpha_def2})
converges, as $\alpha\rightarrow 0$, to a region that includes the
secrecy capacity region of the channel described by
(\ref{general_mimo_equivalent_def1})-(\ref{general_mimo_equivalent_def2}),
is sufficient. To this end, consider an arbitrary
$(2^{nR_1},\ldots,2^{nR_K},n)$ code which can be transmitted with
vanishingly small probability of error and in perfect secrecy when
it is used in the channel given in
(\ref{general_mimo_equivalent_def1})-(\ref{general_mimo_equivalent_def2}).
We will show that the same code can also be transmitted with
vanishingly small probability of error and in perfect secrecy when
it is used in the channel given in
(\ref{general_mimo_alpha_def1})-(\ref{general_mimo_alpha_def2}) as
$\alpha\rightarrow 0$. This will imply that the secrecy capacity
region of the channel given in
(\ref{general_mimo_alpha_def1})-(\ref{general_mimo_alpha_def2})
converges to a region that includes the secrecy capacity region of
the channel given in
(\ref{general_mimo_equivalent_def1})-(\ref{general_mimo_equivalent_def2}).
We first note that
\begin{align}
\bar{\bby}_k &= \left(\hat{\bblambda}_k+\alpha
\hat{\bbi}_k\right)\bbv_k\bbx+\hat{\bbn}_k \\
&= \left[
\begin{array}{cc}
\alpha \hat{\bbi}^A_k \bbv_k \bbx \\
\hat{\bblambda}^B_k\bbv_k \bbx
\end{array}
\right]+ \left[
\begin{array}{cc}
\hat{\bbn}^A_k  \\
\hat{\bbn}^B_k
\end{array}
\right] \\
&=\left[
\begin{array}{cc}
 \bar{\bby}^A_k  \\
\bar{\bby}^B_k
\end{array}
\right],\qquad \qquad \qquad \qquad  k=1,\ldots,K
\end{align}
where $\hat{\bbi}^A_k$ contains the first $(t-\hat{r}_k)$ rows of
$\hat{\bbi}_k$, and $\hat{\bblambda}^B_k$ contains the last
$\hat{r}_k$ rows of $\hat{\bblambda}_k$. $\hat{\bbn}_k^A$ is a
Gaussian random vector that contains the first $(t-\hat{r}_k)$
entries of $\hat{\bbn}_k$, and $\hat{\bbn}_k^B$ is a vector that
contains the last $\hat{r}_k$ entries. The covariance matrices of
$\hat{\bbn}^A_k,\hat{\bbn}^B_k$ are
$\hat{\bbsigma}^A_k,\hat{\bbsigma}^B_k$, respectively, and
$\hat{\bbn}_k^A$ and $\hat{\bbn}_k^B$  are independent as can be
observed through (\ref{noise_in_equivalent_of_general_users}).
Similarly, we can write
\begin{align}
\hat{\bby}_k &= \hat{\bblambda}_k \bbv_k\bbx+\hat{\bbn}_k \\
&= \left[
\begin{array}{cc}
\bzero \\
\hat{\bblambda}^B_k\bbv_k \bbx
\end{array}
\right]+ \left[
\begin{array}{cc}
\hat{\bbn}^A_k  \\
\hat{\bbn}^B_k
\end{array}
\right] \\
&=\left[
\begin{array}{cc}
 \hat{\bby}^A_k  \\
\hat{\bby}^B_k
\end{array}
\right],\qquad \qquad \qquad \qquad  k=1,\ldots,K
\end{align}
We note that $\bar{\bby}_k^B=\hat{\bby}_k^B,~k=1,\ldots,K$, thus
we have
\begin{align}
\bbx\rightarrow \bar{\bby}_k \rightarrow \hat{\bby}_k,\quad
k=1,\ldots,K
\end{align}
which ensures the any message rate that is decodable by the $k$th
user of the channel given in
(\ref{general_mimo_equivalent_def1})-(\ref{general_mimo_equivalent_def2})
is also decodable by the $k$th user of the channel given in
(\ref{general_mimo_alpha_def1})-(\ref{general_mimo_alpha_def2}).
Thus, any $(2^{nR_1},\ldots,2^{nR_K},n)$ code which can be
transmitted with vanishingly small probability of error in the
channel defined by
(\ref{general_mimo_equivalent_def1})-(\ref{general_mimo_equivalent_def2})
can be transmitted with vanishingly small probability of error in
the channel defined by
(\ref{general_mimo_alpha_def1})-(\ref{general_mimo_alpha_def2}) as
well.

We now check the secrecy constraints. To this end, we note that
\begin{align}
\bar{\bbz} &= \left(\hat{\bblambda}_Z+\alpha
\hat{\bbi}_Z\right)\bbv_Z\bbx+\hat{\bbn}_Z \\
&= \left[
\begin{array}{cc}
\alpha \hat{\bbi}^A_Z \bbv_Z \bbx \\
\hat{\bblambda}^B_Z\bbv_Z \bbx
\end{array}
\right]+ \left[
\begin{array}{cc}
\hat{\bbn}^A_Z  \\
\hat{\bbn}^B_Z
\end{array}
\right] \\
&=\left[
\begin{array}{cc}
 \bar{\bbz}^A \\
\bar{\bbz}^B
\end{array}
\right]
\end{align}
where $\hat{\bbi}^A_Z$ contains the first $(t-\hat{r}_Z)$ rows of
$\hat{\bbi}_Z$, and $\hat{\bblambda}_Z^B$ contains the last
$\hat{r}_Z$ rows of $\hat{\bblambda}_Z$. $\hat{\bbn}_Z^A$ is a
Gaussian random vector that contains the first $t-\hat{r}_Z$
entries of $\hat{\bbn}_Z$, and $\hat{\bbn}_Z^B$ is a vector that
contains the last $\hat{r}_Z$ entries. The covariance matrices of
$\hat{\bbn}_Z^A,\hat{\bbn}_Z^B$ are
$\hat{\bbsigma}_Z^A,\hat{\bbsigma}_Z^B$, respectively, and
$\hat{\bbn}_Z^A$ and $\hat{\bbn}_Z^B$ are independent as can be
observed through
(\ref{noise_in_equivalent_of_general_eavesdropper}). Similarly, we
can write
\begin{align}
\hat{\bbz} &= \hat{\bblambda}_Z\bbv_Z\bbx+\hat{\bbn}_Z \\
&= \left[
\begin{array}{cc}
\bzero \\
\hat{\bblambda}^B_Z\bbv_Z \bbx
\end{array}
\right]+ \left[
\begin{array}{cc}
\hat{\bbn}^A_Z  \\
\hat{\bbn}^B_Z
\end{array}
\right] \\
&=\left[
\begin{array}{cc}
 \hat{\bbz}^A \\
\hat{\bbz}^B
\end{array}
\right]
\end{align}
We note that $\bar{\bbz}^B=\hat{\bbz}^B$, and thus we have
\begin{align}
\bbx\rightarrow \bar{\bbz}\rightarrow \hat{\bbz}
\end{align}
We now show that any $(2^{nR_1},\ldots,2^{nR_K})$ code that
achieves the perfect secrecy rates $(R_1,\ldots,\break R_K)$ in
the channel given in
(\ref{general_mimo_equivalent_def1})-(\ref{general_mimo_equivalent_def2})
also achieves the same perfect secrecy rates in the channel given
in (\ref{general_mimo_alpha_def1})-(\ref{general_mimo_alpha_def2})
when $\alpha\rightarrow 0$. To this end, let $\mathcal{S}$ be a
non-empty subset of $\{1,\ldots,K\}$. We consider the following
equivocation
\begin{align}
H(W_{\mathcal{S}}|\bar{\bbz}^n)&=H(W_{\mathcal{S}})-I(W_{\mathcal{S}};\bar{\bbz}^n)\\
&=H(W_{\mathcal{S}}|\hat{\bbz}^n)+I(W_{\mathcal{S}};\hat{\bbz}^n)-I(W_{\mathcal{S}};\bar{\bbz}^n)\\
&=H(W_{\mathcal{S}}|\hat{\bbz}^{A,n},\hat{\bbz}^{B,n})+I(W_{\mathcal{S}};\hat{\bbz}^{A,n},\hat{\bbz}^{B,n})-I(W_{\mathcal{S}};\bar{\bbz}^{A,n},\bar{\bbz}^{B,n})\\
&=H(W_{\mathcal{S}}|\hat{\bbz}^{A,n},\hat{\bbz}^{B,n})+I(W_{\mathcal{S}};\hat{\bbz}^{B,n})-I(W_{\mathcal{S}};\bar{\bbz}^{A,n},\hat{\bbz}^{B,n})
\label{alpha_channel_achieves_same_secrecy_1}\\
&=H(W_{\mathcal{S}}|\hat{\bbz}^{A,n},\hat{\bbz}^{B,n})-I(W_{\mathcal{S}};\bar{\bbz}^{A,n}|\hat{\bbz}^{B,n})
\label{alpha_channel_achieves_same_secrecy_2}
\end{align}
where (\ref{alpha_channel_achieves_same_secrecy_1}) follows from
the facts that $W_{\mathcal{S}}$ and
$\hat{\bbz}^{A,n}=\hat{\bbn}^{A,n}$ are independent, and
$\bar{\bbz}^{B,n}=\hat{\bbz}^{B,n}$. We now bound the mutual
information term in (\ref{alpha_channel_achieves_same_secrecy_2})
\begin{align}
I(W_{\mathcal{S}};\bar{\bbz}^{A,n}|\hat{\bbz}^{B,n})&\leq
I(\bbx^n;\bar{\bbz}^{A,n}|\hat{\bbz}^{B,n}) \label{alpha_channel_achieves_same_secrecy_3} \\
& =
h(\bar{\bbz}^{A,n}|\hat{\bbz}^{B,n})-h(\bar{\bbz}^{A,n}|\hat{\bbz}^{B,n},\bbx^n)\\
&=
h(\bar{\bbz}^{A,n}|\hat{\bbz}^{B,n})-h(\bar{\bbz}^{A,n}|\bbx^n)\label{alpha_channel_achieves_same_secrecy_4}\\
&\leq
h(\bar{\bbz}^{A,n})-h(\bar{\bbz}^{A,n}|\bbx^n)\label{alpha_channel_achieves_same_secrecy_5}\\
&=I(\bbx^n;\bar{\bbz}^{A,n})\\
&\leq \sum_{i=1}^n I(\bbx_i;\bar{\bbz}^{A}_i)
\label{alpha_channel_achieves_same_secrecy_6} \\
&\leq \sum_{i=1}^n \max_{E\left[\bbx_i \bbx_i^\top\right]\preceq \bbs}I(\bbx_i;\bar{\bbz}^{A}_i) \\
&\leq \sum_{i=1}^n \frac{1}{2} \log \frac{\left|\alpha^2
\hat{\bbi}^{A}_Z\bbv_Z \bbs \bbv_Z^\top (\hat{\bbi}_Z^A
)^\top+\hat{\bbsigma}_Z^A
\right|}{\left|\hat{\bbsigma}_Z^A\right|}
\label{alpha_channel_achieves_same_secrecy_7}\\
&= \frac{n}{2} \log \frac{\left|\alpha^2 \hat{\bbi}^{A}_Z\bbv_Z
\bbs \bbv_Z^\top (\hat{\bbi}_Z^A )^\top+\hat{\bbsigma}_Z^A
\right|}{\left|\hat{\bbsigma}_Z^A\right|}
\label{alpha_channel_achieves_same_secrecy_8}
\end{align}
where (\ref{alpha_channel_achieves_same_secrecy_3}) follows from
the Markov chain $W_{\mathcal{S}}\rightarrow \bbx^n \rightarrow
(\bar{\bbz}^{A,n},\hat{\bbz}^{B,n})$,
(\ref{alpha_channel_achieves_same_secrecy_4}) is due to the Markov
chain $\bar{\bbz}^{A,n}\rightarrow \bbx^n \rightarrow
\hat{\bbz}^{B,n}$, (\ref{alpha_channel_achieves_same_secrecy_5})
comes from the fact that conditioning cannot increase entropy,
(\ref{alpha_channel_achieves_same_secrecy_6}) is a consequence of
the fact that channel is memoryless,
(\ref{alpha_channel_achieves_same_secrecy_7}) is due to the fact
that subject to a covariance constraint, Gaussian distribution
maximizes the differential entropy. Thus, plugging
(\ref{alpha_channel_achieves_same_secrecy_8}) into
(\ref{alpha_channel_achieves_same_secrecy_2}) yields
\begin{align}
\frac{1}{n}H(W_{\mathcal{S}}|\bar{\bbz}^n)\geq
\frac{1}{n}H(W_{\mathcal{S}}|\hat{\bbz}^{n}) -\frac{1}{2} \log
\frac{\left|\alpha^2 \hat{\bbi}^{A}_Z\bbv_Z \bbs \bbv_Z^\top
(\hat{\bbi}_Z^A )^\top+\hat{\bbsigma}_Z^A
\right|}{\left|\hat{\bbsigma}_Z^A\right|}
\end{align}
which implies that
\begin{align}
\lim_{n\rightarrow\infty}\frac{1}{n}H(W_{\mathcal{S}}|\bar{\bbz}^n)&\geq
\lim_{n\rightarrow
\infty}\frac{1}{n}H(W_{\mathcal{S}}|\hat{\bbz}^{n})-\lim_{\alpha\rightarrow
0}\frac{1}{2} \log \frac{\left|\alpha^2 \hat{\bbi}^{A}_Z\bbv_Z
\bbs \bbv_Z^\top (\hat{\bbi}_Z^A )^\top+\hat{\bbsigma}_Z^A
\right|}{\left|\hat{\bbsigma}_Z^A\right|}\\
&=\lim_{n\rightarrow
\infty}\frac{1}{n}H(W_{\mathcal{S}}|\hat{\bbz}^{n})\label{alpha_channel_achieves_same_secrecy_9}\\
&\geq \sum_{k\in\mathcal{S}}R_k
\label{alpha_channel_achieves_same_secrecy_10}
\end{align}
where (\ref{alpha_channel_achieves_same_secrecy_9}) follows from
the fact that $\log |\alpha^2 \bba+\bbb|$ is continuous in
$\alpha$ for positive definite matrices $\bba,\bbb$, and
(\ref{alpha_channel_achieves_same_secrecy_10}) comes from our
assumption that the codebook under consideration achieves perfect
secrecy in the channel given in
(\ref{general_mimo_equivalent_def1})-(\ref{general_mimo_equivalent_def2}).
Thus, we have shown that if a codebook achieves the perfect
secrecy rates $(R_1,\ldots,R_K)$ in the channel defined by
(\ref{general_mimo_equivalent_def1})-(\ref{general_mimo_equivalent_def2}),
then it also achieves the same perfect secrecy rates in the
channel defined by
(\ref{general_mimo_alpha_def1})-(\ref{general_mimo_alpha_def2}) as
$\alpha\rightarrow 0$. Thus, the secrecy capacity region of the
latter channel converges to a region that includes the secrecy
capacity region of the channel in
(\ref{general_mimo_equivalent_def1})-(\ref{general_mimo_equivalent_def2}),
and also the secrecy capacity region of the original channel in
(\ref{general_mimo_original_def1})-(\ref{general_mimo_original_def2}).
Since the channel in
(\ref{general_mimo_alpha_def1})-(\ref{general_mimo_alpha_def2}) is
an aligned channel, its secrecy capacity region is given by
Theorem~\ref{theorem_main_result_aligned}, and it is equal to the
dirty-paper coding region. Thus, to find the region that the
secrecy capacity region of the channel in
(\ref{general_mimo_alpha_def1})-(\ref{general_mimo_alpha_def2})
converges to as $\alpha\rightarrow 0$, it is sufficient to
consider the region which the dirty-paper coding region converges
to as $\alpha\rightarrow 0$. For that purpose, pick the $k$th
user, and the identity encoding order, i.e.,
$\pi(k)=k,~k=1,\ldots,K$. The corresponding secrecy rate is
\begin{align}
\lefteqn{\frac{1}{2}\log
\frac{\left|\bar{\bbh}_{\pi(k)}\left(\sum_{i=1}^k
\bbk_{\pi(i)}\right)\bar{\bbh}_{\pi(k)}^\top+\hat{\bbsigma}_{\pi(k)}\right|}{\left|\bar{\bbh}_{\pi(k)}\left(\sum_{i=1}^{k-1}
\bbk_{\pi(i)}\right)\bar{\bbh}_{\pi(k)}^\top+\hat{\bbsigma}_{\pi(k)}\right|}
- \frac{1}{2}\log\frac{\left|\bar{\bbh}_{Z}\left(\sum_{i=1}^k
\bbk_{\pi(i)}\right)\bar{\bbh}_Z^\top+\hat{\bbsigma}_Z\right|}{\left|\bar{\bbh}_{Z}\left(\sum_{i=1}^{k-1}
\bbk_{\pi(i)}\right)\bar{\bbh}_{Z}^\top+\hat{\bbsigma}_Z\right|}}\nonumber
\\
&=\frac{1}{2}\log \frac{\left|\left(\hat{\bbh}_{\pi(k)}+\alpha
\hat{\bbi}_{\pi(k)}\bbv_{\pi(k)}\right) \left(\sum_{i=1}^k
\bbk_{\pi(i)}\right)\left(\hat{\bbh}_{\pi(k)}+\alpha
\hat{\bbi}_{\pi(k)}\bbv_{\pi(k)}\right)^\top+\hat{\bbsigma}_{\pi(k)}\right|}
{\left|\left(\hat{\bbh}_{\pi(k)}+\alpha
\hat{\bbi}_{\pi(k)}\bbv_{\pi(k)}\right)\left(\sum_{i=1}^{k-1}
\bbk_{\pi(i)}\right)\left(\hat{\bbh}_{\pi(k)}+\alpha
\hat{\bbi}_{\pi(k)}\bbv_{\pi(k)}\right)^\top+
\hat{\bbsigma}_{\pi(k)}\right|}
\nonumber\\
&\quad-\frac{1}{2}\log \frac{\left|\left(\hat{\bbh}_{Z}+\alpha
\hat{\bbi}_Z \bbv_Z\right)\left(\sum_{i=1}^k
\bbk_{\pi(i)}\right)\left(\hat{\bbh}_{Z}+\alpha \hat{\bbi}_Z
\bbv_Z\right)^\top
+\hat{\bbsigma}_{Z}\right|}{\left|\left(\hat{\bbh}_{Z}+\alpha
\hat{\bbi}_Z \bbv_Z\right)\left(\sum_{i=1}^{k-1}
\bbk_{\pi(i)}\right)\left(\hat{\bbh}_{Z}+\alpha \hat{\bbi}_Z
\bbv_Z\right)^\top +\hat{\bbsigma}_{Z}\right|}
\end{align}
which converges to
\begin{align}
\frac{1}{2}\log \frac{\left|\hat{\bbh}_{\pi(k)}\left(\sum_{i=1}^k
\bbk_{\pi(i)}\right)\hat{\bbh}_{\pi(k)}^\top+\hat{\bbsigma}_{\pi(k)}\right|}{\left|\hat{\bbh}_{\pi(k)}\left(\sum_{i=1}^{k-1}
\bbk_{\pi(i)}\right)\hat{\bbh}_{\pi(k)}^\top
+\hat{\bbsigma}_{\pi(k)}\right|} -\frac{1}{2}\log
\frac{\left|\hat{\bbh}_{Z}\left(\sum_{i=1}^k
\bbk_{\pi(i)}\right)\hat{\bbh}_{Z}^\top+
\hat{\bbsigma}_{Z}\right|}{\left|\hat{\bbh}_{Z}\left(\sum_{i=1}^{k-1}
\bbk_{\pi(i)}\right)\hat{\bbh}_{Z}^\top+\hat{\bbsigma}_{Z}\right|}
\label{alpha_channel_achieves_same_secrecy_11}
\end{align}
as $\alpha\rightarrow 0$ due to the continuity of $\log |\cdot|$
in positive semi-definite matrices. Moreover,
(\ref{alpha_channel_achieves_same_secrecy_11}) is equal to
\begin{align}
\frac{1}{2}\log \frac{\left|\bbh_{\pi(k)}\left(\sum_{i=1}^k
\bbk_{\pi(i)}\right)\bbh_{\pi(k)}^\top+\bbsigma_{\pi(k)}\right|}{\left|\bbh_{\pi(k)}\left(\sum_{i=1}^{k-1}
\bbk_{\pi(i)}\right)\bbh_{\pi(k)}^\top +\bbsigma_{\pi(k)}\right|}
-\frac{1}{2}\log \frac{\left|\bbh_{Z}\left(\sum_{i=1}^k
\bbk_{\pi(i)}\right)\bbh_{Z}^\top+
\bbsigma_{Z}\right|}{\left|\bbh_{Z}\left(\sum_{i=1}^{k-1}
\bbk_{\pi(i)}\right)\bbh_{Z}^\top+\bbsigma_{Z}\right|}
\end{align}
which implies that the secrecy capacity region of the general
Gaussian MIMO multi-receiver wiretap channel is given by the
dirty-paper coding region, completing the proof.

\section{Conclusions}

We characterized the secrecy capacity region of the Gaussian MIMO
multi-receiver wiretap channel. We showed that it is achievable
with a variant of dirty-paper coding with Gaussian signals. Before
reaching this result, we first visited the scalar case, and showed
the necessity of a new proof technique for the converse. In
particular, we showed that the extensions of existing converses
for the Gaussian scalar broadcast channels fall short of resolving
the ambiguity regarding the auxiliary random variables. We showed
that, unlike the stand-alone use of the entropy-power
inequality~\cite{Stam,Blachman}, the use of the relationships
either between the MMSE and the mutual information or between the
Fisher information and the differential entropy resolves this
ambiguity. Extending this methodology to degraded vector channels,
we found the secrecy capacity region of the degraded Gaussian MIMO
multi-receiver wiretap channel. Once we obtained the secrecy
capacity region of the degraded MIMO channel, we generalized it to
arbitrary channels by using the channel enhancement method and
some limiting arguments as in~\cite{Shamai_MIMO,Tie_Liu_MIMO_WT}.

\appendix

\appendixpage

\section{Proof of Lemma~\ref{lemma_conditional_stein_identity}}
\label{proof_of_lemma_conditional_stein_identity} Let
$\rho_i(\bx|\bu)=\frac{\partial \log f(\bx|\bu)}{\partial x_i}$,
i.e., the $i$th component of $\brho(\bx|\bu)$. Then, we have
\begin{align}
E\left[g(\bbx)\rho_{i}(\bbx|\bbu)\right]&=\int g(\bx)
\frac{\frac{\partial f(\bx|\bu)}{\partial x_i}}{f(\bx|\bu)}
f(\bx,\bu)~d\bx ~d\bu\\
&=\int g(\bx) \frac{\partial f(\bx|\bu)}{\partial x_i} f(\bu)~d\bx
~d\bu \\
&=\int \left[\int_{-\infty}^{+\infty}g(\bx) \frac{\partial
f(\bx|\bu)}{\partial x_i} dx_i \right] f(\bu)~d\bx^{-} ~d\bu
\label{cond_stein_proof_1}
\end{align}
where $d\bx^{-}=dx_1\ldots dx_{i-1}dx_{i+1}\ldots dx_{n}$. The
inner integral can be evaluated using integration by parts as
\begin{align}
\int_{-\infty}^{+\infty}g(\bx) \frac{\partial f(\bx|\bu)}{\partial
x_i} dx_i &=
\big[g(\bx)f(\bx|\bu)\big]\Big|_{x_{i}=-\infty}^{+\infty}-
\int_{-\infty}^{+\infty}f(\bx|\bu) \frac{\partial g(\bx)}{\partial
x_i} dx_i \\
&= -\int_{-\infty}^{+\infty}f(\bx|\bu) \frac{\partial
g(\bx)}{\partial x_i} dx_i
\label{assumption_of_the_lemma_implies_1}
\end{align}
where (\ref{assumption_of_the_lemma_implies_1}) comes from the
assumption in (\ref{assumption_of_the_lemma}). Plugging
(\ref{assumption_of_the_lemma_implies_1}) into
(\ref{cond_stein_proof_1}) yields
\begin{align}
E\left[g(\bbx)\rho_{i}(\bbx|\bbu)\right]&=-\int \frac{\partial
g(\bx)}{\partial x_i} f(\bx,\bu)~d\bx ~d\bu \\
&= -E\left[\frac{\partial g(\bx)}{\partial x_i}\right]
\end{align}
which concludes the proof.

\section{Proof of Lemma~\ref{lemma_cond_stein_implies}}
\label{proof_of_lemma_cond_stein_implies}

Let $\rho_i(\bx|\bu)=\frac{\partial \log f(\bx|\bu)}{\partial
x_i}$, i.e., the $i$th component of $\brho(\bx|\bu)$. Then, we
have
\begin{align}
E\left[\brho(\bbx|\bbu)|\bbu=\bu\right]&=\int \frac{\frac{\partial
f(\bx|\bu)}{\partial x_i}}{f(\bx|\bu)}f(\bx|\bu) d\bx \\
&=\int \left[\int_{-\infty}^{+\infty}\frac{\partial
f(\bx|\bu)}{\partial x_i}dx_i\right] d\bx^-
\end{align}
where $d\bx^-=dx_1\ldots dx_{i-1}dx_{i+1}\ldots dx_n$. The inner
integral is
\begin{align}
\int_{-\infty}^{+\infty}\frac{\partial f(\bx|\bu)}{\partial
x_i}dx_i=f(\bx|\bu)\Big|_{x_i=-\infty}^{+\infty}=0
\end{align}
since $f(\bx|\bu)$ is a valid probability density function. This
completes the proof of the first part. For the second part, we
have
\begin{align}
E\left[g(\bbu)\brho(\bbx|\bbu)\right]=E\left[g(\bbu)E\big[\brho(\bbx|\bbu)|\bbu=\bu\big]\right]=0
\end{align}
where the second equality follows from the fact that the inner
expectation is zero as the first part of this lemma states. The
last part of the lemma follows by selecting
$g(\bbu)=E\left[\bbx|\bbu\right]$ in the second part of this
lemma.

\section{Proof of Lemma~\ref{lemma_cond_conv_identity}}

\label{proof_of_lemma_cond_conv_identity}

Throughout this proof, the subscript of $f$ will denote the random
vector for which $f$ is the density. For example, $f_{X}(\bx|\bu)$
is the conditional density of $\bbx$. We first note that
\begin{align}
f_{W}(\bw|\bu)&=\int f_{X,W}(\bx,\bw|\bu)d\bx =\int
f_{X}(\bx|\bu)f_{Y}(\bw-\bx|\bu)d\bx
\label{cond_independence_imply}
\end{align}
where the second equality is due to the conditional independence
of $\bbx$ and $\bby$ given $\bbu$. Differentiating both sides of
(\ref{cond_independence_imply}), we get
\begin{align}
\frac{\partial f_{W}(\bw|\bu)}{\partial w_i}&=\int
f_{X}(\bx|\bu)\frac{\partial f_{Y}(\bw-\bx|\bu)}{\partial
w_i}d\bx\\
&=-\int f_{X}(\bx|\bu)\frac{\partial f_{Y}(\bw-\bx|\bu)}{\partial
x_i}d\bx \label{change_of_variables} \\
&=\big[-f_{X}(\bx|\bu)
f_{Y}(\bw-\bx|\bu)\big]\Big|_{x_i=-\infty}^{\infty}+\int
f_{Y}(\bw-\bx|\bu)\frac{\partial f_{X}(\bx|\bu)}{\partial x_i}d\bx
\label{proper_densities} \\
&=\int f_{Y}(\bw-\bx|\bu)\frac{\partial f_{X}(\bx|\bu)}{\partial
x_i}d\bx \label{cond_conv_proof_1}
\end{align}
where (\ref{change_of_variables}) is due to
\begin{align}
\frac{\partial f_{Y}(\bw-\bx|\bu)}{\partial w_i}&=\frac{\partial
f_{Y}(\bw-\bx|\bu)}{\partial (w_i-x_i)}  \frac{\partial
(w_i-x_i)}{\partial w_i} \\
&=-\frac{\partial f_{Y}(\bw-\bx|\bu)}{\partial (w_i-x_i)}
\frac{\partial
(w_i-x_i)}{\partial x_i} \\
&=-\frac{\partial f_{Y}(\bw-\bx|\bu)}{\partial x_i}
\end{align}
and (\ref{proper_densities}) follows from the fact that
$f_{X}(\bx|\bu),f_{Y}(\bw-\bx|\bu)$ vanish at infinity since they
are probability density functions. Using
(\ref{cond_conv_proof_1}), we get
\begin{align}
\brho_i(\bw|\bu)=\frac{\frac{\partial f_{W}(\bw|\bu)}{\partial
w_i}}{f_{W}(\bw|\bu)} &=\int \frac{
f_{Y}(\bw-\bx|\bu)}{f_{W}(\bw|\bu)}\frac{\partial
f_{X}(\bx|\bu)}{\partial x_i}d\bx\\
 &=\int \frac{f_{X}(\bx|\bu)
f_{Y}(\bw-\bx|\bu)}{f_{W}(\bw|\bu)}\frac{\frac{\partial
f_{X}(\bx|\bu)}{\partial x_i}}{f_{X}(\bx|\bu)}d\bx\\
 &=\int f_{X}(\bx|\bu,\bw)\frac{\frac{\partial
f_{X}(\bx|\bu)}{\partial x_i}}{f_{X}(\bx|\bu)}d\bx
\label{cond_conv_proof_2}\\
&=E\left[\frac{1}{f_{X}(\bx|\bu)}~~ \frac{\partial
f_{X}(\bx|\bu)}{\partial x_i}\Bigg|\bbw=\bw,\bbu=\bu \right]
\label{cond_conv_proof_3}
\end{align}
where (\ref{cond_conv_proof_2}) follows from the fact that
\begin{align}
f_{X}(\bx|\bu,\bw)=\frac{f_{X,W}(\bx,\bw|\bu)}{f_{W}(\bw|\bu)}=\frac{f_{X}(\bx|\bu)f_{Y}(\bw-\bx|\bu)}{f_{W}(\bw|\bu)}
\end{align}
Equation (\ref{cond_conv_proof_3}) implies
\begin{align}
\brho
(\bw|\bu)=E\left[\brho(\bbx|\bbu=\bu)|\bbw=\bw,\bbu=\bu\right]
\end{align}
and due to symmetry, we also have
\begin{align}
\brho
(\bw|\bu)=E\left[\brho(\bby|\bbu=\bu)|\bbw=\bw,\bbu=\bu\right]
\end{align}
which completes the proof.

\section{Proof of Lemma~\ref{lemma_cond_matrix_fii}}

\label{proof_of_lemma_cond_matrix_fii}

Let $\bbw=\bbx+\bby$. We have
\begin{align}
\bzero &\preceq E\left[\Big(\bba \brho (\bbx|\bbu)+(\bbi-\bba)\brho(\bby|\bbu)-\brho(\bbw|\bbu)\Big)\right.\nonumber\\
&\hspace{3cm}\left.\Big(\bba \brho (\bbx|\bbu)+(\bbi-\bba)\brho(\bby|\bbu)-\brho(\bbw|\bbu)\Big)^\top\right] \\
& =\bba  E\left[ \brho (\bbx|\bbu)\brho
(\bbx|\bbu)^\top\right]\bba^\top
+\bba  E\left[ \brho (\bbx|\bbu)\brho (\bby|\bbu)^\top\right](\bbi-\bba)^\top \nonumber \\
&\quad -\bba E\left[ \brho (\bbx|\bbu)\brho
(\bbw|\bbu)^\top\right]
+  (\bbi-\bba) E\left[\brho(\bby|\bbu)\brho(\bbx|\bbu)^\top\right]\bba^\top \nonumber\\
&\quad +(\bbi-\bba)
E\left[\brho(\bby|\bbu)\brho(\bby|\bbu)^\top\right](\bbi-\bba)^\top
-(\bbi-\bba) E\left[\brho(\bby|\bbu)\brho(\bbw|\bbu)^\top\right]\nonumber \\
&\quad -
E\left[\brho(\bbw|\bbu)\brho(\bbx|\bbu)^\top\right]\bba^\top
- E\left[\brho(\bbw|\bbu)\brho(\bby|\bbu)^\top\right](\bbi-\bba)^\top \nonumber \\
&\quad +E\left[\brho(\bbw|\bbu)\brho(\bbw|\bbu)^\top\right]
\label{proof_cond_matrix_fii_step_1}
\end{align}
We note that, from the definition of the conditional Fisher
information matrix, we have
\begin{align}
E\left[ \brho (\bbx|\bbu)\brho (\bbx|\bbu)^\top\right]&=\bbj(\bbx|\bbu)\label{def_of_cond_Fisher_imply_1}\\
E\left[ \brho (\bby|\bbu)\brho (\bby|\bbu)^\top\right]&=\bbj(\bby|\bbu)\label{def_of_cond_Fisher_imply_2}\\
E\left[ \brho (\bbw|\bbu)\brho
(\bbw|\bbu)^\top\right]&=\bbj(\bbw|\bbu)\label{def_of_cond_Fisher_imply_3}
\end{align}
Moreover, we have
\begin{align}
E\left[\brho(\bbx|\bbu)\brho(\bby|\bbu)^\top\right]&=\left(E\left[\brho(\bby|\bbu)\brho(\bbx|\bbu)^\top\right]\right)^\top\\
&=\left(E\left[E\left[\brho(\bbx|\bbu)\big|\bbu=\bu\right]E\left[\brho(\bby|\bbu)\big|\bbu=\bu\right]\right]\right)^\top\label{conditional_independence_imply}\\
&=\bzero \label{stein_implies_1}
\end{align}
where (\ref{conditional_independence_imply}) comes from the fact
that given $\bbu=\bu$, $\bbx$ and $\bby$ are conditionally
independent, and (\ref{stein_implies_1}) follows from the first
part of Lemma~\ref{lemma_cond_stein_implies}, namely
\begin{align}
E\left[\brho(\bbx|\bbu)\big|\bbu=\bu\right]=E\left[\brho(\bby|\bbu)\big|\bbu=\bu\right]=\bzero
\end{align}
Furthermore, we have
\begin{align}
E\left[\brho(\bbx|\bbu)\brho(\bbw|\bbu)^\top\right]&= E\left[E\left[\brho(\bbx|\bbu=\bu)\big|\bbw=\bw,\bbu=\bu\right]\brho(\bbw|\bbu)^\top\right]\\
&=E\left[\brho(\bbw|\bbu) \brho(\bbw|\bbu)^\top\right]\label{cond_conv_identity_implies}\\
&=\bbj(\bbw|\bbu) \label{def_of_cond_Fisher_imply}
\end{align}
where (\ref{cond_conv_identity_implies}) follows from
Lemma~\ref{lemma_cond_conv_identity}, and
(\ref{def_of_cond_Fisher_imply}) comes from the definition of the
conditional Fisher information matrix. Similarly, we also have
\begin{align}
E\left[\brho(\bby|\bbu)\brho(\bbw|\bbu)^\top\right]=
E\left[\brho(\bbw|\bbu)\brho(\bbx|\bbu)^\top\right]
=E\left[\brho(\bbw|\bbu)\brho(\bby|\bbu)^\top\right]=\bbj(\bbw|\bbu)
\label{cond_conv_identity_implies_2}
\end{align}
Thus, using
(\ref{def_of_cond_Fisher_imply_1})-(\ref{def_of_cond_Fisher_imply_3}),
(\ref{stein_implies_1}),
(\ref{def_of_cond_Fisher_imply})-(\ref{cond_conv_identity_implies_2})
in (\ref{proof_cond_matrix_fii_step_1}), we get
\begin{align}
\bzero &\preceq  \bba  \bbj(\bbx|\bbu)\bba^\top -\bba
\bbj(\bbw|\bbu) +(\bbi-\bba) \bbj(\bby|\bbu)(\bbi-\bba)^\top
-(\bbi-\bba) \bbj(\bbw|\bbu)\\
&\quad - \bbj(\bbw|\bbu)\bba^\top
-\bbj(\bbw|\bbu)(\bbi-\bba)^\top +\bbj(\bbw|\bbu) \\
&= \bba  \bbj(\bbx|\bbu)\bba^\top +(\bbi-\bba)
\bbj(\bby|\bbu)(\bbi-\bba)^\top - \bbj(\bbw|\bbu)
\end{align}
which completes the proof.

\section{Proof of Lemma~\ref{lemma_conditioning_increases_Fisher}}

\label{proof_of_lemma_conditioning_increases_Fisher}

Consider $\bbj(\bbx|\bbu)$
\begin{align}
\bbj(\bbx|\bbu)&=\bbj(\bbx|\bbu,\bbv) \label{Markov_chain_implies} \\
&= E\left[\nabla_{\bx} \log f(\bbx|\bbu,\bbv) \nabla_{\bx} \log f(\bbx|\bbu,\bbv)^\top\right] \\
&= E\left[\nabla_{\bx} \log f(\bbx,\bbu,\bbv) \nabla_{\bx} \log f(\bbx,\bbu,\bbv)^\top\right] \label{derivative_zero}\\
&= E\left[\big(\nabla_{\bx} \log f(\bbx,\bbv)+\nabla_{\bx} \log f(\bbu|\bbx,\bbv)\big)\right.\nonumber\\
&\hspace{4.5cm}\left.  \big(\nabla_{\bx} \log f(\bbx,\bbv)+\nabla_{\bx} \log f(\bbu|\bbx,\bbv)\big)^\top\right]\label{bayes_rule} \\
&= E\left[\nabla_{\bx} \log f(\bbx,\bbv)\nabla_{\bx} \log f(\bbx,\bbv)^\top\right]\nonumber\\
&\quad + E\left[\nabla_{\bx} \log f(\bbx,\bbv)  \nabla_{\bx} \log f(\bbu|\bbx,\bbv)^\top\right]\nonumber\\
&\quad+E\left[\nabla_{\bx} \log f(\bbu|\bbx,\bbv) \nabla_{\bx} \log f(\bbx,\bbv)^\top\right]\nonumber\\
&\quad + E\left[\nabla_{\bx} \log f(\bbu|\bbx,\bbv) \nabla_{\bx}
\log f(\bbu|\bbx,\bbv)^\top\right]
\label{conditioning_decreases_Fisher_step1}
\end{align}
where (\ref{Markov_chain_implies}) is due to the Markov chain
$\bbv\rightarrow \bbu\rightarrow \bbx$, (\ref{derivative_zero})
comes from the fact that
\begin{align}
\nabla_{\bx} \log f(\bx|\bu,\bv)&=\nabla_{\bx} \big(\log f(\bx,\bu,\bv)-\log f(\bu,\bv)\big)\\
&=\nabla_{\bx} \log f(\bx,\bu,\bv)
\end{align}
and (\ref{bayes_rule}) is due to the fact that
$f(\bx,\bu,\bv)=f(\bx,\bv)f(\bu|\bx,\bv)$. We note that
\begin{align}
\bbj(\bbx|\bbv)=E\left[\nabla_{\bx} \log f(\bbx,\bbv)
\nabla_{\bx} \log f(\bbx,\bbv)^\top\right]
\label{conditioning_decreases_Fisher_step2}
\end{align}
and
\begin{align}
E\left[\nabla_{\bx} \log f(\bbu|\bbx,\bbv) \nabla_{\bx} \log
f(\bbu|\bbx,\bbv)^\top\right]\succeq \bzero
\label{conditioning_decreases_Fisher_step3}
\end{align}
Using (\ref{conditioning_decreases_Fisher_step2}) and
(\ref{conditioning_decreases_Fisher_step3}) in
(\ref{conditioning_decreases_Fisher_step1}), we get
\begin{align}
\bbj(\bbx|\bbu)&\succeq \bbj(\bbx|\bbv)+E\left[\nabla_{\bx} \log f(\bbx,\bbv)  \nabla_{\bx} \log f(\bbu|\bbx,\bbv)^\top\right]\nonumber\\
&\quad+E\left[\nabla_{\bx} \log f(\bbu|\bbx,\bbv) \nabla_{\bx}
\log f(\bbx,\bbv)^\top\right]
\label{conditioning_decreases_Fisher_step4}
\end{align}
We now show that the cross-terms in
(\ref{conditioning_decreases_Fisher_step4}) vanish. To this end,
consider the $(i,j)$th entry of the first cross-term
\begin{align}
E\Big[\nabla_{\bx} \log f(\bbx,\bbv)  \nabla_{\bx} \log   f(\bbu|&\bbx,\bbv)^\top\Big]_{ij}=E\left[ \frac{\partial \log f(\bbx,\bbv)}{\partial x_i} \frac{\partial \log f(\bbu|\bbx,\bbv)}{\partial x_j}\right]\\
& =\int  \frac{\frac{\partial f(\bx,\bv)}{\partial x_i}}{f(\bx,\bv)}  \frac{\frac{\partial f(\bu|\bx,\bv)}{\partial x_j}}{f(\bu|\bx,\bv)}~~f(\bx,\bu,\bv) ~d\bu~ d\bv~ d\bx\\
& =\int  \frac{\partial f(\bx,\bv)}{\partial x_i} \frac{\partial f(\bu|\bx,\bv)}{\partial x_j} ~d\bu~ d\bv~ d\bx\\
& =\int  \frac{\partial f(\bx,\bv)}{\partial x_i}\left[ \int
\frac{\partial f(\bu|\bx,\bv)}{\partial x_j} ~d\bu\right]~ d\bv~
d\bx \label{conditioning_decreases_Fisher_step_XX}
\end{align}
where the inner integral can be evaluated as
\begin{align}
 \int \frac{\partial f(\bu|\bx,\bv)}{\partial x_j} ~d\bu =\frac{\partial}{\partial x_j} \left[ \int f(\bu|\bx,\bv) ~d\bu\right]=\frac{\partial (1)}{\partial x_j}=0
\label{derivative_zero_2}
\end{align}
where the interchange of the differentiation and the integration
is justified by the assumption given
in~(\ref{assumption_of_the_lemma_cond_inc_Fisher}). Thus, using
(\ref{derivative_zero_2}) in
(\ref{conditioning_decreases_Fisher_step_XX}) implies that
\begin{align}
E\Big[\nabla_{\bx} \log f(\bbx,\bbv)  \nabla_{\bx} \log
f(\bbu|&\bbx,\bbv)^\top\Big]=\bzero
\label{conditioning_decreases_Fisher_step_YY}
\end{align}
Thus, using (\ref{conditioning_decreases_Fisher_step_YY}) in
(\ref{conditioning_decreases_Fisher_step4}), we get
\begin{align}
\bbj(\bbx|\bbu)&\succeq \bbj(\bbx|\bbv)
\end{align}
which completes the proof.

\section{Proof of Lemma~\ref{lemma_construct_degraded_channel}}

\label{proof_of_lemma_construct_degraded_channel}

Since we assumed $\mu_j>0,~j=1,\ldots,m,$ we can select
\begin{align}
\tilde{\bbsigma}_{j+1}=\left[\left(\sum_{i=1}^j \bbk_i
+\bbsigma_{j+1}\right)^{-1}+\frac{1}{\mu_{j+1}}\bbm_{j+1}\right]^{-1}-\sum_{i=1}^{j}\bbk_i,\quad
j=0,1\ldots,m-1
\end{align}
which is equivalent to
\begin{align}
\mu_{j+1}\left(\sum_{i=1}^j \bbk_i
+\tilde{\bbsigma}_{j+1}\right)^{-1}=\mu_{j+1}\left(\sum_{i=1}^j
\bbk_i +\bbsigma_{j+1}\right)^{-1}+\bbm_{j+1},\quad
j=0,1\ldots,m-1 \label{definition_enhanced_noise_0}
\end{align}
and that implies $\bzero \preceq \tilde{\bbsigma}_j \preceq
\bbsigma_j,~j=1,\ldots,m$. Furthermore, for $j=0,\ldots,m-1$, we
have
\begin{align}
\lefteqn{\sum_{i=1}^{j+1}\bbk_i+\tilde{\bbsigma}_{j+1}=\bbk_{j+1}+\left(\sum_{i=1}^{j}\bbk_i+\tilde{\bbsigma}_{j+1}\right)}\\
&=\bbk_{j+1}+\left[\left(\sum_{i=1}^{j}\bbk_i+\bbsigma_{j+1}\right)^{-1}+\frac{1}{\mu_{j+1}}\bbm_{j+1}\right]^{-1}
\label{definition_enhanced_noise}\\
&= \bbk_{j+1}+\left[\bbi+\frac{1}{\mu_{j+1}}\left(\sum_{i=1}^{j}\bbk_i+\bbsigma_{j+1}\right) \bbm_{j+1}\right]^{-1}\left(\sum_{i=1}^{j}\bbk_i+\bbsigma_{j+1}\right) \\
&=
\bbk_{j+1}+\left[\bbi+\frac{1}{\mu_{j+1}}\left(\sum_{i=1}^{j+1}\bbk_i+\bbsigma_{j+1}\right)
\bbm_{j+1}\right]^{-1}\left(\sum_{i=1}^{j}\bbk_i+\bbsigma_{j+1}\right)
\label{kkt_imply_1}\\
&=\bbk_{j+1}+\left[\left(\sum_{i=1}^{j+1}\bbk_i+\bbsigma_{j+1}\right)^{-1}+
\frac{1}{\mu_{j+1}}\bbm_{j+1}\right]^{-1}\left(\sum_{i=1}^{j+1}\bbk_i+\bbsigma_{j+1}\right)^{-1}
\left(\sum_{i=1}^{j}\bbk_i+\bbsigma_{j+1}\right)\\
&=\bbk_{j+1}+\left[\left(\sum_{i=1}^{j+1}\bbk_i+\bbsigma_{j+1}\right)^{-1}+ \frac{1}{\mu_{j+1}}\bbm_{j+1}\right]^{-1}\left(\sum_{i=1}^{j+1}\bbk_i+\bbsigma_{j+1}\right)^{-1} \nonumber\\
&\qquad \times
\left(\sum_{i=1}^{j+1}\bbk_i+\bbsigma_{j+1}-\bbk_{j+1}\right)\\
&=\bbk_{j+1}+\left[\left(\sum_{i=1}^{j+1}\bbk_i+\bbsigma_{j+1}\right)^{-1}+ \frac{1}{\mu_{j+1}}\bbm_{j+1}\right]^{-1} \nonumber\\
&\qquad
-\left[\left(\sum_{i=1}^{j+1}\bbk_i+\bbsigma_{j+1}\right)^{-1}+
\frac{1}{\mu_{j+1}}\bbm_{j+1}\right]^{-1}
\left(\sum_{i=1}^{j+1}\bbk_i+\bbsigma_{j+1}\right)^{-1}
\bbk_{j+1} \\
&=\bbk_{j+1}+\left[\left(\sum_{i=1}^{j+1}\bbk_i+\bbsigma_{j+1}\right)^{-1}+ \frac{1}{\mu_{j+1}}\bbm_{j+1}\right]^{-1} \nonumber\\
&\qquad
-\left[\left(\sum_{i=1}^{j+1}\bbk_i+\bbsigma_{j+1}\right)^{-1}+
\frac{1}{\mu_{j+1}}\bbm_{j+1}\right]^{-1}
\left[\left(\sum_{i=1}^{j+1}\bbk_i+\bbsigma_{j+1}\right)^{-1}+
\frac{1}{\mu_{j+1}}\bbm_{j+1}\right]
\bbk_{j+1} \label{kkt_imply_2}\\
&=\bbk_{j+1}+\left[\left(\sum_{i=1}^{j+1}\bbk_i+\bbsigma_{j+1}\right)^{-1}+
\frac{1}{\mu_{j+1}}\bbm_{j+1}\right]^{-1}-
\bbk_{j+1}\\
&=\left[\left(\sum_{i=1}^{j+1}\bbk_i+\bbsigma_{j+1}\right)^{-1}+
\frac{1}{\mu_{j+1}}\bbm_{j+1}\right]^{-1}
\label{converse_proof_aligned_step_17}
\end{align}
where (\ref{definition_enhanced_noise}) follows from
(\ref{definition_enhanced_noise_0}), (\ref{kkt_imply_1}) and
(\ref{kkt_imply_2}) are consequences of the KKT conditions
$\bbm_j\bbk_j=\bbk_j\bbm_j=\bzero,~j=1,\ldots,m$. Finally,
(\ref{converse_proof_aligned_step_17}) is equivalent to
\begin{align}
\mu_{j+1}\left(\sum_{i=1}^{j+1}\bbk_i+\tilde{\bbsigma}_{j+1}\right)^{-1}
&=\mu_{j+1}\left(\sum_{i=1}^{j+1}\bbk_i+\bbsigma_{j+1}\right)^{-1}+
\bbm_{j+1},\quad j=0,\ldots,m-1
\label{converse_proof_aligned_step_18}
\end{align}
Plugging (\ref{definition_enhanced_noise_0}) and
(\ref{converse_proof_aligned_step_18}) into the KKT conditions in
(\ref{converse_proof_aligned_step_13}) and
(\ref{converse_proof_aligned_step_14}) yields the third part of
the lemma.

We now prove the second part of the lemma. To this end, consider
the second equation of the third part of the lemma, i.e., the
following
\begin{align}
\mu_m\left(\sum_{i=1}^m\bbk_i+\tilde{\bbsigma}_m\right)^{-1}=
\mu_{m}\left(\sum_{i=1}^m\bbk_i+\bbsigma_Z\right)^{-1}+\bbm_Z
\end{align}
which implies $\tilde{\bbsigma}_m\preceq \bbsigma_Z$. Now,
consider the first equation of the third part of the lemma for
$j=m-1$, i.e., the following
\begin{align}
\mu_{m-1}\left(\sum_{i=1}^{m-1}\bbk_i+\tilde{\bbsigma}_{m-1}\right)^{-1}
-\mu_{m-1}\left(\sum_{i=1}^{m-1}\bbk_i+\bbsigma_Z\right)^{-1}&=
\mu_{m}\left(\sum_{i=1}^{m-1}\bbk_i+\tilde{\bbsigma}_{m}\right)^{-1}\nonumber\\
&\quad -\mu_{m}\left(\sum_{i=1}^{m-1}\bbk_i+\bbsigma_Z\right)^{-1}
\label{conv_proof_ordering}
\end{align}
Since the matrix on the right hand side of the equation is
positive semi-definite due to the fact that
$\tilde{\bbsigma}_m\preceq \bbsigma_Z$, and we assume that
$\mu_m\geq \mu_{m-1}$, (\ref{conv_proof_ordering}) implies
\begin{align}
\left(\sum_{i=1}^{m-1}\bbk_i+\tilde{\bbsigma}_{m-1}\right)^{-1}
-\left(\sum_{i=1}^{m-1}\bbk_i+\bbsigma_Z\right)^{-1}&\succeq
\left(\sum_{i=1}^{m-1}\bbk_i+\tilde{\bbsigma}_{m}\right)^{-1}
-\left(\sum_{i=1}^{m-1}\bbk_i+\bbsigma_Z\right)^{-1}
\end{align}
which in turn implies $\tilde{\bbsigma}_{m-1}\preceq
\tilde{\bbsigma}_m \preceq  \bbsigma_Z$. Similarly, if one keeps
checking the first equation of the third part of the lemma in the
reverse order, one can get
\begin{align}
\tilde{\bbsigma}_1 \preceq \ldots \preceq\tilde{\bbsigma}_m\preceq
\bbsigma_Z
\end{align}
Moreover, the definition of $\tilde{\bbsigma}_1$, i.e.,
(\ref{definition_enhanced_noise_0}) for $j=0$,
\begin{align}
\tilde{\bbsigma}_1 =
\left[\bbsigma_1^{-1}+\frac{1}{\mu_1}\bbm_1\right]^{-1}
\end{align}
implies that $\tilde{\bbsigma}_1\succ \bzero $ completing the
proof of the second part of the lemma.

We now show the fourth part of the lemma
\begin{align}
\lefteqn{\hspace{-1cm}\left(\sum_{i=1}^{j+1}\bbk_i+\tilde{\bbsigma}_{j+1}\right)^{-1} \left(\sum_{i=1}^j \bbk_i +\tilde{\bbsigma}_{j+1}\right)}\nonumber\\
&= \left(\sum_{i=1}^{j+1}\bbk_i+\tilde{\bbsigma}_{j+1}\right)^{-1} \left(\sum_{i=1}^{j+1} \bbk_i +\tilde{\bbsigma}_{j+1}-\bbk_{j+1}\right)\\
&=\bbi-\left(\sum_{i=1}^{j+1}\bbk_i+\tilde{\bbsigma}_{j+1}\right)^{-1} \bbk_{j+1}\\
&=\bbi-\left[\left(\sum_{i=1}^{j+1}\bbk_i+\bbsigma_{j+1}\right)^{-1}+\frac{1}{\mu_{j+1}}\bbm_{j+1}\right] \bbk_{j+1}\label{previous_parts_lemma_imply}\\
&=\bbi-\left(\sum_{i=1}^{j+1}\bbk_i+\bbsigma_{j+1}\right)^{-1} \bbk_{j+1} \label{kkt_imply_3} \\
&=\left(\sum_{i=1}^{j+1}\bbk_i+\bbsigma_{j+1}\right)^{-1}\left(\sum_{i=1}^{j+1}\bbk_i+\bbsigma_{j+1}\right)-\left(\sum_{i=1}^{j+1}\bbk_i+\bbsigma_{j+1}\right)^{-1} \bbk_{j+1}\\
&=\left(\sum_{i=1}^{j+1}\bbk_i+\bbsigma_{j+1}\right)^{-1}\left(\sum_{i=1}^{j}\bbk_i+\bbsigma_{j+1}\right),\quad
j=0,\ldots,m-1
\end{align}
where (\ref{previous_parts_lemma_imply}) follows from
(\ref{converse_proof_aligned_step_18}) and (\ref{kkt_imply_3}) is
a consequence of the KKT conditions
$\bbk_j\bbm_j=\bbm_j\bbk_j=\bzero,~j=1,\ldots,m$.

The proof of the fifth part of the lemma follows similarly
\begin{align}
\left(\bbs+\tilde{\bbsigma}_m\right)\left(\sum_{i=1}^m
\bbk_i+\tilde{\bbsigma}_m\right)^{-1}&=
\left(\bbs-\sum_{i=1}^m \bbk_i +\sum_{i=1}^m\bbk_i+\tilde{\bbsigma}_m\right)\left(\sum_{i=1}^m \bbk_i+\tilde{\bbsigma}_m\right)^{-1} \nonumber\\
&=\left(\bbs-\sum_{i=1}^m \bbk_i \right)\left(\sum_{i=1}^m \bbk_i+\tilde{\bbsigma}_m\right)^{-1}+\bbi \\
&=\left(\bbs-\sum_{i=1}^m \bbk_i \right)\left[\left(\sum_{i=1}^m \bbk_i+\bbsigma_Z\right)^{-1}+\frac{1}{\mu_m}\bbm_Z\right]+\bbi\label{previous_parts_lemma_imply_1}\\
&=\left(\bbs-\sum_{i=1}^m \bbk_i \right)\left(\sum_{i=1}^m
\bbk_i+\bbsigma_Z\right)^{-1}+\bbi
\label{kkt_imply_4} \\
&=\left(\bbs-\sum_{i=1}^m \bbk_i \right)\left(\sum_{i=1}^m \bbk_i+\bbsigma_Z\right)^{-1}+\nonumber\\
&\quad \left(\sum_{i=1}^m \bbk_i+\bbsigma_Z\right) \left(\sum_{i=1}^m \bbk_i+\bbsigma_Z\right)^{-1} \\
&=\left(\bbs+\bbsigma_Z \right)\left(\sum_{i=1}^m
\bbk_i+\bbsigma_Z\right)^{-1}
\end{align}
where (\ref{previous_parts_lemma_imply_1}) follows from the second
equation of the third part of the lemma, and (\ref{kkt_imply_4})
is a consequence of the KKT condition in
(\ref{converse_proof_aligned_step_11}), completing the proof.

\bibliographystyle{unsrt}
\bibliography{IEEEabrv,references2}
\end{document}